\documentclass[onecolumn,10pt]{IEEEtran}

\usepackage{cite}
\usepackage{amsmath,amssymb,amsfonts}
\usepackage[ruled,vlined]{algorithm2e}
\usepackage{graphicx}
\usepackage{subcaption}
\usepackage{cleveref}
\usepackage{color, xcolor}
\usepackage{multirow}
\usepackage{pstricks}

\usepackage{amsthm}
\usepackage{bm}
\usepackage{booktabs}
\usepackage{multirow}
\usepackage{algpseudocode}
\usepackage{caption}
\usepackage{manyfoot}
\newtheorem{theorem}{Theorem}
\newtheorem{lemma}{Lemma}

\newtheorem{corollary}{Corollary}

\newtheorem{remark}{Remark}

\newcommand{\tabcaption}{\def\@captype{table}\caption}

\title{Rack-Aware MSR Codes with Linear Field Size and Smaller Sub-Packetization for Tolerating Multiple Erasures}
\date{}
\author{\IEEEauthorblockN{Hengming Zhao, Dianhua Wu,  Minquan Cheng}
\thanks{H. Zhao  is with the Key Lab of Education Blockchain and Intelligent Technology, Ministry of Education, and also with the Guangxi Key Lab of Multi-source Information Mining $\&$ Security, Guangxi Normal University
	Guilin, 541004, China, and also with School of Mathematics and Statistics, Nanning Normal University, Nanning 530100, China (e-mail: hengmingzh@163.com).}
\thanks{D. Wu is with Guangxi Key Lab of Multi-source Information Mining $\&$ Security, Guangxi Normal University, Guilin 541004, China; D. Wu is also with The Center for Applied Mathematics of Guangxi (Guangxi Normal University), Guilin 541006, China (e-mail: dhwu@gxnu.edu.cn). }
\thanks{M. Cheng is with the Key Lab of Education Blockchain and Intelligent Technology, Ministry of Education, and also with the Guangxi Key Lab of Multi-source Information Mining $\&$ Security, Guangxi Normal University, 541004 Guilin, China  (e-mail: chengqinshi@hotmail.com).}
}
\begin{document}
		\maketitle
	% 使用定义环境
\begin{abstract}In an $(n,k,d)$ rack-aware storage model, the system consists of $n$ nodes uniformly distributed across $\bar{n}$ successive racks, such that each rack contains $u$ nodes of equal capacity and the reconstructive degree satisfies $k=\bar{k}u+v$ where $1\leq\bar{k}\leq\lfloor k/u\rfloor$ and $0\leq v\leq u-1$. Suppose there are $h\geq1$ failed nodes in a rack (called the host rack). Then together with its surviving nodes, the host rack downloads recovery data from $\bar{d}$ helper racks and repairs its failed nodes. In this paper, we focus on studying the rack-aware minimum storage regenerating (MSR) codes for repairing $h$ failed nodes within the same rack. By using the coupled-layer construction with the alignment technique, we construct the first class of rack-aware MSR codes for all $\bar{k}+1\leq\bar{d}\leq\bar{n}-1$ which achieve the small sub-packetization $l=\bar{s}^{\lceil\bar{n}/\bar{s}\rceil}$ where the field size $q$ increases linearly with $n$ and $\bar{s}=\bar{d}-\bar{k}+1$. In addition, these codes achieve the optimal repair bandwidth for $1\leq h\leq u-v$, and the asymptotically  optimal repair bandwidth for $u-v+1\leq h\leq u$. In particular, they achieve the optimal access when $h=u-v$. It is worth noting that the existing rack-aware MSR codes which achieve the same sub-packetization $l=\bar{s}^{\lceil\bar{n}/\bar{s}\rceil}$ are only known for the special case of $\bar{d}=\bar{n}-1$, $h=1$, and the field size is much larger than ours. Then, based on our first construction we further develop another class of explicit rack-aware MSR codes with even smaller sub-packetization $l=\bar{s}^{\lceil\bar{n}/(\bar{s}+1)\rceil}$ for all admissible values of $\bar{d}$. 
\end{abstract}
\begin{IEEEkeywords}
Distributed storage, MDS array codes, rack-aware MSR codes, sub-packetization
\end{IEEEkeywords} 	 
\section{Introduction}

Distributed storage systems (DSSs) are commonly implemented by major companies such as Facebook and Google, using multiple independent but unreliable devices (termed nodes) for data storage. In an $(n, k, d)$ distributed storage system, an original file is encoded across $n$ nodes such that it can be recovered from any $k$ out of $n$ nodes. This property, called the $(n, k)$ maximum distance separable (MDS) property, ensures optimal fault tolerance, where $k$ is known as the reconstructive degree.
In the event of a single node failure, exact repair can be achieved by utilizing any $d$ helper nodes selected from the surviving $n-1$ nodes. The storage capacity denotes the per-node maximum data volume, while repair bandwidth quantifies the total data transferred by helper nodes during failure recovery.  These metrics characterize storage redundancy and repair efficiency respectively. The  work of Dimakis et al. \cite{DGWWR} derived the optimal storage-repair bandwidth trade-off via cut-set bounds, introducing minimum storage regenerating (MSR) codes  and minimum bandwidth regenerating (MBR) codes as the two extreme points on this trade-off curve. For details of the MBR codes, please see \cite{Rashmi2011, Ernvall2014, Zhang2020} and the references therein.

The MSR codes with minimum (optimal) repair bandwidth are widely studied such as \cite{Rashmi2011,LLT,LTT,LWHY,RSE,TYLH,TWB,VBV,YB1,YB2,HLSH,C-Barg,HLH,ZZ,WC}. However, in practical large-scale storage systems, multiple-node failures are more common than single-node failures. Moreover, the lazy repair strategy, which delays recovery operations until accumulating $h$ ($h\leq n-k$) node failures, is implemented in systems \cite{Cadambe2013}. This strategy reduces repair costs by avoiding immediate reconstruction after each single failure. 
%The repair is then performed by downloading required data from any $d$ helper nodes to recover the failed ones. 
In multi-node failure scenarios, two primary repair models have emerged in the literature, i.e., the centralized repair \cite{Cadambe2013} and the cooperative repair \cite{HXWZL}. In the centralized model, the data center repairs all failed nodes. In the cooperative repair model, the recovery of $h$ failed nodes involves $h$ new nodes executing a two-phase process, i.e, downloading data from any $d$ helper nodes and mutual data exchanging among the $h$ new nodes to complete recovery. There are many studies on the cooperative repair model such as \cite{SH,Li2014,LCT,YB,Ye,ZZW, Zhang2025, LWCT} and references therein.

In homogeneous distributed storage systems, every node and all communication between them are considered equal. However, modern data centers often adopt hierarchical topologies by organizing nodes into racks. Specifically, in an $(n,k,d)$ rack-aware storage model, the system consists of $n$ nodes uniformly distributed across $\bar{n}$ consecutive racks, with each rack containing $u$ nodes of identical capacity and the reconstructive degree  $k=\bar{k}u+v$ where $1\leq\bar{k}\leq\lfloor k/u\rfloor$ and $0\leq v\leq u-1$. We refer to a rack with failed nodes as a host rack. Assume that the host rack has $h$ failed nodes.  Then, the host rack utilizes the data downloaded from $\bar{d}$ helper racks and its internal $u-h$ helper nodes (referred to as local nodes) to repair its $h$ failed nodes. Clearly, the number of helper nodes satisfies $d=\bar{d}u+u-h$.
As rack sizes continue to grow, conventional networking equipment becomes inadequate, requiring specialized switches and routers to maintain sufficient bandwidth \cite{VFF}. For a rack-aware distributed storage model, the repair bandwidth is defined as the total amount of inter-rack data transmission required for failures recovery, while intra-rack communication is typically assumed to have zero bandwidth overhead. The MDS array codes that achieve the optimal repair bandwidth in the rack-aware storage model are referred to as rack-aware MSR codes.

Recently, Hou et al. in \cite{HLSH} presented the optimal trade-off between storage  and repair bandwidth, and proposed constructions for rack-aware MSR codes. However, their constructions require certain constraints on the parameters and rely on sufficiently large finite fields. Chen and Barg \cite{C-Barg} first proposed explicit constructions of rack-aware MSR codes with sub-packetization $l=\bar{s}^{\bar{n}}$ for single-node failure recovery where $\bar{s}=\bar{d}-\bar{k}+1$. Additionally, they proposed a construction achieving the optimal access when  $v=u-1$. Hou et al. \cite{HLH} introduced a coding framework capable of transforming MSR codes into rack-aware MSR codes, the resulting constructions are only guaranteed to exist over sufficiently large finite fields.
Zhou and Zhang \cite{ZZ} provided a construction of rack-aware MSR codes for all admissible parameters with small sub-packetization $l=\bar{s}^{\lceil\bar{n}/(u-v)\rceil}$ and reduced field size. This construction is actually an optimal access construction when  $v=u-1$.  Chen \cite{Chen2022} constructed a family of optimal access rack-aware MSR codes with sub-packetization $l=\bar{s}^{\bar{n}}(u-v)^{\bar{n}}$ for any $0\leq v\leq u$. Later, Wang and Chen \cite{WC} further reduce the sub-packetization to  $l=\bar{s}^{\bar{n}}$. 

The aforementioned rack-aware MSR codes are limited to single-node failure recovery. However, in rack-aware distributed storage systems, multiple node failures occur more frequently than single node failures. For repairing multiple node failures, Zhang and Zhou \cite{Zhou2022} considered  a relaxed repair model for failed nodes in the same rack. They enhanced erasure tolerance at the expense of storage efficiency and their proposed minimum storage codes sacrifice the MDS property. 
Wang et al. \cite{WZLT} constructed a class of rack-aware MSR codes with sub-packetization $l=\bar{s}^{\bar{n}}$ for multiple failures within the same rack, achieving the optimal access when $h=u-v$. Subsequently, in \cite{Wang2025}, they presented another class of rack-aware MSR codes that also feature optimal access and the same sub-packetization $l=\bar{s}^{\bar{n}}$, but designed for the scenario of multiple sequential failures within a rack. As far as we know, there are few studies on optimal access rack-aware MSR codes for repairing multiple failed nodes. The main existing rack-aware MSR code are listed in Table \ref{tab-Schemes}. 
\begin{table}[http!]
	\renewcommand{\arraystretch}{2}
	\setlength\tabcolsep{3.5pt} 
	\centering
	\caption{The existing codes and new codes in this paper where $\bar{s}=\bar{d}-\bar{k}+1$, $k=\bar{k}u+v$, $h$ and $\textit{O}_{\bar{s}}(n)$ represent the largest erasure tolerance and a linear function with $n$ respectively. 
		%		The notation $\textit{O}_{\bar{s}}$ indicates that $\bar{s}$ is treated as a constant, since $\bar{s}\leq \bar{n}-\bar{k}$. 	
		\label{tab-Schemes}
	} 
	\begin{tabular}{|c|c|c|c|c|c|c|c|c|}	\hline
		&  sub-packetization  $l$ & field size $q$ &optimal access &  $\bar{d}$ & $h$ \\ 	\hline
		\cite{C-Barg} &	 $\bar{s}^{\bar{n}}$ & $u|(q-1)$,\ $q\geq n+\bar{s}-1$ &  optimal if $v=u-1$& $\bar{k}+1\leq\bar{d}\leq \bar{n}-1$  & 1\\ \hline
		\cite{HLH} &	 $\bar{s}^{\lceil {\bar{n}}/{\bar{s}} \rceil}$ & $q> kl\sum_{i=1}^{min\{\bar{k},\bar{n}\}}{n-\bar{n}\choose{k-i}}{\bar{n}\choose i}$ &  no & $\bar{d}=\bar{n}-1$  & 1\\
		\hline
		\cite{ZZ} &	 $\bar{s}^{\lceil {\bar{n}}/{(u-v)} \rceil}$ & $u|(q-1)$, $q>n$ &  optimal if $v=u-1$ & $\bar{k}+1\leq\bar{d}\leq \bar{n}-1$  & 1\\
		\hline
		\cite{WZLT} &	 $\bar{s}^{\bar{n}}$ & $u|(q-1)$, $q>n$ &  optimal  & $\bar{k}+1\leq\bar{d}\leq \bar{n}-1$  & $u-v$\\
		\hline
		\cite{WZLT} &	 $\bar{s}^{\bar{n}}$ & $u|(q-1)$, $q>n$ &  no  & $\bar{k}+1\leq\bar{d}\leq \bar{n}-1$  & $1\leq h<u-v$\\
		\hline
		Theorem \ref{th1} &	$\bar{s}^{\lceil{\bar{n}}/{\bar{s}}\rceil}$  (best-known)  & $u|(q-1)$, $q\geq O_{\bar{s}}(n)$&  optimal  & $\bar{k}+1\leq\bar{d}\leq \bar{n}-1$  & $u-v$\\
		\hline
		Theorem \ref{th1} &	 $\bar{s}^{\lceil{\bar{n}}/{\bar{s}}\rceil}$ & $u|(q-1)$, $q\geq O_{\bar{s}}(n)$ &  no  & $\bar{k}+1\leq\bar{d}\leq \bar{n}-1$  & $1\leq h<u-v$\\
		\hline
		Theorem \ref{th2} &	 $\bar{s}^{\lceil{\bar{n}}/{(\bar{s}+1)}\rceil}$ (best-known) & $u|(q-1)$, $q\geq O_{\bar{s}}(n)$ &  no  & $\bar{k}+1\leq\bar{d}\leq \bar{n}-1$  & $1\leq h\leq u-v$ \\
		\hline
	\end{tabular}
\end{table}
\subsection{Contributions}
In this paper, we focus on the rack-aware MSR codes for repairing $h\geq 1$ failed nodes within the same rack. The work in \cite{WZLT} constructed a family of such codes with a sub-packetization level of $\bar{s}^{\bar{n}}$ via permutation matrices. In their approach, the number of racks is identity-mapped to the exponent of the sub-packetization level. However, since every row of each parity-check submatrix contains only one non-zero element, this method makes it difficult to further reduce the sub-packetization level. To significantly reduce the sub-packetization level, we partition the racks into equal-sized groups and then, using an alignment technique to extend coupled-layer construction, identity-map the number of these rack groups to the exponent of the sub-packetization level to construct the parity-check matrix. As a result, the proof of the MDS property for our codes no longer reduces to showing the invertibility of a  Vandermonde matrix; indeed, it becomes more complex. Nevertheless, we can prove this property using the Combinatorial Nullstellensatz. 
For the repair proccess, as the host rack downloads coded symbols from each helper rack, utilizing this data requires transforming the rack-aware MSR code $\mathcal{C}$ into an intimate code $\bar{\mathcal{C}}(w)$ for $w \in [h]$, imposing the requirement that  $\bar{\mathcal{C}}(w)$ is MDS. Consequently, the parity-check matrix of $\mathcal{C}$ must be carefully designed to ensure that both $\mathcal{C}$ itself and every $\bar{\mathcal{C}}(w)$ satisfy the MDS property. Although the expression forms of our rack-aware MSR codes are not very different from the codes in \cite{LWHY}, our approach to repair is fundamentally different. We construct two classes of rack-aware MSR codes with smaller sub-packetization level and linear field size which are listed in Table \ref{tab-Schemes}.  Specifically the main results can be summarized as follows. 
\begin{itemize}
\item We construct our first family of rack-aware MSR codes, i.e., Theorem \ref{th1}, which have the optimal repair bandwidth when $1\leq h\leq u-v$ and asymptotically optimal repair bandwidth when $u-v+1\leq h\leq u$. In addition,  the field size is linear in $n$, the sub-packetization is exponential in $\lceil\bar{n}/\bar{s}\rceil$, and our array codes have the optimal access when  $h=u-v$. 
	
\item Inspired by our first construction, we obtain our second family of rack-aware MSR codes, i.e., Theorem \ref{th2},  which also have the  optimal repair bandwidth for $1\leq h\leq u-v$ and asymptotically optimal repair bandwidth for $u-v+1\leq h\leq u$. In addition,  the field size is linear in $n$ and the sub-packetization is exponential in $\lceil\bar{n}/(\bar{s}+1)\rceil$. 
%for repairing $h\geq1$ failed nodes in the same rack where which is also smaller than the existing rack-aware MSR codes. In addition, our new codes have 
\end{itemize}

By Table \ref{tab-Schemes}, when $h=1$, the code in Theorem \ref{th1} requires the same sub-packetization level as that in \cite{HLH} but operates over a smaller field size; when $h>1$, the code in Theorem \ref{th1} achieves a significantly smaller sub-packetization level of $\bar{s}^{\lceil\bar{n}/\bar{s}\rceil}$ compared to the $ \bar{s}^{\bar{n}}$ required in \cite{WZLT} while maintaining optimal access when $h=u-v$; the code in Theorem \ref{th2} attains the smallest known sub-packetization $ \bar{s}^{\lceil \bar{n}/(\bar{s}+1) \rceil} $ among existing rack-aware MSR constructions.

%In order to construct rack-aware MSR codes and develop their repair scheme, we must carefully design the linear combinations of failed nodes in the host rack using parts of symbols from helper racks and local nodes.  
%
%
% and construct the rack-aware MSR codes with smaller sub-packetizations and linear field size that can achieve either optimal access or (asymptotically) optimal repair bandwidth.
%For the recovery of failed nodes in the host rack, their linear combinations are represented by the parts of symbols in helper racks, supposing that helper racks and  local nodes are available, then the failed nodes are recovered by solving the linear equations. 
%Unlike the multi-node failure repair process in homogeneous model, for the recovery of failed nodes in the host rack, their linear combinations are represented by the parts of symbols in helper racks, supposing that helper racks and  local nodes are available, then the failed nodes are recovered by solving the linear equations. Therefore, the construction method and repair scheme used in homogeneous models cannot be directly extended to rack-aware storage model.  
%The key challenge of this problem is to {\color{red}construct rack-aware MSR codes that achieve smaller sub-packetization levels and linear field sizes compared to existing codes, while developing their corresponding repair schemes. } 

\subsection{Organization and notations} 

The rest of this paper is arranged as follows. Section \ref{sect-preliminaries} introduces the system model for multiple failures and its related  fundamental terminology and concepts. Main results are presented in Section \ref{sect-main}. Proofs of main results are shown in  Sections \ref{Sect3} and \ref{Sect4}. Finally, Conclusion is included in Section \ref{sect-conclusion}.
  
In this paper, we use boldface capital letters, boldface lowercase letters, and calligraphic letters to denote matrices, vectors, and sets, respectively.
In addition, the following notations are used unless otherwise stated.
\begin{itemize}
	\item  $\mathbb{F}_q$ is a finite field with order $q$,  $\xi$ is a primitive element of $\mathbb{F}_q$,  $\theta$ is an element of $\mathbb{F}_q$ with multiplicative order $u$ such that $u|(q-1)$.
	\item ${\bf I}_m$ is the $m\times m$ identical matrix over $\mathbb{F}_q$.
	\item $|\cdot|$ denotes the cardinality of a set and $\top$ denotes the transpose operator.
	\item For positive integers $a$ and $m$, $[m]$ denotes the set $\{0,1,\ldots,m-1\}$, and $a+[m]=\{a+i:i\in[m]\}$.
	\item For any integers $a$ and $b$ with $a<b$, $[a,b)$ and $[a,b]$ denote the  sets $\{a,a+1,\ldots,b-1\}$  and
	$\left\{ a,a+1,\ldots,b\right\}$, respectively. 
	\item Given an $m$-vector ${\bf x}=(x_0,x_1,\ldots,x_{m-1})$, for any positive integer $a$, an element $\vartheta\in\mathbb{F}_q$ and any subset of $\mathcal{A}\subseteq [m]$, let $\vartheta{\bf x}_{\mathcal{A}}^a:=(\vartheta x_i^a)_{i\in\mathcal{A}}$. For any nonnegative integer set $\mathcal{B}\subseteq[u]$, let $\vartheta^{\mathcal{B}}:=\{\vartheta^{i}: i\in\mathcal{B}\}$ and $\vartheta^{\mathcal{B}} x_{\mathcal{A}}^a=\{\vartheta^{i}x_{\mathcal{A}}^a: i\in \mathcal{B}\}$.
	\item For any element $x\in \mathbb{F}_q$ and any positive integer $t$, we 	
	define a column vector of length $t$ as $L^{(t)}(x)=(1,x,\ldots,x^{t-1})^{\top}$. 
	\item Given an integer $i\in[s^n]$ where $n$ is a positive integer, with $i=\sum_{z=0}^{n-1}i_zs^z$ for integer $i_z\in[s]$, we refer to $i=(i_0,i_1,\ldots,i_{n-1})$ as the $s$-ary representation of $i$
	
\end{itemize}

\section{Preliminaries} 
\label{sect-preliminaries}

In this section, we first introduce the system model, then review fundamental concepts and present relevant results.

\subsection{The system model } 
\label{subsect-system} 
Consider an $(n, k, l)$ array code $\mathcal{C}$ over a finite field $\mathbb{F}_q$. To be specific, each codeword in $\mathcal{C}$ is of the form  
$\mathbf{C} = (\mathbf{c}_0, \mathbf{c}_1, \ldots, \mathbf{c}_{n-1})$,  
where $\mathbf{c}_i=(c_{i,0}, c_{i,1}, \ldots, c_{i,l-1})^\top \in \mathbb{F}_q^l$ for $i\in[n]$.  We assume $\mathcal{C}$ is a
linear subspace over $\mathbb{F}_q$. The array code $\mathcal{C}$ is MDS if each codeword in $\mathcal{C}$ can be reconstructed from any $k$ of its coordinates. In a distributed storage system, assume that there is an original data which is independently and identically
uniformly distributed in $\mathbb{F}_q^{kl}$. We divide the original data into $k$ blocks, encoded into a codeword $\mathbf{C}\in \mathcal{C}$, and stored in $n$ nodes. Given two positive integers $\bar{n}$ and $u$ satisfying that $n=\bar{n}u$, we divide $n$ nodes $[n]$ into consecutive $\bar{n}$ groups (referred to as racks) each of which has $u$ nodes. In this case, for each integer $i'\in[\bar{n}]$ the $i'$-th rack consists of  $u$ nodes $({\bf c}_{i'u},{\bf c}_{i'u+1},\ldots,{\bf c}_{i'u+u-1})$. Here, we make no distinction between the coordinates of the codeword and their physical representations as storage nodes, uniformly referring to both as nodes.
% Consequently, the coordinates of the codeword $C$ are divided into $\bar{n}$ segments of size $u$, and we denote them as ${\bf c}_i$, $i\in[n]$, where $i=i'u+g$, , $g\in[u]$. We make no distinction between the coordinates of the codeword and their physical representations as storage nodes, uniformly referring to both as nodes. 
% 
% 
% Clearly,  
Each rack is equipped with a relayer node that has access to the contents of the other nodes within the same rack.  We focus on studying the case of $u\in[1,k]$. Otherwise, a single node failure could be trivially repaired by the  surviving nodes within the same rack. Moreover, we also assume that $u\leq n-k$  to guarantee that the code has the ability to be repaired if full-rack fails\cite{WZLT}. Let $k=\bar{k}u+v$ for an integer $v\in [u]$, we have $r=n-k=\bar{n}u-(\bar{k}u+v)=(\bar{n}-\bar{k})u-v=\bar{r}u-v$ where $\bar{r}=\bar{n}-\bar{k}$. 

In practice, the host rack usually contain multiple failed nodes. Suppose that there are $h\geq 1$ failed nodes in the host rack. Let $i'$ be the index of the host rack and $\mathcal{I}=\{g_0,g_1,\ldots,g_{h-1}\}$ be the index set of the failed nodes, and
${\mathcal R}=\{j_0,j_1,\ldots,j_{\bar{d}-1}\}\subseteq[\bar{n}]\setminus\{i'\}$ be the index set of the $\bar{d}$ help racks. To recover the $h$ failed nodes,
the host rack downloads $\alpha$ coded symbols from each helper rack $j\in{\mathcal R}$. Then, the host rack uses the $\bar{d}\alpha$ symbols along with the symbols from its $u-h$ local nodes to repair the $h$ failed nodes. Clearly, the number of helper nodes $d=\bar{d}u+u-h$. Sine intra-rack communication is  assumed to have zero bandwidth overhead in the rack-aware storage model, the repair bandwidth of recovering the $h$ failed nodes is defined as $\beta(\bar{d},h)=\bar{d}\alpha$. Note that $h\leq \min\{u,r\}$ and $\bar{d}u+u-h\geq k$, we have $h\leq \min\{u,(\bar{d}-\bar{k}+1)u-v\}$\cite{WC}. Let $\bar{s}=\bar{d}-\bar{k}+1$. The authors in \cite{C-Barg} derived the following lower bound of the repair bandwidth for the $h$ failed nodes,
\begin{align}\label{eq1-1}
	\beta(\bar{d},h)\geq \dfrac{\bar{d}hl}{\bar{s}}.
\end{align}
Moreover, when $\bar{k}\geq1$ they showed that the equality in \eqref{eq1-1} holds if and only if the host rack downloads ${hl}/\bar{s}$ symbols from each helper rack.
An MDS array code achieving the lower bound in \eqref{eq1-1} is called a rack-aware MSR code, which provides optimal repair bandwidth. The scenario where $u=1$ corresponds to the homogeneous model. At this point $\bar{d}=d$, $\bar{k}=k$, and $h=1$, then inequality \eqref{eq1-1} can be written as  
\begin{align}\label{eq1-1-1}
	\beta(d,1)\geq \dfrac{dl}{s},
\end{align}
which was initially derived by Dimaki et al. in \cite{DGWWR} where $s=d-k+1$. 

Besides bandwidth, the repair efficiency and communication complexity are also affected by the number of symbols accessed on helper racks to generate the data downloaded by the host rack. Then, it is desirable to construct codes that provide optimal repair bandwidth and low access. A lower bound on the number of symbols accessed for single-node repair was derived in \cite{C-Barg};  this bound was later improved in \cite{Li2021} for the special case where $\bar{d}=\bar{n}-1$. Subsequently, the work in \cite{WC} extended these results to the multi-node repair scenario.
\begin{lemma}(\cite{WC})\label{oa}\rm
Let $\mathcal C$ be an $(n,k,l)$ rack-aware MSR codes over $\mathbb{F}_q$ for any $h\leq \min\{u,\bar{s}u-v\}$ within a single rack from any $\bar{d}\geq \bar{k}$ helper racks, each of which provides the same number of data for repair. For any $\bar{d}$ helper racks, the number of symbols accessed on the helper racks satisfies 
\begin{align} \label{eq1-2}
  \gamma\geq\dfrac{\bar{d}hul}{\bar{s}(u-v)}.
\end{align}
If $v>0$, the equality in \eqref{eq1-2} holds if and only if the number of symbols accessed on the node $i=ju+g$ satisfies $\alpha_i=hl/(\bar{s}(u-v))$ for all $j\in{\mathcal R}$ and $g\in[u]$. 
\end{lemma}
\begin{remark}[\cite{DGWWR}]\label{rm1}
If $u=1$, then $\bar{d}=d$, $\bar{k}=k$, and $h=1$, we have $\gamma\geq{dl}/s$,  where $s=d-k+1$. This lower bound is the same as that of \eqref{eq1-1-1}.
\end{remark}

An $(n,k,l)$ rack-aware MSR code achieving the lower bound in \eqref{eq1-2} is termed optimal access.

\subsection{Fundamental terminology and relevant results}
In this paper, when we mention that $\mathbf{B}$ is an $m\times n$ block matrix, we always suppose that $\mathbf{B}$ is partitioned uniformly, i.e., the size of every block entry of $\mathbf{B}$ is identical. In addition, we use the following matrix operator $\boxtimes$  and blow-up map. For a matrix $\mathbf{A}$ and an $m\times n$ block matrix  
\begin{align*} 
	\mathbf{B}={\left[\begin{array}{llll}
			\mathbf{B}_{0,0} & \ldots & \mathbf{B}_{0,n-1} \\			
			\ \ \vdots & \ddots &\ \ \vdots \\
			\mathbf{B}_{m-1,0} & \ldots & \mathbf{B}_{m-1,n-1} \\
		\end{array}
		\right]},
\end{align*}  
\noindent we define 
\begin{align} \label{eq1-3}
	\mathbf{A}\boxtimes \mathbf{B}={\left[\begin{array}{llll}
			\mathbf{A}\otimes \mathbf{B}_{0,0} & \ldots & \mathbf{A}\otimes \mathbf{B}_{0,n-1} \\			
			\ \ \ \vdots & \ddots &\ \ \ \vdots \\
			\mathbf{A}\otimes \mathbf{B}_{m-1,0} & \ldots & \mathbf{A}\otimes \mathbf{B}_{m-1,n-1} \\
		\end{array}
		\right]},
\end{align}  
\noindent where $\otimes$ is the Kronecker product. It should be noted that the result of $\mathbf{A}\boxtimes \mathbf{B}$  relies on the way in which the rows and columns of $\mathbf{B}$ are partitioned. Moreover, we will make clear the partition whenever we employ this notation. When each block entry of $\mathbf{B}$ is a scalar over $\mathbb{F}_q$, we have $\mathbf{A}\boxtimes \mathbf{B}=\mathbf{B}\otimes \mathbf{A}$.

Let  $\tilde{n}$ be a positive integer. For any integer $a\in[\tilde{n}]$, we blow up an $\bar{s}\times \bar{s}$ block matrix 
\begin{align*} 
	\mathbf{U}={\left[\begin{array}{llll}
			\mathbf{U}_{0,0} & \ldots & \mathbf{U}_{0,\bar{s}-1} \\			
			\ \ \vdots & \ddots &\ \ \vdots \\
			\mathbf{U}_{\bar{s}-1,0} & \ldots & \mathbf{U}_{\bar{s}-1,\bar{s}-1} \\
		\end{array}
		\right]} 
\end{align*}  
\noindent to get an $\bar{s}^{\tilde{n}}\times \bar{s}^{\tilde{n}}$ block matrix via
\begin{align} \label{eq1}
	\varPsi_{\tilde{n},a}(\mathbf{U})&={\bf I}_{\bar{s}^{\tilde{n}-a-1}}\otimes({\bf I}_{\bar{s}^a}\boxtimes \mathbf{U})\\
	\nonumber	&={\bf I}_{\bar{s}^{\tilde{n}-a-1}}\otimes{\left[\begin{array}{llll}
			{\bf I}_{\bar{s}^a}\otimes \mathbf{U}_{0,0} & \ldots & {\bf I}_{\bar{s}^a}\otimes \mathbf{U}_{0,\bar{s}-1} \\			
			\ \ \vdots & \ddots &\ \ \vdots \\
			{\bf I}_{\bar{s}^a}\otimes \mathbf{U}_{\bar{s}-1,0} & \ldots & {\bf I}_{\bar{s}^a}\otimes \mathbf{U}_{\bar{s}-1,\bar{s}-1} \\
		\end{array}
		\right]}.
\end{align}

Let $l=\bar{s}^{\tilde{n}}$. A matrix $\mathbf{U}$ and its blown-up block matrix $\varPsi_{\tilde{n},a}(\mathbf{U})$  have the following close relationship.
\begin{lemma}(\cite{Zhang2025}) \label{lm1}
	For any $i,j\in[l]$, the block entry of $\varPsi_{\tilde{n},a}(\mathbf{U})$ at the $i$-th block row and $j$-th block column is
	\[
	\varPsi_{\tilde{n},a}(\mathbf{U})(i,j)= 
	\begin{cases}
		\mathbf{U}(i_a,j_a), & {\rm if } \ i_z=j_z \ \forall z\in[\tilde{n}]\setminus\{a\},\\
		0, & \rm{otherwise, } 
	\end{cases}
	\]
	where $\mathbf{U}(i_a, j_a)$ is the block entry of the matrix $\mathbf{U}$ which lies in the $i_a$-th block row and the $j_a$-th block column.
\end{lemma}
From Lemma 4 in \cite{Zhang2025}, it is not difficult to obtain the following result.
\begin{lemma}\label{lm1-1}
	If $\mathbf{U}_0$ and $\mathbf{U}_1$ are two $\bar{s}\times\bar{s}$ block matrices, then $\varPsi_{\tilde{n},a}(\mathbf{U}_0)\varPsi_{\tilde{n},a}(\mathbf{U}_1)=\varPsi_{\tilde{n},a}(\mathbf{U}_0\mathbf{U}_1)$, where $\mathbf{U}_0\mathbf{U}_1$ is a matrix product.
\end{lemma}

Now, let us introduce kernel maps, which are useful for designing parity-check submatrices. This construction method was first introduced in \cite{YB2}. For any positive integers $b\in[\bar{s}+1]$ and $t$, define 
\begin{align*} 
	\varphi_b^{(t)}:\mathbb{F}_q^{\bar{s}}\to\mathbb{F}_q^{\bar{s}t\times \bar{s}}
\end{align*} 
which maps $\mathbf{x}_{[\bar{s}]}$ to the $\bar{s}t\times\bar{s}$ matrix $\varphi_b^{(t)}(\mathbf{x}_{[\bar{s}]})$ over $\mathbb{F}_q$. We regard $\varphi_b^{(t)}(\mathbf{x}_{[\bar{s}]})$ as a block matrix of size $\bar{s}\times\bar{s}$, where the block entries are column vectors of length $t$. For $i,j,b\in[\bar{s}]$, the block entry locating at the
$i$-th block row and the $j$-th block column is
\begin{align}\label{eq2}
	\varphi_b^{(t)}(\mathbf{x}_{[\bar{s}]})(i,j)= 
	\begin{cases}
		L^{(t)}(x_j) & \text{if } i=j, \\
		-L^{(t)}(x_j) & \text{if } i=b, i\neq j,\\
		0 & \text{otherwise}.
	\end{cases}
\end{align}
While for $b=\bar{s}$ and $i,j\in[\bar{s}]$, $\varphi_{\bar{s}}^{(t)}(\mathbf{x}_{[\bar{s}]})(i,j)$ can be defined as
\begin{align}\label{eq2-1-1}
	\varphi_{\bar{s}}^{(t)}(\mathbf{x}_{[\bar{s}]})(i,j)= 
	\begin{cases}
		L^{(t)}(x_j) & \text{if } i=j, \\	
		0 & \text{otherwise}.
	\end{cases}
\end{align}
Note that $\varphi_b^{(t)}(\mathbf{x}_{[\bar{s}]})$ is an $\bar{s}\times\bar{s}$ block matrix.  When no confusion will be caused, we abbreviate matrix $[\varphi_b^{(t)}(\mathbf{x}_{[\bar{s}]}) \ \varphi_b^{(t)}(\mathbf{x}_{\bar{s}+[\bar{s}]})$ $ \ldots \varphi_b^{(t)}(\mathbf{x}_{(t'-1)\bar{s}+[\bar{s}]})]$
as $\boldsymbol{\varphi}_b^{(t)}(\mathbf{x}_{[\bar{s}t']})$, where $t'$ is a positive integer.

For a non-empty subset ${\mathcal B}=\{b_0,b_1,\ldots,b_{t-1}\}\subseteq [\bar{s}+1]$ with $b_0<b_1<\ldots<b_{t-1}$, define
\begin{align}\label{eq2-1}
	{\mathcal G}_{\mathcal B}=\{{\mathcal G}_{b_0},{\mathcal G}_{b_1},\ldots,{\mathcal G}_{b_{t-1}}\},
\end{align}
where ${\mathcal G}_{b_j}=\{g_{b_j,0},g_{b_j,1},\ldots,g_{b_j,m_j-1}\}\subseteq [u]$ for $0\leq j\leq t-1$, $m_j\geq1$ and  $|{\mathcal G}_{\mathcal B}|=m_0+m_1+\ldots+m_{t-1}=\delta$. For an $\bar{s}t$-length vector $\mathbf{x}_{[\bar{s}t]}$, we define 
\begin{align*}
	\theta^{{\mathcal G}_{\mathcal B}}\odot \mathbf{x}_{[\bar{s}t]}=(\theta^{{\mathcal G}_{b_0}} \mathbf{x}_{[\bar{s}]},\theta^{{\mathcal G}_{b_1}} \mathbf{x}_{\bar{s}+[\bar{s}]},\ldots,\theta^{{\mathcal G}_{b_{t-1}}} \mathbf{x}_{(t-1)\bar{s}+[\bar{s}]}).
\end{align*}
Moreover, for any positive integer $m\geq \delta$, we define the $m\bar{s}\times \delta\bar{s}$ matrix
\begin{align}\label{eq4}
	\nonumber &	\boldsymbol{\varphi}_{\mathcal B}^{(m)}(\theta^{{\mathcal G}_{\mathcal B}}\odot \mathbf{x}_{[\bar{s}t]})\\
	\nonumber =&[	\varphi_{b_0}^{(m)}(\theta^{{\mathcal G}_{b_0}}\mathbf{x}_{[\bar{s}]}) \ \ \varphi_{b_1}^{(m)}(\theta^{{\mathcal G}_{b_1}}\mathbf{x}_{\bar{s}+[\bar{s}]})\ldots 	\varphi_{b_{t-1}}^{(m)}(\theta^{{\mathcal G}_{b_{t-1}}}\mathbf{x}_{(t-1)\bar{s}+[\bar{s}]})] \\
	=&[\varphi_{b_0}^{(m)}(\theta^{g_{b_0,0}}\mathbf{x}_{[\bar{s}]})\ldots \varphi_{b_0}^{(m)}(\theta^{g_{b_0,m_0-1}}\mathbf{x}_{[\bar{s}]}) \ldots \varphi_{b_{t-1}}^{(m)}(\theta^{g_{b_{t-1},0}}\mathbf{x}_{(t-1)\bar{s}+[\bar{s}]})\ldots \varphi_{b_{t-1}}^{(m)}(\theta^{g_{b_{t-1},m_{t-1}-1}}\mathbf{x}_{(t-1)\bar{s}+[\bar{s}]})],
\end{align}
and  the $ml\times \delta l$ matrix
\begin{align}\label{eq5}
	\nonumber	\boldsymbol{\varPsi}_{\tilde{n},a}(\boldsymbol{\varphi}_{\mathcal B}^{(m)}(\theta^{{\mathcal G}_{\mathcal B}}\odot \mathbf{x}_{[\bar{s}t]}))
	\nonumber=&[\varPsi_{\tilde{n},a}(\varphi_{b_0}^{(m)}(\theta^{g_{b_0,0}}\mathbf{x}_{[\bar{s}]}))\ldots \varPsi_{\tilde{n},a}(\varphi_{b_0}^{(m)}(\theta^{g_{b_0,m_0-1}}\mathbf{x}_{[\bar{s}]}))	\ldots
	\varPsi_{\tilde{n},a}(\varphi_{b_{t-1}}^{(m)}(\theta^{g_{b_{t-1},0}}\mathbf{x}_{(t-1)\bar{s}+[\bar{s}]}))\\ & \ldots \varPsi_{\tilde{n},a}(\varphi_{b_{t-1}}^{(m)}(\theta^{g_{b_{t-1},m_{t-1}-1}}\mathbf{x}_{(t-1)\bar{s}+[\bar{s}]}))].
\end{align}
If ${\mathcal G}_{b_j}=\{0\}$ for all $j\in[t]$, we have $\theta^{{\mathcal G}_{\mathcal B}}\odot \mathbf{x}_{[\bar{s}t]}=\mathbf{x}_{[\bar{s}t]}$ and $m_0+m_1+\ldots+m_{t-1}=t$, then
\begin{align}\label{eq501}
	\boldsymbol{\varphi}_{\mathcal B}^{(m)}(\mathbf{x}_{[\bar{s}t]})
	=[\varphi_{b_0}^{(m)}(\mathbf{x}_{[\bar{s}]})\ \ \varphi_{b_1}^{(m)}(\mathbf{x}_{\bar{s}+[\bar{s}]})\ldots \varphi_{b_{t-1}}^{(m)}(\mathbf{x}_{(t-1)\bar{s}+[\bar{s}]})],
\end{align}
and 
\begin{align}\label{eq502}
	\boldsymbol{\varPsi}_{\tilde{n},a}(\boldsymbol{\varphi}_{\mathcal B}^{(m)}( \mathbf{x}_{[\bar{s}t]}))
	=[\varPsi_{\tilde{n},a}(\varphi_{b_0}^{(m)}(\mathbf{x}_{[\bar{s}]})) \ \varPsi_{\tilde{n},a}(\varphi_{b_1}^{(m)}(\mathbf{x}_{\bar{s}+[\bar{s}]}))\ldots \varPsi_{\tilde{n},a}(\varphi_{b_{t-1}}^{(m)}(\mathbf{x}_{(t-1)\bar{s}+[\bar{s}]}))].
\end{align}
If $m=\delta$, then $\boldsymbol{\varphi}_{\mathcal B}^{(m)}(\theta^{{\mathcal G}_{\mathcal B}}\odot \mathbf{x}_{[\bar{s}t]})$ is a square matrix, and we write $\boldsymbol{\varphi}_{\mathcal B}^{(m)}(\theta^{{\mathcal G}_{\mathcal B}}\odot  \mathbf{x}_{[\bar{s}t]})$ as $\boldsymbol{\varphi}_{\mathcal B}(\theta^{{\mathcal G}_{\mathcal B}}\odot  \mathbf{x}_{[\bar{s}t]})$.

\begin{lemma}\label{lm2}
	Following the notations introduced above, the matrix $\boldsymbol{\varphi}_{\mathcal B}(\theta^{{\mathcal G}_{\mathcal B}}\odot \mathbf{x}_{[\bar{s}t]})$ is invertible if and only if $\boldsymbol{\varPsi}_{\tilde{n},a}(\boldsymbol{\varphi}_{\mathcal B}(\theta^{{\mathcal G}_{\mathcal B}}\odot \mathbf{x}_{[\bar{s}t]}))$ is invertible.
\end{lemma}
\begin{proof}
Our conclusion can be derived from Corollary 1 of \cite{LWHY} by replacing $s$ with $\bar{s}$, $\bar{n}$ with $\tilde{n}$, and $\mathbf{x}_{[\bar{s}t]}$ with $\theta^{\mathcal{G}_{\mathcal{B}}} \odot \mathbf{x}_{[\bar{s}t]}$.	
\end{proof}

\begin{lemma}\label{lm3} 	
	Following the notations introduced above, if $\theta^{{\mathcal G}_{\mathcal B}}\odot \mathbf{x}_{[\bar{s}t]}$ be $\bar{s}\delta$ distinct elements over $\mathbb{F}_q$ such that $\boldsymbol{\varPsi}_{\tilde{n},a}(\boldsymbol{\varphi}_{\mathcal B}^{(\delta)}(\theta^{{\mathcal G}_{\mathcal B}}\odot \mathbf{x}_{[\bar{s}t]}))$ is invertible, then for any integer $m>\delta$, there exists an $ml\times ml$ matrix $\mathbf{V}$ such that
	
	{\rm(1)} \begin{center}$\mathbf{V}\boldsymbol{\varPsi}_{\tilde{n},a}(\boldsymbol{\varphi}_{\mathcal B}^{(m)}(\theta^{{\mathcal G}_{\mathcal B}}\odot \mathbf{x}_{[\bar{s}t]}))={\left[\begin{array}{llll}
				\boldsymbol{\varPsi}_{\tilde{n},a}(\boldsymbol{\varphi}_{\mathcal B}^{(\delta)}(\theta^{{\mathcal G}_{\mathcal B}}\odot \mathbf{x}_{[\bar{s}t]}))  \\
				
				\ \ \ \ \ \ \ \ \ \ \	0  \\
			\end{array}
			\right]} $,
	\end{center}
	where $\boldsymbol{0}$ represents the $(m-\delta)l\times \delta l$ all zero matrix.
	
	{\rm(2)} For any $e\in [\tilde{n}]\setminus\{a\}$, $h\in[\bar{s}+1]$, and $\bar{s}$ elements $\mathbf{y}_{[\bar{s}]}$ which have no common elements with $\theta^{{\mathcal G}_{\mathcal B}}\odot \mathbf{x}_{[\bar{s}t]}$, 
	\begin{center}
		$\mathbf{V}\varPsi_{\tilde{n},e}(\varphi_h^{(m)}(\mathbf{y}_{[\bar{s}]}))={\left[\begin{array}{llll}
				\varPsi_{\tilde{n},e}(\varphi_h^{(\delta)}(\mathbf{y}_{[\bar{s}]}))  \\					
				\widehat{\varPsi}_{\tilde{n},e}(\varphi_h^{(m-\delta)}(\mathbf{y}_{[\bar{s}]}))  
			\end{array}
			\right]}$,
	\end{center}
	where $\widehat{\varPsi}_{\tilde{n},e}(\varphi_h^{(m-\delta)}(\mathbf{y}_{[\bar{s}]}))$ is an  $(m-\delta)l\times l$ matrix which is column equivalent to  $\varPsi_{\tilde{n},e}(\varphi_h^{(m-\delta)}(\mathbf{y}_{[\bar{s}]}))$.
	
	{\rm(3)} For any $x\notin \theta^{{\mathcal G}_{\mathcal B}}\odot \mathbf{x}_{[\bar{s}t]}$,
	\begin{center}
		$\mathbf{V}({\bf I}_l\otimes L^{(m)}(x))={\left[\begin{array}{llll}
				{\bf I}_l\otimes L^{(\delta)}(x) \\					
				({\bf I}_l\otimes L^{(m-\delta)}(x))\mathbf{Q}
			\end{array}
			\right]}$,
	\end{center}
	where $\mathbf{Q}$ is an $l\times l$ invertible matrix.
\end{lemma}
\begin{proof}
In Lemma 3 of \cite{LWHY}, by replacing $s$ with $\bar{s}$, $\bar{n}$ with $\tilde{n}$,  and $\mathbf{x}_{[\bar{s}t]}$ with $\theta^{{\mathcal G}_{\mathcal B}}\odot \mathbf{x}_{[\bar{s}t]}$, the conclusion can be obtained since $\theta^{{\mathcal G}_{\mathcal B}}\odot \mathbf{x}_{[\bar{s}t]}$ are $\bar{s}\delta$ distinct elements over $\mathbb{F}_q$ such that $\boldsymbol{\varPsi}_{\tilde{n},a}(\boldsymbol{\varphi}_{\mathcal B}^{(\delta)}(\theta^{{\mathcal G}_{\mathcal B}}\odot \mathbf{x}_{[\bar{s}t]}))$ is invertible.
\end{proof}
Based on the notations in \cite{LWHY}, we introduce the following definitions and notations, which can be used in the repair of multiple node failures under rack-aware storage model. For any $b\in[\bar{s}]$,
we employ $\mathbf{e}_b$ to represent the $b$-th row vector of ${\bf I}_{\bar{s}}$. In other terms, we
have ${\bf I}_{\bar{s}}=(\mathbf{e}_0^{\top},\mathbf{e}_1^{\top},\ldots,\mathbf{e}_{\bar{s}-1}^{\top})^{\top}$. Note that $l=\bar{s}^{\tilde{n}}$, we set $\bar{l}=\bar{s}^{\tilde{n}-1}$.
So, for any $a\in[\tilde{n}]$ and $b\in[\bar{s}]$, we define an $\bar{l}\times l$
matrix
\begin{align}\label{eq6}
	\mathbf{R}_{a,b}={\bf I}_{\bar{s}^{\tilde{n}-a-1}}\otimes \mathbf{e}_b\otimes {\bf I}_{\bar{s}^a}.
\end{align}
Then, \eqref{eq6} can be rewritten as
\begin{align*}
	\mathbf{R}_{a,b}={\bf I}_{\bar{s}^{\tilde{n}-a-1}}\otimes ({\bf I}_{\bar{s}^a}\boxtimes \mathbf{e}_b),
\end{align*}
if $\mathbf{e}_b$ can be viewed as a $1\times s$ block matrix. It is not difficult to check that 
\begin{align*}
	\sum_{z\in[\bar{s}]}\mathbf{R}_{a,z}^{\top}\mathbf{R}_{a,z}={\bf I}_{l}.
\end{align*}
If $\mathbf{A}$ is an $l\times l$ matrix, one can verify that $\mathbf{R}_{a,b}\mathbf{A}$ is the submatrix which is obtained by extracting rows from $\mathbf{A}$ provided that the row index $i$ satisfies $i_a=b$. In a similar way, $\mathbf{A}\mathbf{R}_{a,b}^{\top}$ is the submatrix that is formed by extracting columns from $\mathbf{A}$, with the column index $j$ satisfying $j_a=b$.

Furthermore, we employ $\mathbf{1}$ to represent the all 1 row vector of length
$\bar{s}$ over $\mathbb{F}_q$. For $a\in[\tilde{n}]$, set
\begin{align}\label{eq6-1}
	\mathbf{R}_{a,\bar{s}}={\bf I}_{\bar{s}^{\tilde{n}-a-1}}\otimes {\bf 1}\otimes {\bf I}_{\bar{s}^a},
\end{align}
we can obtain $\mathbf{R}_{a,\bar{s}}=\sum_{z\in[\bar{s}]}\mathbf{R}_{a,z}$.

For $a\in[\tilde{n}-1]$, $b\in[\bar{s}+1]$ and a positive integer $t$, we define the following $\bar{l}t\times \bar{l}$ matrix 
\begin{align}\label{df-varPsi}
	\bar{\varPsi}_{\tilde{n},a}(\varphi_b^{(t)}(\mathbf{x}_{[\bar{s}]}))={\bf I}_{\bar{s}^{\tilde{n}-a-2}}\otimes({\bf I}_{\bar{s}^a}\boxtimes \varphi_b^{(t)}(\mathbf{x}_{[\bar{s}]})).
\end{align}
Likewise, we regard $\varphi_b^{(t)}(\mathbf{x}_{[\bar{s}]})$ as an $\bar{s}\times \bar{s}$ block matrix,  with its block entries being column vectors of length  $t$. 

The following two lemmas will be applied in the repair process of our rack-aware MSR codes. 
\begin{lemma}\label{lm4}
	For any $a\in[\tilde{n}]$, $b,z\in[\bar{s}]$, and any positive integer
	$t$, we can get the following results.\\
	{\rm (1)} \begin{center}
		$(\mathbf{R}_{a,b}\otimes {\bf I}_t)\varPsi_{\tilde{n},a}(\varphi_b^{(t)}(\mathbf{x}_{[\bar{s}]}))\mathbf{R}_{a,z}^{\top}=
		\begin{cases}
			{\bf I}_{\bar{l}}\otimes L^{(t)}(x_b) & {\rm if\ } z=b, \\	
			-{\bf I}_{\bar{l}}\otimes L^{(t)}(x_z) & {\rm if\ } z\neq b.
		\end{cases}$
	\end{center}
	{\rm (2)} For any $h\in[\bar{s}+1]\setminus\{b\}$,
	\begin{center}
		$(\mathbf{R}_{a,b}\otimes {\bf I}_t)\varPsi_{\tilde{n},a}(\varphi_h^{(t)}(\mathbf{x}_{[\bar{s}]}))\mathbf{R}_{a,z}^{\top}=
		\begin{cases}
			{\bf I}_{\bar{l}}\otimes L^{(t)}(x_b) & {\rm if\ } z=b, \\	
			0 & {\rm if\ } z\neq b.
		\end{cases}$
	\end{center}
	{\rm (3)} For any $e\in[\tilde{n}]\setminus\{a\}$, $h\in[\bar{s}+1]$,
	\begin{center}
		$(\mathbf{R}_{a,b}\otimes {\bf I}_t)\varPsi_{\tilde{n},e}(\varphi_h^{(t)}(\mathbf{x}_{[\bar{s}]}))\mathbf{R}_{a,z}^{\top}=
		\begin{cases}
			\bar{\varPsi}_{\tilde{n},\bar{e}}(\varphi_h^{(t)}(\mathbf{x}_{[\bar{s}]}))	 & {\rm if\ } z=b, \\	
			0 & {\rm if\ } z\neq b,
		\end{cases}$
	\end{center}
	where 
	\begin{center}
		$\bar{e}=
		\begin{cases}
			e	 & {\rm if\ } e<a, \\	
			e-1 & {\rm if\ } e>a.
		\end{cases}$
	\end{center}
\end{lemma}
\begin{proof}
	In Lemma 4 of \cite{LWHY}, by replacing $s$ with $\bar{s}$, $\bar{n}$ with $\tilde{n}$, the conclusion can be obtained.
\end{proof}
\begin{lemma}\label{lm5}
	For any $a\in [\tilde{n}]$, $b,z\in[\bar{s}]$, and any positive integer
	$t$, we can get the following results.\\
	{\rm (1)} \begin{center}
		$(\mathbf{R}_{a,\bar{s}}\otimes {\bf I}_t)\varPsi_{\tilde{n},a}(\varphi_{\bar{s}}^{(t)}(\mathbf{x}_{[\bar{s}]}))\mathbf{R}_{a,z}^{\top}=
		{\bf I}_{\bar{l}}\otimes L^{(t)}(x_z).$
	\end{center}
	{\rm (2)} For any $b\in[\bar{s}]$,
	\begin{center}
		$(\mathbf{R}_{a,\bar{s}}\otimes {\bf I}_t)\varPsi_{\tilde{n},a}(\varphi_b^{(t)}(\mathbf{x}_{[\bar{s}]}))\mathbf{R}_{a,z}^{\top}=
		\begin{cases}
			{\bf I}_{\bar{l}}\otimes L^{(t)}(x_b) & {\rm if\ } z=b, \\	
			0 & {\rm if\ } z\neq b.
		\end{cases}$
	\end{center}
	{\rm (3)} For any $e\in[\tilde{n}]\setminus\{a\}$, $h\in[\bar{s}+1]$,
	\begin{center}
		$(\mathbf{R}_{a,\bar{s}}\otimes {\bf I}_t)\varPsi_{\tilde{n},e}(\varphi_h^{(t)}(\mathbf{x}_{[\bar{s}]}))\mathbf{R}_{a,z}^{\top}=
		\bar{\varPsi}_{\tilde{n},\bar{e}}(\varphi_h^{(t)}(\mathbf{x}_{[\bar{s}]})), $
	\end{center}
	where $\bar{e}$ is defined the same as that in Lemma \ref{lm4}.
\end{lemma}
\begin{proof}
From the relation $\mathbf{R}_{a,\bar{s}}=\sum_{b\in[\bar{s}]}\mathbf{R}_{a,b}$, for any $e\in[\tilde{n}]$ and $h\in[\bar{s}+1]$, we obtain
\begin{align*}
	(\mathbf{R}_{a,\bar{s}}\otimes {\bf I}_t)\varPsi_{\tilde{n},e}(\varphi_h^{(t)}(\mathbf{x}_{[\bar{s}]}))\mathbf{R}_{a,z}^{\top}=\sum_{b\in[\bar{s}]}(\mathbf{R}_{a,b}\otimes {\bf I}_t)\varPsi_{\tilde{n},e}(\varphi_h^{(t)}(\mathbf{x}_{[\bar{s}]}))\mathbf{R}_{a,z}^{\top}.
\end{align*}
Hence the conclusion follows directly from  the above Lemma.
\end{proof}

For convenience, the parameters and notations for the array codes are listed in Table \ref{tab2}.

\begin{table}[http!]
	\renewcommand{\arraystretch}{1.2}
	\setlength\tabcolsep{3pt} 
	\centering
	\caption{ Code Parameters and Notation
		\label{tab2}
	} 
	\begin{tabular}{|c|c|c|c|c|c|c|c|c|}	
		\hline
		Notation & Meaning & Notation &	 Meaning \\ 
		\hline
		$n$ & code length & $k$ &	 code dimension/the number of systematic nodes \\ 		
		\hline
		$d$ &	the number of helper nodes  & $h$ &	the number of helper nodes\\
		\hline
		$r$  &	the number of parity nodes $n-k$ & $l$ &  the sub-packetization level of  $\bar{s}^{\tilde{n}}$\\
		\hline
		$\bar{n}$ & the number of racks in the code & $u$ & the size of a rack\\	
		\hline
		$v$ & an integer within $[u]$ & $\bar{k}$ & the number $(k-v)/u$ of systematic racks \\		
		\hline	
		$\bar{r}$ & the number $\bar{n}-\bar{k}$ of  parity racks  &	$\bar{d}$ & the number of helper racks\\		
		\hline	
		$\bar{s}$ & $\bar{d}-\bar{k}+1$ &	$\bar{h}$ &  the number of host racks\\		
		\hline	
		$\beta$ & repair bandwidth & $\gamma$&	the number of symbols accessed on the helper racks\\
		\hline	
			$\bar{l}$  & $\bar{s}^{\tilde{n}-1}$ & $\boxtimes$  & block Kronecker product in \eqref{eq1-3}  \\
			\hline	
		$\mathcal{C}$& the rack-aware $(n,k,l)$ MSR code & $\bar{\mathcal{C}}(w)$	& the intermediate $(\bar{n},\bar{k},l)$ MSR code\\
		\hline
		$\mathbf{H}_{iu+g}$&the parity check sub-matrix of $\mathbf{c}_{iu+g}$ in the code $\mathcal{C}$&$\bar{\mathbf{H}}_{i}$&the parity check sub-matrix of $\bar{\mathbf{c}}_{i}(w)$ in the code $\bar{\mathcal{C}}(w)$\\
		\hline
		$a\in[\tilde{n}]$&  rack group index& $b\in[\bar{s}+1]$& in rack group index \\
		\hline	
		 $\varphi_b^{(t)}(\cdot)$  & the kernel map in \eqref{eq2} or \eqref{eq2-1-1}& $\varPsi_{\tilde{n},a}(\cdot)$&blow-up transformation in \eqref{eq1} \\
		\hline	
		$\mathbf{e}_i$  & the $i$-th row vector of ${\bf I}_{\bar{s}}$ & $\mathbf{1}$ & the all 1 row vector of length
		$\bar{s}$\\
		\hline	
		$\mathbf{e}_i'$&the $i$-th row vector of ${\bf I}_{\bar{r}}$ & $\mathbf{R}_{a,b}$& the $\bar{l}\times l$ repair matrix in \eqref{eq6}\\
		\hline	
		$\mathbf{R}_{a,\bar{s}}$&the $\bar{l}\times l$ repair matrix in \eqref{eq6-1}& ${\mathcal H}$& the index set for the helper racks\\
		\hline
	\end{tabular}
\end{table}

\section{The main results}
\label{sect-main}
In this section, by employing the alignment technique to extend the coupled-layer construction in \cite{LWHY}, we propose two classes of rack-aware MSR codes with smaller sub-packetization level and linear field size, each of which contains a rack-aware MSR code and an asymptotically rack-aware MSR code.

\begin{theorem}\label{th1} 
For any positive integers $\bar{n}$, $\bar{k}$, $\bar{d}$, $h$ and $q$  satisfying  $\bar{d}\in[\bar{k}+1,\bar{n}-1]$ and $q\geq n\bar{s}+u(\Omega(\bar{s},u)+(\bar{s}-1)2^{\bar{s}-2})$ where $\Omega(\bar{s},u)$ is defined in \eqref{eq3-6-1-1},  there exist 
\begin{itemize}
\item an $(n,k,l=\bar{s}^{\lceil \bar{n}/\bar{s} \rceil})$ rack-aware MSR code over $\mathbb{F}_q$ that achieves optimal repair bandwidth for the repair of any $h\in [1, u-v]$ failed nodes within the same rack, and in particular, achieves optimal access when $h=u-v$.  
\item  an $(n,k,l=\bar{s}^{\lceil \bar{n}/\bar{s} \rceil})$ rack-aware MSR code over $\mathbb{F}_q$ that  achieves asymptotically optimal repair bandwidth for the repair of any $h\in[u-v+1,u]$ failed nodes within the same rack.
\end{itemize}
\end{theorem}
In the proof of Theorem \ref{th1} we first construct the rack-aware MSR code under the assumption that $\bar{n}$ is divisible by $\bar{s}$. Based on this code, when $\bar{n}$ is not divisible by $\bar{s}$, the construction can be easily derived by truncating a slightly longer code from the divisible case. For the details, please see Section \ref{Sect3}. By modifying the construction of rack-aware MSR codes in Theorem \ref{th1}, we can further reduce the sub-packetization level of the rack-aware MSR codes while maintaining a linear field size. That is the following result whose proof is included in Section \ref{Sect4}.
%However, the code does not achieve optimal access when $h=u-v$.
\begin{theorem}\label{th2} 
	For any positive integers $\bar{n}$, $\bar{k}$, $\bar{d}$, $h$ and $q$ satisfying  $\bar{d}\in[\bar{k}+1,\bar{n}-1]$ and $q\geq n\bar{s}+u(\Omega(\bar{s}+1,u)+\bar{s}2^{\bar{s}-1})$, there exist 
	\begin{itemize}
		\item  an $(n,k,l=\bar{s}^{\lceil\bar{n}/(\bar{s}+1)\rceil})$ rack-aware MSR code over $\mathbb{F}_q$ that achieves optimal repair bandwidth when repairing any $h \in [1, u-v]$ failed nodes within the same rack;
		\item an $(n,k,l=\bar{s}^{\lceil\bar{n}/(\bar{s}+1)\rceil})$ rack-aware MSR code over $\mathbb{F}_q$ that  achieves asymptotically optimal repair bandwidth when repairing any $h \in [u-v+1, u]$ failed nodes within the same rack.
		
	\end{itemize}
\end{theorem} 
%Similarly, for the $(n, k, l=\bar{s}^{\lceil\bar{n}/(\bar{s}+1)\rceil})$ rack-aware MSR code, we assume $\bar{n}$ is divisible by $\bar{s}+1$.

\subsection{An example of construction in Theorem \ref{th1}}

In this subsection, we provide an $(n=8,k=4,l=4)$ rack-aware MSR code to illustrate the main ideas of the code construction in Theorem \ref{th1}. First, we present the construction of a rack-aware MSR code.
Next, we demonstrate the idea of proving the MDS property. Finally, we present a repair scheme for the code that achieves both optimal repair bandwidth and optimal access.  
\subsubsection{ Construct a rack-aware MSR code with $l=4$}
When $(n,u,k,r,d,h,v)=(8,2,4,4,6,2,0)$, we have
$\bar{n}=n/u=4$, $\bar{k}=k/u=2$, $\bar{r}=r/u=2$,
$\bar{d}=d/u=3$, $\bar{s}=\bar{d}-\bar{k}+1=2$, $\tilde{n}=\bar{n}/\bar{s}=2$ and	$l=\bar{s}^{\tilde{n}}=4$. Let $\xi$ be a root of the primitive polynomial $x^3+2x+1$ over
$\mathbb{F}_{27}$, $\theta=-1$ be an element in $\mathbb{F}_{27}$ with multiplicative order
$2$, and $\lambda_{i}=\xi^i$ for each integer $i\in[8]$. For any $b\in[2]$ and $t\in[4]$, from \eqref{eq2} we have 
\begin{align*}
	&\varphi_0^{(t)}(\boldsymbol{\lambda}_{[2]})={\left[\begin{array}{llll}
			L^{(t)}(\lambda_0) & -L^{(t)}(\lambda_1) \\			
			0 & L^{(t)}(\lambda_1) \\	
		\end{array}
		\right]},
	&\varphi_0^{(t)}(-\boldsymbol{\lambda}_{[2]})={\left[\begin{array}{llll}
			L^{(t)}(-\lambda_0) & -L^{(t)}(-\lambda_1) \\			
			0 & L^{(t)}(-\lambda_1) \\	
		\end{array}
		\right]},\\
	&	\varphi_1^{(t)}(\boldsymbol{\lambda}_{2+[2]})={\left[\begin{array}{llll}
			L^{(t)}(\lambda_2) & 0 \\			
			-L^{(t)}(\lambda_2) & L^{(t)}(\lambda_3) \\	
		\end{array}
		\right]},
	&	\varphi_1^{(t)}(-\boldsymbol{\lambda}_{2+[2]})={\left[\begin{array}{llll}
			L^{(t)}(-\lambda_2) & 0 \\			
			-L^{(t)}(-\lambda_2) & L^{(t)}(-\lambda_3) \\	
		\end{array}
		\right]},\\
	&\varphi_0^{(t)}(\boldsymbol{\lambda}_{4+[2]})={\left[\begin{array}{llll}
			L^{(t)}(\lambda_4) & -L^{(t)}(\lambda_5) \\			
			0 & L^{(t)}(\lambda_5) \\	
		\end{array}
		\right]},
	&\varphi_0^{(t)}(-\boldsymbol{\lambda}_{4+[2]})={\left[\begin{array}{llll}
			L^{(t)}(-\lambda_4) & -L^{(t)}(-\lambda_5) \\			
			0 & L^{(t)}(-\lambda_5) \\	
		\end{array}
		\right]},\\
	&	\varphi_1^{(t)}(\boldsymbol{\lambda}_{6+[2]})={\left[\begin{array}{llll}
			L^{(t)}(\lambda_6) & 0 \\			
			-L^{(t)}(\lambda_6) & L^{(t)}(\lambda_7) \\	
		\end{array}
		\right]},
	&	\varphi_1^{(t)}(-\boldsymbol{\lambda}_{6+[2]})={\left[\begin{array}{llll}
			L^{(t)}(-\lambda_6) & 0 \\			
			-L^{(t)}(-\lambda_6) & L^{(t)}(-\lambda_7) \\	
		\end{array}
		\right]}.
\end{align*}
For each integer $t\in[4]$, from \eqref{eq1} we can obtain 
\begin{align*}
	&\varPsi_{2,0}(\varphi_0^{(t)}(\boldsymbol{\lambda}_{[2]}))=\mathbf{I}_2\otimes \varphi_0^{(t)}(\boldsymbol{\lambda}_{[2]}); \ \ \ \ \ \ \ \varPsi_{2,0}(\varphi_0^{(t)}(-\boldsymbol{\lambda}_{[2]}))=\mathbf{I}_2\otimes \varphi_0^{(t)}(-\boldsymbol{\lambda}_{[2]});\\
	&\varPsi_{2,0}(\varphi_1^{(t)}(\boldsymbol{\lambda}_{2+[2]}))=\mathbf{I}_2\otimes \varphi_1^{(t)}(\boldsymbol{\lambda}_{2+[2]}); \ \varPsi_{2,0}(\varphi_1^{(t)}(-\boldsymbol{\lambda}_{2+[2]}))=\mathbf{I}_2\otimes \varphi_1^{(t)}(-\boldsymbol{\lambda}_{2+[2]});\\
	&\varPsi_{2,1}(\varphi_0^{(t)}(\boldsymbol{\lambda}_{4+[2]}))=\mathbf{I}_2\boxtimes \varphi_0^{(t)}(\boldsymbol{\lambda}_{4+[2]}); \  \varPsi_{2,1}(\varphi_0^{(t)}(-\boldsymbol{\lambda}_{4+[2]}))=\mathbf{I}_2\boxtimes \varphi_0^{(t)}(-\boldsymbol{\lambda}_{4+[2]});\\
	&\varPsi_{2,1}(\varphi_1^{(t)}(\boldsymbol{\lambda}_{6+[2]}))=\mathbf{I}_2\boxtimes \varphi_1^{(t)}(\boldsymbol{\lambda}_{6+[2]}); \ \varPsi_{2,1}(\varphi_1^{(t)}(-\boldsymbol{\lambda}_{6+[2]}))=\mathbf{I}_2\boxtimes \varphi_1^{(t)}(-\boldsymbol{\lambda}_{6+[2]}).
\end{align*} 
Let us consider an array code $\mathcal{C}$ defined by the following parity check equations over $\mathbb{F}_{27}$.
\begin{equation}\label{eq3-1e}
	\mathcal{C}=\{({\bf c}_0,{\bf c}_1,\ldots,{\bf c}_{7}):\sum_{a\in[2]}\sum_{b\in[2]}\sum_{g\in[2]}\mathbf{H}_{4a+2b+g}{\bf c}_{4a+2b+g}=0\},
\end{equation}
\noindent where 
\begin{align}\label{eq3-2e}
\mathbf{H}_{4a+2b+g}=\varPsi_{2,a}(\varphi_b^{(4)}((-1)^g\boldsymbol{\lambda}_{4a+2b+[2]})), a\in[2], b\in[2], g\in[2].
\end{align}
To provide a clear description of the code construction, the sequence of parity check submatrices $\mathbf{H}_0, \mathbf{H}_1, \ldots, \mathbf{H}_7$ is presented as follows:
\begin{align*}
	&\mathbf{H}_0={\left[\begin{array}{llllll}
			L^{(4)}(\lambda_0) & -L^{(4)}(\lambda_1) & 0&0\\			
			0 & L^{(4)}(\lambda_1) & 0&0\\
			0& 0& L^{(4)}(\lambda_0) & -L^{(4)}(\lambda_1)\\
			0& 0& 0 & L^{(4)}(\lambda_1)
		\end{array}
		\right]},
	&\mathbf{H}_1={\left[\begin{array}{llllllll}
			L^{(4)}(-\lambda_0) & -L^{(4)}(-\lambda_1) & 0&0\\			
			0 & L^{(4)}(-\lambda_1) & 0&0\\
			0& 0& L^{(4)}(-\lambda_0) & -L^{(4)}(-\lambda_1)\\
			0& 0& 0 & L^{(4)}(-\lambda_1)
		\end{array}
		\right]},\\
	&\mathbf{H}_2={\left[\begin{array}{llllllll}
			L^{(4)}(\lambda_2) & 0 & 0&0\\			
			-L^{(4)}(\lambda_2) & L^{(4)}(\lambda_3) & 0&0\\
			0 & 0 & L^{(4)}(\lambda_2) & 0\\			
			0&0 &-L^{(4)}(\lambda_2) & L^{(4)}(\lambda_3)
		\end{array}
		\right]},
	&\mathbf{H}_3={\left[\begin{array}{llllllll}
			L^{(4)}(-\lambda_2) & 0 & 0&0\\			
			-L^{(4)}(-\lambda_2) & L^{(4)}(-\lambda_3) & 0&0\\
			0 & 0 & L^{(4)}(-\lambda_2) & 0\\			
			0 & 0 &-L^{(4)}(-\lambda_2) & L^{(4)}(-\lambda_3)
		\end{array}
		\right]},
\end{align*}
\begin{align*}
	&\mathbf{H}_4={\left[\begin{array}{llllllll}
			L^{(4)}(\lambda_4) & 0 & -L^{(4)}(\lambda_5)&0\\			
			0&L^{(4)}(\lambda_4) & 0 & -L^{(4)}(\lambda_5)\\
			0 & 0 &L^{(4)}(\lambda_5)&0\\			
			0&0 & 0 & L^{(4)}(\lambda_5)\\
		\end{array}
		\right]},
	&\mathbf{H}_5={\left[\begin{array}{llllllll}
			L^{(4)}(-\lambda_4) & 0 & -L^{(4)}(-\lambda_5)&0\\			
			0&L^{(4)}(-\lambda_4) & 0 & -L^{(4)}(-\lambda_5)\\
			0 & 0 &L^{(4)}(-\lambda_5)&0\\			
			0&0 & 0 & L^{(4)}(-\lambda_5)\\
		\end{array}
		\right]},\\
	&\mathbf{H}_6={\left[\begin{array}{llllllll}
			L^{(4)}(\lambda_6) & 0 & 0&0\\			
			0&L^{(4)}(\lambda_6) & 0 & 0\\
			-L^{(4)}(\lambda_6) & 0 &L^{(4)}(\lambda_7)&0\\			
			0&-L^{(4)}(\lambda_6) & 0 & L^{(4)}(\lambda_7)\\
		\end{array}
		\right]},
	&\mathbf{H}_7={\left[\begin{array}{llllllll}
			L^{(4)}(-\lambda_6) & 0 & 0&0\\			
			0&L^{(4)}(-\lambda_6) & 0 & 0\\
			-L^{(4)}(-\lambda_6) & 0 &L^{(4)}(-\lambda_7)&0\\			
			0&-L^{(4)}(-\lambda_6) & 0 & L^{(4)}(-\lambda_7)\\
		\end{array}
		\right]}.
\end{align*}

Let $\mathbf{e}_i'$ be $i$-th row of ${\bf I}_4$, and $\mathbf{Q}_w=[(\mathbf{e}_w')^{\top},(\mathbf{e}_{2+w}')^{\top}]^{\top}$ where $w\in[2]$. Note that ${\bf I}_{2^a}\otimes({\bf I}_2\otimes \mathbf{Q}_w)={\bf I}_{2^a}\boxtimes({\bf I}_2\otimes \mathbf{Q}_w)$ where $w\in[2]$, it is easy to see that 
\begin{align*}
	{\bf I}_4\otimes Q_w=&{\bf I}_{2^{2-a-1}}\otimes({\bf I}_{2^a}\otimes({\bf I}_{2}\otimes Q_w))\\
	=&{\bf I}_{2^{2-a-1}}\otimes({\bf I}_{2^a}\boxtimes({\bf I}_{2}\otimes Q_w))\\
	=&\varPsi_{2,a}({\bf I}_2\otimes \mathbf{Q}_w).
\end{align*}
By the parity check equation of this code we have
\begin{align}
	&({\bf I}_4\otimes Q_w)\sum_{a\in[2]}\sum_{b\in[2]}\sum_{g\in[2]}\mathbf{H}_{4a+2b+g}{\bf c}_{4a+2b+g}	\nonumber\\
=&\sum_{a\in[2]}\sum_{b\in[2]}\sum_{g\in[2]}\varPsi_{2,a}({\bf I}_2\otimes \mathbf{Q}_w)\varPsi_{2,a}(\varphi_b^{(4)}((-1)^g\boldsymbol{\lambda}_{4a+2b+[2]})){\bf c}_{4a+2b+g}	\nonumber \\
=&\sum_{a\in[2]}\sum_{b\in[2]}\sum_{g\in[2]}\varPsi_{2,a}(({\bf I}_2\otimes \mathbf{Q}_w)\varphi_b^{(4)}((-1)^g\boldsymbol{\lambda}_{4a+2b+[2]})){\bf c}_{4a+2b+g}	\label{eq3-6-1}\\
=&\sum_{a\in[2]}\sum_{b\in[2]}\sum_{g\in[2]}\varPsi_{2,a}((-1)^{gw}\varphi_b^{(2)}(\boldsymbol{\lambda}_{4a+2b+[2]}^2){\rm diag}(\boldsymbol{\lambda}_{4a+2b+[2]}^w)){\bf c}_{4a+2b+g}	\nonumber\\
=&\sum_{a\in[2]}\sum_{b\in[2]}\sum_{g\in[2]}(-1)^{gw}\varPsi_{2,a}(\varphi_b^{(2)}(\boldsymbol{\lambda}_{4a+2b+[2]}^2))\varPsi_{2,a}({\rm diag}(\boldsymbol{\lambda}_{4a+2b+[2]}^w)){\bf c}_{4a+2b+g}	\nonumber\\
=&\sum_{a\in[2]}\sum_{b\in[2]}\varPsi_{2,a}(\varphi_b^{(2)}(\boldsymbol{\lambda}_{4a+2b+[2]}^2))\varPsi_{2,a}({\rm diag}(\boldsymbol{\lambda}_{4a+2b+[2]}^w))\sum_{g\in[2]}(-1)^{gw}{\bf c}_{4a+2b+g}	\nonumber\\
=&0.	\label{eq3-6-2}
\end{align}
Here, the equation \eqref{eq3-6-1} can be derived by Lemma \ref{lm1-1}. For each integer $w\in[2]$ the equation \eqref{eq3-6-2} can be written as 
\begin{align}\label{eq3-6-3}
	\sum_{a\in[2]}\sum_{b\in[2]}\bar{\mathbf{H}}_{2a+b}\bar{\bf c}_{2a+b}(w)=0,
\end{align}
where for $a\in[2]$, $b\in[2]$
\begin{align}
\bar{\mathbf{H}}_{2a+b}&=\varPsi_{2,a}(\varphi_b^{(2)}(\boldsymbol{\lambda}_{4a+2b+[2]}^2)),\ \ \ \text{and}\nonumber\\ 
\bar{\bf c}_{2a+b}(w)&=\varPsi_{2,a}({\rm diag}(\boldsymbol{\lambda}_{4a+2b+[2]}^w))\sum_{g\in[2]}(-1)^{gw}{\bf c}_{4a+2b+g}.\label{eq3-6-2-1}
\end{align}
Clearly,  for each $w\in[2]$ the equation \eqref{eq3-6-3} defines a new array code $\bar{\mathcal{C}}(w)=(\bar{\bf c}_0(w),\bar{\bf c}_1(w),\bar{\bf c}_2(w),\bar{\bf c}_3(w))$ with code length $4$ and dimension $2$. 
Specifically,  the parity check submatrices in \eqref{eq3-6-3} are provided by 
\begin{align*}
	&\bar{\mathbf{H}}_0={\left[\begin{array}{llllll}
			L^{(2)}(\lambda_0^2) & -L^{(2)}(\lambda_1^2) & 0&0\\			
			0 & L^{(2)}(\lambda_1^2) & 0&0\\
			0& 0& L^{(2)}(\lambda_0^2) & -L^{(2)}(\lambda_1^2)\\
			0& 0& 0 & L^{(2)}(\lambda_1^2)
		\end{array}
		\right]},
	&\bar{\mathbf{H}}_1={\left[\begin{array}{llllllll}
			L^{(2)}(\lambda_2^2) & 0 & 0&0\\			
			-L^{(2)}(\lambda_2^2) & L^{(2)}(\lambda_3^2) & 0&0\\
			0& 0& L^{(2)}(\lambda_2^2) & 0\\
			0& 0& -L^{(2)}(\lambda_2^2) & L^{(2)}(\lambda_3^2)
		\end{array}
		\right]},\\
	&\bar{\mathbf{H}}_2={\left[\begin{array}{llllllll}
			L^{(2)}(\lambda_4^2) & 0 &-L^{(2)}(\lambda_5^2)&0\\			
			0 & L^{(2)}(\lambda_4^2) & 0 &-L^{(2)}(\lambda_5^2)\\
			0 & 0 & L^{(2)}(\lambda_5^2)& 0\\			
			0&0 &0 & L^{(2)}(\lambda_5^2)
		\end{array}
		\right]},
	&\bar{\mathbf{H}}_3={\left[\begin{array}{llllllll}
			L^{(2)}(\lambda_6^2) & 0 &0&0\\			
			0 & L^{(2)}(\lambda_6^2) & 0 &0\\
			-L^{(2)}(\lambda_6^2) & 0 & L^{(2)}(\lambda_7^2)& 0\\			
			0&-L^{(2)}(\lambda_6^2) &0 & L^{(2)}(\lambda_7^2)
		\end{array}
		\right]}.
\end{align*}

\subsubsection{MDS property}
For each integer $a\in[2]$, a non-empty subset ${\mathcal B}=\{b_0,b_1,\ldots,b_{t-1}\}\subseteq [2]$ with $b_0<b_1<\ldots<b_{t-1}$, we define
\begin{align*}
	\mathbf{x}_{2{\mathcal B}+[2]}=(x_{2b+j}:b\in{\mathcal B},j\in[2]).
\end{align*}
Given a set ${\mathcal G}_{\mathcal B}=\{{\mathcal G}_{b_0},{\mathcal G}_{b_1},\ldots,{\mathcal G}_{b_{t-1}}\}$ satisfying that ${\mathcal G}_{b_j}\subseteq[2]$, $|{\mathcal G}_{b_j}|=m_j\geq1$ for $j\in [t]$, and $|{\mathcal G}_{\mathcal B}|=m_0+m_1+\ldots+m_{t-1}=\delta$, from \eqref{eq4} and \eqref{eq501} we have the $2m\times 2\delta$ matrix
\begin{align}\label{eq3-3e}
	\nonumber	\boldsymbol{\varphi}_{\mathcal B}^{(m)}((-1)^{{\mathcal G}_{\mathcal B}}\odot \boldsymbol{\lambda}_{4a+2{\mathcal B}+[2]})	&=[\varphi_{b_0}^{(m)}((-1)^{g_{b_0,0}}\boldsymbol{\lambda}_{4a+2b_0+[2]})\ldots \varphi_{b_0}^{(m)}((-1)^{g_{b_0,m_0-1}}\boldsymbol{\lambda}_{4a+2b_0+[2]})\\ & \ldots \varphi_{b_{t-1}}^{(m)}((-1)^{g_{b_{t-1},0}}\boldsymbol{\lambda}_{4a+2b_{t-1}+[2]})\ldots \varphi_{b_{t-1}}^{(m)}((-1)^{g_{b_{t-1},m_{t-1}-1}}\boldsymbol{\lambda}_{4a+2b_{t-1}+[2]})],
\end{align}
and the $2m\times 2t $ matrix
\begin{align}\label{eq3-17-1e}
	\boldsymbol{\varphi}_{\mathcal B}^{(m)}(\boldsymbol{\lambda}_{4a+2{\mathcal B}+[2]}^2)	=[\varphi_{b_0}^{(m)}(\boldsymbol{\lambda}_{4a+2b_0+[2]}^2) \ \varphi_{b_1}^{(m)}(\boldsymbol{\lambda}_{4a+2b_1+[2]}^2) \ldots \varphi_{b_{t-1}}^{(m)}(\boldsymbol{\lambda}_{4a+2b_{t-1}+[2]}^2)],
\end{align}
where $m\geq \delta$. It is not difficult to verify that the elements $\boldsymbol{\lambda}_{[8]}$ satisfy 
\begin{align*}
\text{det}(\boldsymbol{\varphi}_{\mathcal B}^{(\delta)}((-1)^{{\mathcal G}_{\mathcal B}}\odot \boldsymbol{\lambda}_{4a+2{\mathcal B}+[2]}))\not=0,
\end{align*}
and 
\begin{align}\label{ep3}
	\text{det}(\boldsymbol{\varphi}_{\mathcal B}^{(t)}(\boldsymbol{\lambda}_{4a+2{\mathcal B}+[2]}^2))\not=0.
\end{align}
According to Lemma \ref{lm2}, we know that 
\begin{align} \label{ep1}
	\text{det}(\boldsymbol{\varPsi}_{2,a}(\boldsymbol{\varphi}_{\mathcal B}^{(\delta)}((-1)^{{\mathcal G}_{\mathcal B}}\odot \boldsymbol{\lambda}_{4a+2{\mathcal B}+[2]})))\not=0, a\in[2], \emptyset\not={\mathcal B}\subseteq [2], \emptyset\not={\mathcal G}_{b_j}\subseteq[2], j\in[t].
\end{align}
and 
\begin{align}\label{ep2}
	\text{det}(\boldsymbol{\varPsi}_{2,a}(\boldsymbol{\varphi}_{\mathcal B}^{(t)}(\boldsymbol{\lambda}_{4a+2{\mathcal B}+[2]}^2)))\not=0, a\in[2], \emptyset\not={\mathcal B}\subseteq [2].
\end{align}
%We term the conditions presented in \eqref{ep1} as local  constraints of the code defined in \eqref{eq3-1e}. 
To illustrate the proof of the MDS property of this code, we here will show how to recover ${\bf c_0}, {\bf c_3}, {\bf c_4}, {\bf c_5}$ from the remaining nodes ${\bf c_1}, {\bf c_2}, {\bf c_6}, {\bf c_7}$. We claim that the square matrix 
\begin{align*}
	[\mathbf{H}_0 \ \mathbf{H}_3 \ \mathbf{H}_4 \ \mathbf{H}_5]=[\varPsi_{2,0}(\varphi_0^{(4)}(\boldsymbol{\lambda}_{[2]})) \ \varPsi_{2,0}(\varphi_1^{(4)}(-\boldsymbol{\lambda}_{2+[2]})) \ \varPsi_{2,1}(\varphi_0^{(4)}(\boldsymbol{\lambda}_{4+[2]})) \ \varPsi_{2,1}(\varphi_0^{(4)}(-\boldsymbol{\lambda}_{4+[2]}))] 
\end{align*} 
 is invertible by the following reason. Since the elements $\boldsymbol{\lambda}_{[8]}$ satisfy the conditions in \eqref{ep1}, the square matrices 
\begin{align*}
	\boldsymbol{\varPsi}_{2,0}(\boldsymbol{\varphi}_{{\mathcal B}_0}^{(2)}((-1)^{{\mathcal G}_{{\mathcal B}_0}}\odot \boldsymbol{\lambda}_{2{{\mathcal B}_0}+[2]}))=[\varPsi_{2,0}(\varphi_0^{(2)}(\boldsymbol{\lambda}_{[2]})) \ \varPsi_{2,0}(\varphi_1^{(2)}(-\boldsymbol{\lambda}_{2+[2]}))], {\mathcal B}_0=\{0,1\}, {\mathcal G}_0=\{0\}, {\mathcal G}_1=\{1\},
\end{align*}
and 
\begin{align*}
	\boldsymbol{\varPsi}_{2,1}(\boldsymbol{\varphi}_{{\mathcal B}_1}^{(2)}((-1)^{{\mathcal G}_{{\mathcal B}_1}}\odot \boldsymbol{\lambda}_{4+2{{\mathcal B}_1}+[2]}))=[\varPsi_{2,1}(\varphi_0^{(2)}(\boldsymbol{\lambda}_{4+[2]})) \ \varPsi_{2,1}(\varphi_0^{(2)}(-\boldsymbol{\lambda}_{4+[2]}))], {\mathcal B}_1=\{0\}, {\mathcal G}_0=\{0,1\}.
\end{align*}
are invertible. To shorten our notations, denote $\boldsymbol{\varPsi}_{2,a}(\boldsymbol{\varphi}_{{\mathcal B}_a}^{(\delta)}((-1)^{{\mathcal G}_{{\mathcal B}_a}}\odot \boldsymbol{\lambda}_{4a+2{{\mathcal B}_a}+[2]}))$  by $\mathbf{M}_{a,{\mathcal B}_a}^{(\delta)}$ where $a=0,1$. By applying Lemma \ref{lm3}, there exists a $16\times 16$ matrix $\mathbf{V}$ such that 
\begin{align*}
	\mathbf{V}\cdot[\mathbf{M}_{0,{\mathcal B}_0}^{(4)} \ \mathbf{M}_{1,{\mathcal B}_1}^{(4)}]={\left[\begin{array}{llll}
			\mathbf{M}_{0,{\mathcal B}_0}^{(2)} & \mathbf{M}_{1,{\mathcal B}_1}^{(2)} \\			
			0 & \widehat{\mathbf{M}}_{1,{\mathcal B}_1}^{(2)}
		\end{array}
		\right]},
\end{align*} 
where the square matrix $\widehat{\mathbf{M}}_{1,{\mathcal B}_1}^{(2)}=\widehat{\boldsymbol{\varPsi}}_{2,1}(\boldsymbol{\varphi}_{{\mathcal B}_1}^{(2)}((-1)^{{\mathcal G}_{{\mathcal B}_1}}\odot \boldsymbol{\lambda}_{4+2{{\mathcal B}_1}+[2]}))$  is column equivalent to the matrix $\mathbf{M}_{1,{\mathcal B}_1}^{(2)}$. Hence,
\begin{align*}
	\text {det}(\mathbf{V}\cdot[\mathbf{M}_{0,{\mathcal B}_0}^{(4)} \ \mathbf{M}_{1,{\mathcal B}_1}^{(4)}])=\text{det}(\mathbf{M}_{0,{\mathcal B}_0}^{(2)})\cdot\text{det}(\widehat{\mathbf{M}}_{1,{\mathcal B}_1}^{(2)})\not=0,
\end{align*}
which implies that $[\mathbf{M}_{0,{\mathcal B}_0}^{(4)} \ \mathbf{M}_{1,{\mathcal B}_1}^{(4)}]=[\mathbf{H}_0 \ \mathbf{H}_3 \ \mathbf{H}_4 \ \mathbf{H}_5]$ is invertible.

Since the elements $\boldsymbol{\lambda}_{[8]}$ satisfy the conditions in \eqref{ep2}, similar to the above we can show that the code $\bar{\mathcal{C}}(w)$ defined in \eqref{eq3-6-3} also possesses the MDS property for $w\in[2]$.

\begin{algorithm}[H]
	\caption{repair$(\mathcal{F}, \mathcal{H})$}
	\SetAlgoLined
	\KwIn{Two subsets $\mathcal{F}=\{0,1\}$ and $\mathcal{H}=\{1,2,3\}$, which is the index set of failed nodes and the index set of helper racks respectively.}
	\KwOut{The repaired nodes $\{\mathbf{c}_0, \mathbf{c}_1\}$}
	\For{$i\in\mathcal{H}$}{
		\For{$w=0,1$}{
			Rack $i$ computes $\bar{\mathbf{c}}_{i}(w)=\varPsi_{2,\lfloor i/2\rfloor}({\rm diag}(\boldsymbol{\lambda}_{2i+[2]}^w))\sum_{g\in[2]}(-1)^{gw}{\bf c}_{2i+g}$, $\tilde{\mathbf{c}}_{i}(w)=\mathbf{R}_{0,0}\bar{\mathbf{c}}_{i}(w)$\;
			Rack $i$ transmits $\tilde{\mathbf{c}}_{i}(w)$ to rack $0$\;
		}
	}
	\For{$w=0,1$}{
	Rack $0$ repairs $\bar{\mathbf{c}}_0(w)$ from $\{\tilde{\mathbf{c}}_{i}(w):i\in\mathcal{H}\}$;
     }
	\For{$j \in \mathcal{F}$}{
		Node $j$ repairs $\mathbf{c}_j$ from $\{\bar{\mathbf{c}}_0(w):w=0,1\}$ 
	}
	\Return $\{\mathbf{c}_j: j \in \mathcal{F}\}$\;
\end{algorithm}

\subsubsection{Repair scheme of an $(8,4,4)$ rack-aware MSR code}

In the following, we will show how to repair two failed nodes within the same rack by optimal bandwidth. We only provide the process of repairing ${\bf c}_0$ and ${\bf c}_1$ from the remaining $\bar{d}=3$ helper racks. To better illustrate the repair scheme, Algorithm 1 outlines the corresponding repair algorithm. Note that the repair matrix is defined in \eqref{eq6} 
\begin{align*}
	\mathbf{R}_{0,0}={\bf I}_{2}\otimes [1\ 0].
\end{align*}
For $w\in[2]$, by the parity check equation \eqref{eq3-6-3} we have 
\begin{align}
&(\mathbf{R}_{0,0}\otimes {\bf I}_{2})\sum_{a\in[2]}\sum_{b\in[2]}\bar{\mathbf{H}}_{2a+b}\bar{\bf c}_{2a+b}(w)	\nonumber	\\ =&\sum_{a\in[2]}\sum_{b\in[2]}(\mathbf{R}_{0,0}\otimes \mathbf{I}_2)\bar{\mathbf{H}}_{2a+b}(\sum_{z\in[2]}\mathbf{R}_{0,z}^{\top}\mathbf{R}_{0,z})\bar{\bf c}_{2a+b}(w)\label{eq3-6-4}\\
=&\sum_{a\in[2]}\sum_{b\in[2]}\sum_{z\in[2]}[(\mathbf{R}_{0,0}\otimes \mathbf{I}_2)\bar{\mathbf{H}}_{2a+b}\mathbf{R}_{0,z}^{\top}](\mathbf{R}_{0,z}\bar{\bf c}_{2a+b}(w))	\nonumber	\\
=&\sum_{z\in[2]}[(\mathbf{R}_{0,0}\otimes \mathbf{I}_2)\bar{\mathbf{H}}_0\mathbf{R}_{0,z}^{\top}](\mathbf{R}_{0,z}\bar{\bf c}_0(w))+\sum_{(a,b)\in[2]\times[2]\setminus\{(0,0)\}}[(\mathbf{R}_{0,0}\otimes \mathbf{I}_2)\bar{\mathbf{H}}_{2a+b}\mathbf{R}_{0,0}^{\top}](\mathbf{R}_{0,0}\bar{\bf c}_{2a+b}(w))	\label{eq3-6-5}	\\	
=&0,	\nonumber	
\end{align}
where \eqref{eq3-6-4} holds due to $\sum_{z\in[2]}\mathbf{R}_{0,z}^{\top}\mathbf{R}_{0,z} = {\bf I}_{4}$, while \eqref{eq3-6-5} is derived directly from Lemma \ref{lm4}. So for $ w\in[2]$ the above equation becomes
\begin{align}\label{eq3-6-6}
	\widetilde{\mathbf{H}}_0^{(0)}\tilde{\bf c}_0^{(0)}(w)+\widetilde{\mathbf{H}}_0^{(1)}\tilde{\bf c}_0^{(1)}(w)+\sum_{i=1}^{3}\widetilde{\mathbf{H}}_i\tilde{\bf c}_i(w)=0,
\end{align}
where 
\begin{align*}
	\widetilde{\mathbf{H}}_0^{(z)}=(\mathbf{R}_{0,0}\otimes \mathbf{I}_2)\bar{\mathbf{H}}_0\mathbf{R}_{0,z}^{\top}, \tilde{\bf c}_0^{(z)}(w)=\mathbf{R}_{0,z}\bar{\bf c}_0(w), z=0,1,
\end{align*}
and 
\begin{align}\label{eq3-6-7}
	\widetilde{\mathbf{H}}_i=(\mathbf{R}_{0,0}\otimes \mathbf{I}_2)\bar{\mathbf{H}}_{i}\mathbf{R}_{0,0}^{\top}, \tilde{\bf c}_i(w)=\mathbf{R}_{0,0}\bar{\bf c}_i(w), i=1,2,3.
\end{align}
More precisely, from Lemma \ref{lm4} the parity check submatrices in \eqref{eq3-6-6} are presented by 
\begin{align*}
	&\widetilde{\mathbf{H}}_0^{(0)}={\left[\begin{array}{llll}
			L^{(2)}(\lambda_0^2) & 0 \\			
			0 & L^{(2)}(\lambda_0^2) \\	
		\end{array}
		\right]},
	\widetilde{\mathbf{H}}_0^{(1)}={\left[\begin{array}{llll}
			L^{(2)}(\lambda_1^2) & 0 \\			
			0 & L^{(2)}(\lambda_1^2) \\	
		\end{array}
		\right]},
	\widetilde{\mathbf{H}}_1={\left[\begin{array}{llll}
			L^{(2)}(\lambda_2^2) & 0 \\			
			0 & L^{(2)}(\lambda_2^2) \\	
		\end{array}
		\right]},\\
	&\widetilde{\mathbf{H}}_2=\bar{\varPsi}_{2,1}(\varphi_0^{(2)}(\boldsymbol{\lambda}_{4+[2]}^2))={\left[\begin{array}{llll}
			L^{(2)}(\lambda_4^2) & -L^{(2)}(\lambda_5^2) \\			
			0 & L^{(2)}(\lambda_5^2) \\	
		\end{array}
		\right]},
	\widetilde{\mathbf{H}}_3=\bar{\varPsi}_{2,1}(\varphi_1^{(2)}(\boldsymbol{\lambda}_{6+[2]}^2))={\left[\begin{array}{llll}
			L^{(2)}(\lambda_6^2) & 0 \\			
			-L^{(2)}(\lambda_6^2) & L^{(2)}(\lambda_7^2) 
		\end{array}
		\right]},
\end{align*}
and the column vectors in \eqref{eq3-6-6} are given by
\begin{align*}
	\tilde{\bf c}_0^{(0)}(w)={\left[\begin{array}{llll}
			\bar{c}_{0,0}(w) \\			
			\bar{c}_{0,2}(w)  \\	
		\end{array}
		\right]},
	\tilde{\bf c}_0^{(1)}(w)={\left[\begin{array}{llll}
			\bar{c}_{0,1}(w) \\			
			\bar{c}_{0,3}(w)  \\	
		\end{array}
		\right]},
	\tilde{\bf c}_i(w)={\left[\begin{array}{llll}
			\bar{c}_{i,0}(w) \\			
			\bar{c}_{i,2}(w)  \\	
		\end{array}
		\right]}, i=1,2,3, w\in[2].
\end{align*}
From \eqref{ep3} and by Lemma \ref{lm1-1}, $\widetilde{\mathbf{H}}_2,\widetilde{\mathbf{H}}_3$ satisfy $\text{det}(\bar{\boldsymbol{\varPsi}}_{2,1}(\boldsymbol{\varphi}_{\mathcal B}(\boldsymbol{\lambda}_{4+2{\mathcal B}+[2]}^2)))\not=0$ for $\emptyset\not= {\mathcal B}\subseteq[2]$. Additionally, the remaining 3 matrices $\widetilde{\mathbf{H}}_0^{(0)},\widetilde{\mathbf{H}}_0^{(1)},\widetilde{\mathbf{H}}_1$ are all block-diagonal matrices with identical diagonal entries within each matrix. According to Lemma \ref{lm3} $(3)$ with $u=1$, we know that equation \eqref{eq3-6-6} defines a new MDS
array code $(\tilde{\bf c}_0^{(0)}(w),\tilde{\bf c}_0^{(1)}(w),\tilde{\bf c}_1(w),\tilde{\bf c}_2(w),\tilde{\bf c}_3(w))$ with code length 5 and dimension 3 for $w\in[2]$. Helper rack $i\in\{1,2,3\}$ calculates vectors $\tilde{\bf c}_i(w)=\mathbf{R}_{0,0}\bar{\bf c}_i(w), w\in[2]$, and sends it to rack $0$.
For $w\in[2]$, $\tilde{\bf c}_0^{(0)}(w),\tilde{\bf c}_0^{(1)}(w)$
can be recovered from the remaining 3 vectors $\tilde{\bf c}_1(w),\tilde{\bf c}_2(w),\tilde{\bf c}_3(w)$. Then we can recover all the coordinates of $\bar{\bf c}_0(w)$ by combining $\tilde{\bf c}_0^{(0)}(w),\tilde{\bf c}_0^{(1)}(w)$. These operations correspond to Lines $4-13$ in Algorithm 1.

From \eqref{eq3-6-2-1} and \eqref{eq3-6-7}, we know that for $i\in\{1,2,3\}$ and $w\in[2]$,
\begin{align*}
\tilde{\bf c}_i(w)&=\mathbf{R}_{0,0}\bar{\bf c}_i(w)\\
&=\mathbf{R}_{0,0}\varPsi_{2,\lfloor i/2\rfloor}({\rm diag}(\boldsymbol{\lambda}_{2i+[2]}^w))\sum_{g\in[2]}(-1)^{gw}{\bf c}_{2i+g}\\
&=\mathbf{R}_{0,0}\varPsi_{2,\lfloor i/2\rfloor}({\rm diag}(\boldsymbol{\lambda}_{2i+[2]}^w))(\sum_{z\in[2]}\mathbf{R}_{0,z}^{\top}\mathbf{R}_{0,z})\sum_{g\in[2]}(-1)^{gw}{\bf c}_{2i+g}\\
&=\sum_{z\in[2]}\mathbf{R}_{0,0}\varPsi_{2,\lfloor i/2\rfloor}({\rm diag}(\boldsymbol{\lambda}_{2i+[2]}^w))\mathbf{R}_{0,z}^{\top}\sum_{g\in[2]}(-1)^{gw}\mathbf{R}_{0,z}{\bf c}_{2i+g}\\
&=\mathbf{R}_{0,0}\varPsi_{2,\lfloor i/2\rfloor}({\rm diag}(\boldsymbol{\lambda}_{2i+[2]}^w))\mathbf{R}_{0,0}^{\top}\sum_{g\in[2]}(-1)^{gw}\mathbf{R}_{0,0}{\bf c}_{2i+g}.
\end{align*}
It is not hard to see that we can recover $\bar{\bf c}_0(0), \bar{\bf c}_0(1)$ by downloading  $2\bar{s}^{\tilde{n}-1}=4$ symbols from each helper rack $i\in\{1,2,3\}$, while the amount of data accessed on each helper rack is also $4$. From \eqref{eq3-6-2-1}, we have
$\bar{\bf c}_{0}(w)=\varPsi_{2,0}({\rm diag}(\boldsymbol{\lambda}_{[2]}^2))\sum_{g\in[2]}(-1)^{gw}{\bf c}_{g}$ for $w\in[2]$,
and matrix $\varPsi_{2,0}({\rm diag}(\boldsymbol{\lambda}_{[2]}^2))$ is invertible, then
\begin{align}\label{eq3-6-8}
	\sum_{g\in[2]}(-1)^{gw}{\bf c}_{g}=\varPsi_{2,0}^{-1}({\rm diag}(\boldsymbol{\lambda}_{[2]}^w))\bar{\bf c}_{0}(w).
\end{align}
For each $w\in[2]$, setting $\boldsymbol{\Delta}_w=\varPsi_{2,0}^{-1}({\rm diag}(\boldsymbol{\lambda}_{[2]}^w))\bar{\bf c}_{0}(w)$, then \eqref{eq3-6-8} can be rewritten as 
\begin{align*}
	\sum_{g\in[2]}(-1)^{gw}{\bf c}_{g}=\boldsymbol{\Delta}_w.
\end{align*}
For each $j\in[4]$ and each $w\in[2]$, we know that the data
$\boldsymbol{\Delta}_w=\sum_{g\in[2]}(-1)^{gw}{\bf c}_{g}$,
thus
\begin{align}\label{eq3-6-10}
	\sum_{g\in[2]}(-1)^{gw}c_{g,j}=\Delta_{w,j},
\end{align}
where ${\bf c}_{g}=(c_{g,0},c_{g,1},c_{g,2},c_{g,3})^\top$ and $\boldsymbol{\Delta}_{w}=(\Delta_{w,0},\Delta_{w,1},\Delta_{w,2},\Delta_{w,3})^\top$.
Taking $w=0,1$, \eqref{eq3-6-10} can be rewritten as
\begin{equation} \label{eq3-6-11} 
	{\left[\begin{array}{llll}
			1 & 1  \\
			1 & -1 \\
		\end{array}
		\right]} 
		{\left[\begin{array}{l}  c_{0,j} \\
			c_{1,j} \\ 
		\end{array}
		\right]}={\left[\begin{array}{llll}
			\Delta_{0,j} \\  \Delta_{1,j} \\
		\end{array}
		\right]},
\end{equation}
where $j\in[4]$. Since the matrix on the left side of the linear system (\ref{eq3-6-11})
is invertible, then we can recover $c_{0,j},
c_{1,j}, j\in[4]$, i.e., two nodes $c_0, c_1$ in rack $0$.  These operations correspond to Lines $14-16$ in Algorithm 1. It is not hard to see that we recover these $2$  nodes by downloading  $12$ symbols which attain the lower bound \eqref{eq1-1}, while the amount of date accessed on the helper racks is $12$, which achieve the lower bound \eqref{eq1-2}.

\section{The proof of Theorem \ref{th1}}
\label{Sect3}

In this section, we use an alignment technique to extend coupled-layer construction, and carefully
choose the elements $\boldsymbol{\lambda}_{[\bar{n}\bar{s}]}$ to construct our array code. We set $\tilde{n}=\lceil \bar{n}/\bar{s}\rceil$, and also set the sub-packetization $l=\bar{s}^{\tilde{n}}$. 
For our $(n,k,l=\bar{s}^{\lceil\bar{n}/\bar{s}\rceil})$ rack-aware MSR code, we assume $\bar{s}$ divides $\bar{n}$, and thus $\tilde{n}=\bar{n}/\bar{s}$. When $\bar{n}$ is not divisible by $\bar{s}$, the construction can be easily derived by truncating a slightly longer code from the divisible case. The $\bar{n}$ racks of the $(n,k,l)$ array code is divided into $\tilde{n}$ groups, each of which has a size of $\bar{s}$. 
We employ $a\in[\tilde{n}]$ and $b\in[\bar{s}]$ to denote the index of the group and the index of the rack within its group, respectively.

\subsection{Construction}
\label{subsect-scheme-1} 

Let $\theta$ be an
element of $\mathbb{F}_q$ with multiplicative order $u|(q-1)$, and ${\mathcal Q}=\{\xi^{\alpha}:\alpha\in[(q-1)/u]\}$ be a proper subset of $\mathbb{F}_q$.
Assume that $\{\boldsymbol{\lambda}_{a\bar{s}^2+b\bar{s}+j}:a\in[\tilde{n}],b,j\in[\bar{s}]\}=\boldsymbol{\lambda}_{[\bar{n}\bar{s}]}$ are $\bar{n}\bar{s}$ distinct elements in $\mathbb{F}_q$ such that $\{\theta^g\boldsymbol{\lambda}_{a\bar{s}^2+b\bar{s}+j}:a\in[\tilde{n}],b,j\in[\bar{s}],g\in[u]\}=\theta^{[u]}\boldsymbol{\lambda}_{[\bar{n}\bar{s}]}$ are $n\bar{s}$ distinct elements in $\mathbb{F}_q$. Consider an array code $\mathcal{C}$ defined by the following parity check equations over $\mathbb{F}_q$:
\begin{equation}\label{eq3-1}
	\mathcal{C}=\{({\bf c}_0,{\bf c}_1,\ldots,{\bf c}_{n-1}):\sum_{a\in[\tilde{n}]}\sum_{b\in[\bar{s}]}\sum_{g\in[u]}\mathbf{H}_{(a\bar{s}+b)u+g}{\bf c}_{(a\bar{s}+b)u+g}=0\},
\end{equation}
\noindent where 
\begin{align}\label{eq3-2}
\mathbf{H}_{(a\bar{s}+b)u+g}=\varPsi_{\tilde{n},a}(\varphi_b^{(r)}(\theta^g\boldsymbol{\lambda}_{a\bar{s}^2+b\bar{s}+[\bar{s}]})), a\in[\tilde{n}], b\in[\bar{s}], g\in[u].
\end{align}

For $a\in[\tilde{n}]$, a non-empty subset ${\mathcal B}=\{b_0,b_1,\ldots,b_{t-1}\}\subseteq [\bar{s}]$ with $b_0<b_1<\ldots<b_{t-1}$, we define
\begin{align*}
	\mathbf{x}_{{\mathcal B}\bar{s}+[\bar{s}]}=(x_{b\bar{s}+j}:b\in{\mathcal B},j\in[\bar{s}]).
\end{align*}
For a set ${\mathcal G}_{\mathcal B}=\{{\mathcal G}_{b_0},{\mathcal G}_{b_1},\ldots,{\mathcal G}_{b_{t-1}}\}$ satisfying $|{\mathcal G}_{b_j}|=m_j\geq1$ for $j\in [t]$, and $|{\mathcal G}_{\mathcal B}|=m_0+m_1+\ldots+m_{t-1}=\delta$, from \eqref{eq4} and \eqref{eq5} we have the $m\bar{s}\times \delta\bar{s}$ matrix
\begin{align}\label{eq3-3}
	\nonumber	\boldsymbol{\varphi}_{\mathcal B}^{(m)}(\theta^{{\mathcal G}_{\mathcal B}}\odot \boldsymbol{\lambda}_{a\bar{s}^2+{\mathcal B}\bar{s}+[\bar{s}]})	&=[\varphi_{b_0}^{(m)}(\theta^{g_{b_0,0}}\boldsymbol{\lambda}_{a\bar{s}^2+b_0\bar{s}+[\bar{s}]})\ldots \varphi_{b_0}^{(m)}(\theta^{g_{b_0,m_0-1}}\boldsymbol{\lambda}_{a\bar{s}^2+b_0\bar{s}+[\bar{s}]}) \ldots \varphi_{b_{t-1}}^{(m)}(\theta^{g_{b_{t-1},0}}\boldsymbol{\lambda}_{a\bar{s}^2+b_{t-1}\bar{s}+[\bar{s}]})\\ &\ldots \varphi_{b_{t-1}}^{(m)}(\theta^{g_{b_{t-1},m_{t-1}-1}}\boldsymbol{\lambda}_{a\bar{s}^2+b_{t-1}\bar{s}+[\bar{s}]})],
\end{align}
and the $ml\times \delta l$ matrix
\begin{align}\label{eq3-4}
	\nonumber	\boldsymbol{\varPsi}_{\tilde{n},a}(\boldsymbol{\varphi}_{\mathcal B}^{(m)}(\theta^{{\mathcal G}_{\mathcal B}}\odot \boldsymbol{\lambda}_{a\bar{s}^2+{\mathcal B}\bar{s}+[\bar{s}]}))&	=[\varPsi_{\tilde{n},a}(\varphi_{b_0}^{(m)}(\theta^{g_{b_0,0}}\boldsymbol{\lambda}_{a\bar{s}^2+b_0\bar{s}+[\bar{s}]}))\ldots \varPsi_{\tilde{n},a}(\varphi_{b_0}^{(m)}(\theta^{g_{b_0,m_0-1}}\boldsymbol{\lambda}_{a\bar{s}^2+b_0\bar{s}+[\bar{s}]}))\ldots\\ &	
	\varPsi_{\tilde{n},a}(\varphi_{b_{t-1}}^{(m)}(\theta^{g_{b_{t-1},0}}\boldsymbol{\lambda}_{a\bar{s}^2+b_{t-1}\bar{s}+[\bar{s}]})) \ldots \varPsi_{\tilde{n},a}(\varphi_{b_{t-1}}^{(m)}(\theta^{g_{b_{t-1},m_{t-1}-1}}\boldsymbol{\lambda}_{a\bar{s}^2+b_{t-1}\bar{s}+[\bar{s}]}))],
\end{align}
where $m\geq \delta$.

Similarly, we briefly denote $\boldsymbol{\varPsi}_{\tilde{n},a}(\boldsymbol{\varphi}_{\mathcal B}^{(m)}(\theta^{{\mathcal G}_{\mathcal B}}\odot \boldsymbol{\lambda}_{a\bar{s}^2+{\mathcal B}\bar{s}+[\bar{s}]}))$ as $\mathbf{M}_{a,{\mathcal B}}^{(m)}$, and $\boldsymbol{\varphi}_{\mathcal B}^{(m)}(\theta^{{\mathcal G}_{\mathcal B}}\odot \boldsymbol{\lambda}_{a\bar{s}^2+{\mathcal B}\bar{s}+[\bar{s}]})$ as $\boldsymbol{\phi}_{a,{\mathcal B}}^{(m)}$.
To guarantee the MDS property, we not only require the $\bar{n}\bar{s}$
distinct elements $\boldsymbol{\lambda}_{[\bar{n}\bar{s}]}$ in $\mathbb{F}_q$ such that $\theta^{[u]}\boldsymbol{\lambda}_{[\bar{n}\bar{s}]}$ are $n\bar{s}$ distinct elements in $\mathbb{F}_q$, but also need that
\begin{align}\label{eq3-5}
\text{det}(\boldsymbol{\phi}_{a,{\mathcal B}}^{(\delta)})\not=0, a\in[\tilde{n}], \emptyset\not={\mathcal B}\subseteq [\bar{s}], \emptyset\not={\mathcal G}_{b_j}\subseteq[u], b_j\in\mathcal{B}.
\end{align}
i.e. all the square matrices are invertible. From Lemma \ref{lm2}, we know that these conditions \eqref{eq3-5} are equivalent to
\begin{align}\label{eq3-6}
	\text{det}(\mathbf{M}_{a,{\mathcal B}}^{(\delta)})\not=0, a\in[\tilde{n}], \emptyset\not={\mathcal B}\subseteq [\bar{s}], \emptyset\not={\mathcal G}_{b_j}\subseteq[u], b_j\in\mathcal{B}.
\end{align}
We term the conditions presented in \eqref{eq3-5} and \eqref{eq3-6} as local  constraints for  code $\mathcal{C}$.

The existence of $\bar{n}\bar{s}$ distinct elements $\boldsymbol{\lambda}_{[\bar{n}\bar{s}]}$ that satisfy the conditions in \eqref{eq3-5} and \eqref{eq3-6} within a field of size $O_{\bar{s}}(n)$ is guaranteed by the subsequent two lemmas. It is worth noting that Lemma \ref{lm6} guarantee the existence of elements satisfying the local constraints over a linear field. In this case, the elements must be sought through a linear-complexity algorithm. Conversely, Lemma \ref{lm7} offers an explicit construction, and the field size is comparatively larger than that of Lemma \ref{lm6} when $\bar{n}$ is not large enough in relation to $\bar{s}$.

\begin{lemma}\label{lm6}
Let \begin{align}\label{eq3-6-1-1}
	\Omega(\bar{s},u)= \sum_{t=1}^{\bar{s}}\sum_{\delta=t}^{ut}{\bar{s}-1\choose t-1}{ut\choose \delta}{\delta-1\choose t-1}(\delta-t+1)(\delta-1).
\end{align}
If $q\geq n\bar{s}+u\Omega(\bar{s},u)$,
then we can select $\bar{n}\bar{s}$ distinct
elements $\boldsymbol{\lambda}_{[\bar{n}\bar{s}]}$ from $\mathbb{F}_q$ such that $\theta^{[u]}\boldsymbol{\lambda}_{[\bar{n}\bar{s}]}$ are $n\bar{s}$ distinct elements, and satisfy  \eqref{eq3-5}  and \eqref{eq3-6}. The time complexity associated with the selection of these elements is $O_{\bar{s}}(n)$. To be specific, for each $a\in[\tilde{n}]$, we can select the $\bar{s}^2$ elements $\boldsymbol{\lambda}_{a\bar{s}^2+[\bar{s}^2]}$ in any subset ${\mathcal S}\subseteq {\mathcal Q}$ of size $|{\mathcal S}|\geq \bar{s}^2+\Omega(\bar{s},u)$ such that $\boldsymbol{\phi}_{a,{\mathcal B}}^{(\delta)}$ and $\mathbf{M}_{a,{\mathcal B}}^{(\delta)}$ are invertible for all ${\mathcal B}\subseteq[\bar{s}]$.
\end{lemma}
\begin{lemma}\label{lm7}
Let $\mathbb{F}_q=\mathbb{F}_p[\xi]$, where $p$ is a prime. If $q\geq max\{n\bar{s}, p^{\bar{s}^5u^2}\}$, 
then the $\bar{n}\bar{s}$ elements $\lambda_i=\xi^i$, $i\in[\bar{n}\bar{s}]$ satisfy the conditions that $\theta^{[u]}\boldsymbol{\lambda}_{[\bar{n}\bar{s}]}$ are $n\bar{s}$ distinct elements, and that the local constraints \eqref{eq3-5} and \eqref{eq3-6} hold.	
\end{lemma}

The proofs for these two lemmas are included in Appendix \ref{app-A}.

Since the host rack downloads partial encoded symbols from each rescue rack, to utilize this data for creating a repair scheme, we transform the code $\mathcal{C}$ into the intermediate code $\bar{\mathcal{C}}(w)$ while preserving the MDS property for each $w\in[u]$. Let $t=\bar{t}u+w\in[r]$ be an integer, where $w\in[u]$. Since $r=(\bar{r}-1)u+u-v$, for each integer $w\in[u-v]$, we have $\bar{t}\in[\bar{r}]$; otherwise, $\bar{t}\in[\bar{r}-1]$. Let $\mathbf{e}_i'$ denote the  $i$-th row vector of ${\bf I}_r$. For any integer $w\in[u-v]$, we define  
\[
\mathbf{Q}_w=[(\mathbf{e}_w')^{\top} \ (\mathbf{e}_{u+w}'  )^{\top} \ \ldots \ (\mathbf{e}_{(\bar{r}-1)u+w}'  )^{\top}]^{\top},
\]
and for other values of $w$, we define  
\[
\mathbf{Q}_w = \left[ (\mathbf{e}_w')^\top \ (\mathbf{e}_{u+w}')^\top \ \ldots \ (\mathbf{e}_{(\bar{r}-2)u + w}')^\top \right]^\top.
\]

Note that ${\bf I}_{\bar{s}^a}\otimes({\bf I}_{\bar{s}}\otimes \mathbf{Q}_w)={\bf I}_{\bar{s}^a}\boxtimes({\bf I}_{\bar{s}}\otimes \mathbf{Q}_w)$ where  $w\in[u]$, it is easy to see that 
\begin{align*}
	{\bf I}_{l}\otimes Q_w=&{\bf I}_{\bar{s}^{\widetilde{n}-a-1}}\otimes({\bf I}_{\bar{s}^a}\otimes({\bf I}_{\bar{s}}\otimes Q_w))\\
	=&{\bf I}_{\bar{s}^{\widetilde{n}-a-1}}\otimes({\bf I}_{\bar{s}^a}\boxtimes({\bf I}_{\bar{s}}\otimes Q_w))\\
	=&\varPsi_{\tilde{n},a}({\bf I}_{\bar{s}}\otimes \mathbf{Q}_w).
\end{align*}
For any rack $a\bar{s}+b$, where $a\in[\tilde{n}]$ and $b\in[\bar{s}]$, we know that $\mathbf{H}_{(a\bar{s}+b)u+g}=\varPsi_{\tilde{n},a}(\varphi_b^{(r)}(\theta^g\boldsymbol{\lambda}_{(a\bar{s}+b)\bar{s}+[\bar{s}]}))$. From \eqref{eq3-1},  it follows that
\begin{align*}%\label{eq3-11}
	\sum_{a\in[\tilde{n}]}\sum_{b\in[\bar{s}]}\sum\limits_{g\in[u]}({\bf I}_{l}\otimes Q_w)\varPsi_{\tilde{n},a}(\varphi_b^{(r)}(\theta^g\boldsymbol{\lambda}_{(a\bar{s}+b)\bar{s}+[\bar{s}]})){\bf c}_{(a\bar{s}+b)u+g}=0,
\end{align*}
which implies 
\begin{align}\label{eq3-12}
		\sum_{a\in[\tilde{n}]}\sum_{b\in[\bar{s}]}\sum\limits_{g\in[u]}\varPsi_{\tilde{n},a}({\bf I}_{\bar{s}}\otimes \mathbf{Q}_w)\varPsi_{\tilde{n},a}(\varphi_b^{(r)}(\theta^g\boldsymbol{\lambda}_{(a\bar{s}+b)\bar{s}+[\bar{s}]})){\bf c}_{(a\bar{s}+b)u+g}=0.
\end{align}
According to Lemma \ref{lm1-1}, we have
\begin{align}
	\nonumber	&\varPsi_{\tilde{n},a}({\bf I}_{\bar{s}}\otimes \mathbf{Q}_w)\varPsi_{\tilde{n},a}(\varphi_b^{(r)}(\theta^g\boldsymbol{\lambda}_{(a\bar{s}+b)\bar{s}+[\bar{s}]}))\\
	\nonumber	=&	\varPsi_{\tilde{n},a}(({\bf I}_{\bar{s}}\otimes \mathbf{Q}_w)\varphi_b^{(r)}(\theta^g\boldsymbol{\lambda}_{(a\bar{s}+b)\bar{s}+[\bar{s}]}))\\
 	\nonumber=&\begin{cases}
		\varPsi_{\tilde{n},a}(\theta^{gw}\varphi_b^{(\bar{r})}(\boldsymbol{\lambda}_{(a\bar{s}+b)\bar{s}+[\bar{s}]}^u){\rm diag}(\boldsymbol{\lambda}_{(a\bar{s}+b)\bar{s}+[\bar{s}]}^w))	 & \text{if } w\in[u-v], \\	
		\varPsi_{\tilde{n},a}(\theta^{gw}\varphi_b^{(\bar{r}-1)}(\boldsymbol{\lambda}_{(a\bar{s}+b)\bar{s}+[\bar{s}]}^u){\rm diag}(\boldsymbol{\lambda}_{(a\bar{s}+b)\bar{s}+[\bar{s}]}^w)) & \text{if } w\in[u-v,u).
	\end{cases}\\\label{eq3-13-1}
	=&\begin{cases}
		\theta^{gw}\varPsi_{\tilde{n},a}(\varphi_b^{(\bar{r})}(\boldsymbol{\lambda}_{(a\bar{s}+b)\bar{s}+[\bar{s}]}^u))\varPsi_{\tilde{n},a}({\rm diag}(\boldsymbol{\lambda}_{(a\bar{s}+b)\bar{s}+[\bar{s}]}^w))	 & \text{if } w\in[u-v], \\	
		\theta^{gw}\varPsi_{\tilde{n},a}(\varphi_b^{(\bar{r}-1)}(\boldsymbol{\lambda}_{(a\bar{s}+b)\bar{s}+[\bar{s}]}^u))\varPsi_{\tilde{n},a}({\rm diag}(\boldsymbol{\lambda}_{(a\bar{s}+b)\bar{s}+[\bar{s}]}^w)) & \text{if } w\in[u-v,u).
	\end{cases}
\end{align} 
Combining \eqref{eq3-12} and \eqref{eq3-13-1}, we obtain the following two equations:
\begin{align}\label{eq3-14}
	\sum_{a\in[\tilde{n}]}\sum_{b\in[\bar{s}]}\varPsi_{\tilde{n},a}(\varphi_b^{(\bar{r})}(\boldsymbol{\lambda}_{(a\bar{s}+b)\bar{s}+[\bar{s}]}^u))\varPsi_{\tilde{n},a}({\rm diag}(\boldsymbol{\lambda}_{(a\bar{s}+b)\bar{s}+[\bar{s}]}^w))\sum_{g\in[u]}	\theta^{gw}{\bf c}_{(a\bar{s}+b)u+g}=0 \ {\rm if} \ w\in[u-v];\\
	\label{eq3-14-1}	\sum_{a\in[\tilde{n}]}\sum_{b\in[\bar{s}]}\varPsi_{\tilde{n},a}(\varphi_b^{(\bar{r}-1)}(\boldsymbol{\lambda}_{(a\bar{s}+b)\bar{s}+[\bar{s}]}^u))\varPsi_{\tilde{n},a}({\rm diag}(\boldsymbol{\lambda}_{(a\bar{s}+b)\bar{s}+[\bar{s}]}^w))\sum_{g\in[u]}\theta^{gw}{\bf c}_{(a\bar{s}+b)u+g}=0  \ {\rm if} \   w\in[u-v,u),
\end{align}
then  both \eqref{eq3-14} and \eqref{eq3-14-1} can be expressed as 
\begin{align}\label{eq3-15}
	\sum_{a\in[\tilde{n}]}\sum_{b\in[\bar{s}]}\bar{\mathbf{H}}_{a\bar{s}+b}\bar{\bf c}_{a\bar{s}+b}(w)=0,
\end{align}
where 
\begin{align}\label{eq3-16}
	\bar{\mathbf{H}}_{a\bar{s}+b}=\begin{cases}
		\varPsi_{\tilde{n},a}(\varphi_b^{(\bar{r})}(\boldsymbol{\lambda}_{(a\bar{s}+b)\bar{s}+[\bar{s}]}^u))  \ {\rm if} \  w\in[u-v],\\
		\varPsi_{\tilde{n},a}(\varphi_b^{(\bar{r}-1)}(\boldsymbol{\lambda}_{(a\bar{s}+b)\bar{s}+[\bar{s}]}^u))  \ {\rm if} \  w\in[u-v,u),
	\end{cases}
\end{align}
and 
\begin{align}\label{eq3-16-1}
	\bar{\bf c}_{a\bar{s}+b}(w)=\varPsi_{\tilde{n},a}({\rm diag}(\boldsymbol{\lambda}_{(a\bar{s}+b)\bar{s}+[\bar{s}]}^w))\sum\limits_{g\in[u]}\theta^{gw}{\bf c}_{(a\bar{s}+b)u+g}.
\end{align}
Clearly, \eqref{eq3-15} defines an intermediate code, denoted by $\bar{\mathcal{C}}(w)$, i.e.,
\begin{align}\label{eq3-17}
	\bar{\mathcal{C}}(w)=\{(\bar{\bf c}_0(w),\bar{\bf c}_1(w),\ldots,\bar{\bf c}_{\bar{n}-1}(w)):\sum_{a\in[\tilde{n}]}\sum_{b\in[\bar{s}]}\bar{\mathbf{H}}_{a\bar{s}+b}\bar{\bf c}_{a\bar{s}+b}(w)=0\},
\end{align}
where  $\bar{\mathbf{H}}_{a\bar{s}+b}$ is defined in \eqref{eq3-16}, while $\bar{\bf c}_{a\bar{s}+b}(w)$ is present in \eqref{eq3-16-1}. 

For $a\in[\tilde{n}]$, a non-empty subset ${\mathcal B}=\{b_0,b_1,\ldots,b_{t-1}\}\subseteq [\bar{s}]$ with $b_0<b_1<\ldots<b_{t-1}$, from \eqref{eq501} and \eqref{eq502} we have the $\bar{s}m\times \bar{s}t$ matrix
\begin{align}\label{eq3-17-1}
	\boldsymbol{\varphi}_{\mathcal B}^{(m)}(\boldsymbol{\lambda}_{a\bar{s}^2+{\mathcal B}\bar{s}+[\bar{s}]}^u)	=[\varphi_{b_0}^{(m)}(\boldsymbol{\lambda}_{a\bar{s}^2+b_0\bar{s}+[\bar{s}]}^u) \ \varphi_{b_1}^{(m)}(\boldsymbol{\lambda}_{a\bar{s}^2+b_1\bar{s}+[\bar{s}]}^u) \ldots \varphi_{b_{t-1}}^{(m)}(\boldsymbol{\lambda}_{a\bar{s}^2+b_{t-1}\bar{s}+[\bar{s}]}^u)],
\end{align}
and the $ml\times tl$ matrix
\begin{align}\label{eq3-17-2}
	\boldsymbol{\varPsi}_{\tilde{n},a}(\boldsymbol{\varphi}_{\mathcal B}^{(m)}(\boldsymbol{\lambda}_{a\bar{s}^2+{\mathcal B}\bar{s}+[\bar{s}]}^u))	=[\varPsi_{\tilde{n},a}(\varphi_{b_0}^{(m)}(\boldsymbol{\lambda}_{a\bar{s}^2+b_0\bar{s}+[\bar{s}]}^u)) \ \varPsi_{\tilde{n},a}(\varphi_{b_1}^{(m)}(\boldsymbol{\lambda}_{a\bar{s}^2+b_1\bar{s}+[\bar{s}]}^u))\ldots \varPsi_{\tilde{n},a}(\varphi_{b_{t-1}}^{(m)}(\boldsymbol{\lambda}_{a\bar{s}^2+b_{t-1}\bar{s}+[\bar{s}]}^u))],
\end{align}
where $m\geq t$. To short the notations, we briefly denote $\boldsymbol{\varPsi}_{\tilde{n},a}(\boldsymbol{\varphi}_{\mathcal B}^{(m)}(\boldsymbol{\lambda}_{a\bar{s}^2+{\mathcal B}\bar{s}+[\bar{s}]}^u))$ as $\mathbf{\bar{M}}_{a,{\mathcal B}}^{(m)}$, and $\boldsymbol{\varphi}_{\mathcal B}^{(m)}( \boldsymbol{\lambda}^u_{a\bar{s}^2+{\mathcal B}\bar{s}+[\bar{s}]})$ as $\boldsymbol{\bar{\phi}}_{a,{\mathcal B}}^{(m)}$.

To guarantee the MDS property of the code $\bar{\mathcal{C}}(w)$, we require not only that the $\bar{n}\bar{s}$ distinct elements $\boldsymbol{\lambda}_{[\bar{n}\bar{s}]}$ in $\mathbb{F}_q$ satisfy $\boldsymbol{\lambda}_{[\bar{n}\bar{s}]}^u$ being $\bar{n}\bar{s}$ distinct elements in $\mathbb{F}_q$, but also that 
\begin{align}\label{eq3-18}
	\text{det}(\boldsymbol{\bar{\phi}}_{a,{\mathcal B}}^{(\delta)})\not=0, a\in[\tilde{n}], \emptyset\not={\mathcal B}\subseteq [\bar{s}].
\end{align}
From Lemma \ref{lm2}, we know that these conditions \eqref{eq3-18} are equivalent to
\begin{align}\label{eq3-19}
	\text{det}(\mathbf{\bar{M}}_{a,{\mathcal B}}^{(\delta)})\not=0, a\in[\tilde{n}], \emptyset\not={\mathcal B}\subseteq [\bar{s}].
\end{align}
We term the conditions presented in \eqref{eq3-5} and \eqref{eq3-6} as local  constraints for  the code $\bar{\mathcal{C}}(w)$.

The existence of $\bar{n}\bar{s}$ distinct elements $\boldsymbol{\lambda}_{[\bar{n}\bar{s}]}$ that satisfy the conditions in \eqref{eq3-18} and \eqref{eq3-19} within a field of size $O_{\bar{s}}(n)$ is guaranteed by the subsequent two lemmas. Lemma \ref{lm10} guarantee the existence of elements satisfying the local constraints over a linear field, whereas Lemma \ref{lm11} offers an explicit construction.
\begin{lemma}\label{lm10}
	Assuming $q\geq n\bar{s}+u(\bar{s}-1)2^{\bar{s}-2} $, we can select $\bar{n}\bar{s}$ distinct elements $ \boldsymbol{\lambda}_{[\bar{n}\bar{s}]}$ from $\mathbb{F}_q$ such that $\boldsymbol{\lambda}_{[\bar{n}\bar{s}]}^u $are $\bar{n}\bar{s}$ distinct elements, and conditions \eqref{eq3-18} and \eqref{eq3-19} hold. And
	the time complexity to choose these elements is $O_{\bar{s}}(n)$. To be specific, for each $a\in[\tilde{n}]$, we can select the $\bar{s}^2$ elements $\boldsymbol{\lambda}_{a\bar{s}^2+[\bar{s}^2]}$ in any subset ${\mathcal S}\subseteq \mathcal Q$ of size $|{\mathcal S}|\geq \bar{s}^2+(\bar{s}-1)2^{\bar{s}-2}$ such that $\boldsymbol{\bar{\phi}}_{a,{\mathcal B}}^{(\delta)}$ and $\mathbf{\bar{M}}_{a,{\mathcal B}}^{(\delta)}$ are invertible for all ${\mathcal B}\subseteq[\bar{s}]$.
\end{lemma}
The proof for this lemma is provided in Appendix \ref{app-B}.
\begin{remark}
	It is not hard to see that if $q\geq n\bar{s}+u(\Omega(\bar{s},u)+(\bar{s}-1)2^{\bar{s}-2})$,  we can select $\bar{n}\bar{s}$ distinct elements $ \boldsymbol{\lambda}_{[\bar{n}\bar{s}]}$ from $\mathbb{F}_q$ such that Lemmas \ref{lm6} and  \ref{lm10} hold simultaneously.
\end{remark}
\begin{lemma}\label{lm11}
	Let $\mathbb{F}_q = \mathbb{F}_p[\xi]$, where $p $ is a prime. If $ q \geq \max\{n\bar{s}, p^{\bar{s}^5}\} $, then the $\bar{n}\bar{s}$ elements $\lambda_i = \xi^i$, $i \in [\bar{n}\bar{s}] $, satisfy the conditions that $\boldsymbol{\lambda}_{[\bar{n}\bar{s}]}^u$ are $\bar{n}\bar{s} $ distinct elements and that the local constraints \eqref{eq3-18} and \eqref{eq3-19} hold.
\end{lemma}
We omit the proof of Lemma \ref{lm11}, as it is similar to the proof of Lemma 7 in \cite{LWHY}.
\begin{remark}
	It is not hard to see that if $q\geq max\{n\bar{s}, p^{\bar{s}^5u^2}\}$, then Lemma \ref{lm7} and Lemma \ref{lm11} hold simultaneously.
\end{remark}

\subsection{MDS property}
\label{subsect-scheme-3} 

An $(n, k, l)$ array code $\mathcal C$ defined in \eqref{eq3-1} is an MDS array code if
any $r=n-k$ failed nodes in each codeword of $\mathcal{C}$ can be reconstructed from the remaining $k$ nodes. Define the set of failed node indices as ${\mathcal F}=\{i_1, i_2,\ldots, i_r\}$. With this, the parity check equation in \eqref{eq3-1} can be expressed as
\begin{align*}
	\sum_{i\in{\mathcal F}}\mathbf{H}_i{\bf c}_i=-\sum_{i\in[n]\setminus{\mathcal F}}\mathbf{H}_i{\bf c}_i.
\end{align*}
The recovery of failed nodes $\{{\bf c}_i: i \in {\mathcal F}\}$ requires that the above equation admits a unique solution. As a result, the MDS property of the code $\mathcal{C}$ is directly derived from the following arguments:
\begin{align} \label{eq3-9}
	[\mathbf{H}_{i_1} \ \mathbf{H}_{i_2}\ \ldots \ \mathbf{H}_{i_r}] \ \text{ is \ invertible \ for \ all} \ {\mathcal F}=\{i_1, i_2,\ldots, i_r\}\subset [n].
\end{align}
By recursively applying Lemma \ref{lm3}, we transform the matrices in the global constraints \eqref{eq3-9} into block upper triangular form, leading to our main result. The MDS property of the rack-aware MSR codes introduced in the preceding subsection is then directly established by the following result.
\begin{lemma}\label{lm8}
Let  $a_0,a_1,\ldots,a_{z-1}$ be any  $z$ distinct integers in $[\tilde{n}]$, and  ${\mathcal B}_0,{\mathcal B}_1,\ldots,{\mathcal B}_{z-1}$ be any $z$ non-empty subsets of 
$[\bar{s}]$ such that $|{\mathcal G}_{{\mathcal B}_0}|+|{\mathcal G}_{{\mathcal B}_1}|+\ldots+|{\mathcal G}_{{\mathcal B}_{z-1}}|=m\leq r$. Then, we have
\begin{align}\label{eq3-7}
	{\rm det}([\mathbf{M}_{a_0,\mathcal{B}_0}^{(m)} \ \mathbf{M}_{a_1,\mathcal{B}_1}^{(m)}	\ldots \mathbf{M}_{a_{z-1},\mathcal{B}_{z-1}}^{(m)}])\not=0,
\end{align}
where $\mathbf{M}_{a_i,\mathcal{B}_i}^{(m)}=\boldsymbol{\varPsi}_{\tilde{n},a_i}(\boldsymbol{\varphi}_{{\mathcal B}_i}^{(m)}(\theta^{{\mathcal G}_{{\mathcal B}_i}}\odot \boldsymbol{\lambda}_{a_i\bar{s}^2+{\mathcal B}_i\bar{s}+[\bar{s}]}))$ for $i\in[z]$.
\end{lemma}
\begin{proof}
The proof proceeds by induction on the positive integer  $z$. When
$z=1$, since all the $\lambda_i$s meet local constraints \eqref{eq3-6}, we
obtain $\text{det}(\mathbf{M}_0^{(m)})\not=0$.

Under the inductive hypothesis, we take the conclusion to be true for any positive integer $z$. For the case of $z+1$, we have $|{\mathcal G}_{{\mathcal B}_0}|+|{\mathcal G}_{{\mathcal B}_1}|+\ldots+|{\mathcal G}_{{\mathcal B}_{z}}|=m$. Write $\delta_0=|{\mathcal G}_{{\mathcal B}_0}|$. Since $\mathbf{M}_{a_0,\mathcal{B}_0}^{(\delta_0)}$ is invertible, all the $n\bar{s}$ elements $\theta^{[u]}\boldsymbol{\lambda}_{[\bar{n}\bar{s}]}$ are also distinct, by Lemma \ref{lm3}, there exists $ml\times ml$ matrix $\mathbf{V}$ such that 
\begin{align} \label{eq3-8}
\nonumber \text{det}(\mathbf{V}[\mathbf{M}_{a_0,\mathcal{B}_0}^{(m)} \ \mathbf{M}_{a_1,\mathcal{B}_1}^{(m)}	\ldots \mathbf{M}_{a_z,\mathcal{B}_z}^{(m)}])&=\text{det}{\left(\left[\begin{array}{lllll}
		\mathbf{M}_{a_0,\mathcal{B}_0}^{(\delta_0)}	 & \mathbf{M}_{a_1,\mathcal{B}_1}^{(\delta_0)} & \ldots & \mathbf{M}_{a_z,\mathcal{B}_z}^{(\delta_0)}\\
        0	 & 	\widehat{\mathbf{M}}_{a_1,\mathcal{B}_1}^{(m-\delta_0)} & \ldots & 	\widehat{\mathbf{M}}_{a_z,\mathcal{B}_z}^{(m-\delta_0)}\\
	\end{array}
	\right]\right)}\\
	&=\text{det}\left (\mathbf{M}_{a_0,\mathcal{B}_0}^{(\delta_0)}\right)\text{det}\left ([\widehat{\mathbf{M}}_{a_1,\mathcal{B}_1}^{(m-\delta_0)}\ldots\widehat{\mathbf{M}}_{a_z,\mathcal{B}_z}^{(m-\delta_0)}]\right),
\end{align}
where $\widehat{\mathbf{M}}_{a_i,\mathcal{B}_i}^{(m-\delta_0)}=\widehat{\boldsymbol{\varPsi}}_{\tilde{n},a_i}(\boldsymbol{\varphi}_{{\mathcal B}_i}^{(m-\delta_0)}(\theta^{{\mathcal G}_{{\mathcal B}_i}}\odot \boldsymbol{\lambda}_{a_i\bar{s}^2+{\mathcal B}_i\bar{s}+[\bar{s}]}))$ for $i\in[1,z]$, and each $\widehat{\mathbf{M}}_{a_i,\mathcal{B}_i}^{(m-\delta_0)}$ is column equivalent to $\mathbf{M}_{a_i,\mathcal{B}_i}^{(m-\delta_0)}$. Note that $|{\mathcal G}_{{\mathcal B}_1}|+\ldots+|{\mathcal G}_{{\mathcal B}_{z}}|=m-\delta_0$. By the induction hypothesis, the matrix $[\mathbf{M}_{a_1,\mathcal{B}_1}^{(m-\delta_0)} \ \mathbf{M}_{a_2,\mathcal{B}_2}^{(m-\delta_0)}\ldots \mathbf{M}_{a_z,\mathcal{B}_z}^{(m-\delta_0)}]$ is invertible. It means that matrix $[\widehat{\mathbf{M}}_{a_1,\mathcal{B}_1}^{(m-\delta_0)} \ \widehat{\mathbf{M}}_{a_2,\mathcal{B}_2}^{(m-\delta_0)}\ldots \widehat{\mathbf{M}}_{a_z,\mathcal{B}_z}^{(m-\delta_0)}]$ is also invertible. By \eqref{eq3-8}, we have 
\begin{align*}
	\text{det}([\mathbf{M}_{a_0,\mathcal{B}_0}^{(m)} \ \mathbf{M}_{a_1,\mathcal{B}_1}^{(m)}	\ldots \mathbf{M}_{a_z,\mathcal{B}_z}^{(m)}])\not=0.
\end{align*}
\end{proof}
Observe that setting $m=r$ in Lemma \ref{lm8} leads directly to the global constraints \eqref{eq3-9}.

Similarly, an $(\bar{n}, \bar{k}, l)$ array code ${ \bar{\mathcal C}}(w)$, where $w\in[u-v]$, defined in \eqref{eq3-17} is an MDS array code if it allows the reconstruction of any $\bar{r}=\bar{n}-\bar{k}$ failed nodes using the remaining $\bar{k}$ nodes. Define the set of failed node indices as ${\mathcal F}=\{i_1, i_2,\ldots, i_{\bar{r}}\}$. With this, the parity check equation in \eqref{eq3-17} can be expressed as
\begin{align*}
	\sum_{i\in{\mathcal F}}\bar{\mathbf{H}}_i\bar{\bf c}_i(w)=-\sum_{i\in[\bar{n}]\setminus{\mathcal F}}\bar{\mathbf{H}}_i\bar{\bf c}_i(w).
\end{align*}
The recovery of failed nodes $\{\bar{\mathbf c}_i(w): i \in {\mathcal F}\}$ requires that the above equation admits a unique solution. As a result, the MDS property of the code $ \bar{\mathcal C}(w)$ for $w\in[u-v]$ is directly derived from the following arguments:
\begin{align} \label{eq3-10}
	[\bar{\mathbf{H}}_{i_1} \ \bar{\mathbf{H}}_{i_2}\ \ldots \ \bar{\mathbf{H}}_{i_{\bar{r}}}] \ \text{ is \ invertible \ for \ all} \ {\mathcal F}=\{i_1, i_2,\ldots, i_{\bar{r}}\}\subseteq [\bar{n}].
\end{align}
In the same way, the MDS property of the code ${ \bar{\mathcal C}}(w)$ for $w\in[u-v,u-1]$ is directly obtained from the following arguments:
\begin{align} \label{eq3-10-1}
	[\bar{\mathbf{H}}_{i_1} \ \bar{\mathbf{H}}_{i_2}\ \ldots \ \bar{\mathbf{H}}_{i_{\bar{r}-1}}] \ \text{ is \ invertible \ for \ all} \ {\mathcal F}'=\{i_1, i_2,\ldots, i_{\bar{r}-1}\}\subseteq [\bar{n}].
\end{align}

Using Lemmas \ref{lm3}, \ref{lm10} and \ref{lm11}, we can easily derive the following conclusion in a manner similar to the proof of Lemma \ref{lm8}.
\begin{lemma}\label{lm12}
	Let  $a_0,a_1,\ldots,a_{z-1}$ be any  $z$ distinct integers in $[\tilde{n}]$, and let  ${\mathcal B}_0,{\mathcal B}_1,\ldots,{\mathcal B}_{z-1}$ be any $z$ non-empty subsets of 
	$[\bar{s}]$ such that $|{\mathcal B}_0|+|{\mathcal B}_1|+\ldots+|{\mathcal B}_{z-1}|=m\leq \bar{r}$. Then, we obtain
	\begin{align}\label{eq3-20}
		{\rm det}([\mathbf{\bar{M}}_{a_0,\mathcal{ B}_0}^{(m)} \ \mathbf{\bar{M}}_{a_1,\mathcal{ B}_1}^{(m)}	\ldots \mathbf{\bar{M}}_{a_{z-1},\mathcal{ B}_{z-1}}^{(m)}])\not=0,
	\end{align}
	where $\mathbf{\bar{M}}_{a_i,\mathcal{ B}_i}^{(m)}=\boldsymbol{\varPsi}_{\tilde{n},a_i}(\boldsymbol{\varphi}_{{\mathcal B}_i}^{(m)}(\boldsymbol{\lambda}_{a_i\bar{s}^2+{\mathcal B}_i\bar{s}+[\bar{s}]}^u))$ for $i\in[z]$.
\end{lemma}

Substituting $m = \bar{r}$ and $m = \bar{r} - 1$ into Lemma \ref{lm12}, we immediately derive the global constraints \eqref{eq3-10} and \eqref{eq3-10-1}, respectively. Then, the following rusult is obtained.
\begin{lemma}\label{lm12-1}
  The code $\bar{\mathcal{C}}(w)$, as defined in \eqref{eq3-17}, constitutes an $(\bar{n},\bar{k},l)$ MDS array code when $w\in[u-v]$. For  $w\in[u-v, u-1]$, the code $\bar{\mathcal{C}}(w)$ in \eqref{eq3-17} is an $(\bar{n},\bar{k}+1,l)$ MDS array code.
\end{lemma}

\subsection{Repair scheme of $(n,k,\bar{s}^{\bar{n}/\bar{s}})$ rack-aware MSR codes}
\label{subsect-scheme-4}

In the code $\mathcal C$ defined in \eqref{eq3-1}, suppose that $a\bar{s}+b$, where $a\in[\tilde{n}]$ and $b\in[\bar{s}]$, is the index of the host rack. If there are $h\in[1,u]$ failed nodes within this rack, the index set for these failed nodes in host rack  $a\bar{s}+b$ is ${\mathcal I}=\{g_0,g_1,\ldots,g_{h-1}\}$, and the index set for the $\bar{d}$ helper racks is ${\mathcal H}$, where ${\mathcal H}\subseteq[\bar{n}]\setminus\{a\bar{s}+b\}$.

We first consider the case where $h\in[1,u-v]$. Note that the index of a rack in the code $\mathcal{C}$ corresponds to the index of a node in the code $\bar{\mathcal{C}}(w)$ defined in \eqref{eq3-17}, then ${\mathcal H}$ naturally serves as the index set for helper nodes in the code $\bar{\mathcal{C}}(w)$. For each $w\in[h]$, node $\bar{\bf c}_{a\bar{s}+b}(w)$ of code $\bar{\mathcal{C}}(w)$ is recoverable from the set $\{\mathbf{R}_{a,b}\bar{\bf c}_{i}(w):i\in{\mathcal H}\}$. Therefore, the $h$ failed nodes ${\bf c}_{(a\bar{s}+b)u+g_0}$, ${\bf c}_{(a\bar{s}+b)u+g_1}$, $\ldots$, ${\bf c}_{(a\bar{s}+b)u+g_{h-1}}$ in host rack $a\bar{s}+b$ can be recovered from $\{\bar{\bf c}_{a\bar{s}+b}(w):w\in[h]\}$ together with local helper nodes. For the case of $h\in[u-v+1,u]$, we can treat similarly. To recover the $h$ failed nodes in host rack $a\bar{s}+b$, we need the following result.
\begin{lemma}\label{lm13}
Following the notations introduced above.  For a given $w\in[u-v]$, the repair bandwidth and the number of accessed symbols for single-node repair in the code
 $\bar{\mathcal{C}}(w)$ are both  $\bar{d}l/\bar{s}$. On the other hand, for a given $w\in[u-v,u-1]$, the repair bandwidth of the code $\bar{\mathcal{C}}(w)$  is $(\bar{d}+1)l/\bar{s}$ for  single-node repair.
\end{lemma}
\begin{proof}
 We only need to show the first conclusion , and the second conclusion  can be proved similarly. For any fixed $w\in[u-v]$, assume that node $\bar{\bf c}_{a\bar{s}+b}(w)$, where $a\in[\tilde{n}]$ and $b\in[\bar{s}]$, is failed, and ${\mathcal H}$ is the index set of the helper nodes of size $\bar{d}$.  Since every row of matrix $\mathbf{R}_{a,b}$ contains only one nonzero element, and all helper nodes utilize the same repair matrix $\mathbf{R}_{a,b}$ to download helper symbols, this repair process exhibits both the low-access property and the constant-repair property.  Recall that the parity check equation of the code ${ \bar{\mathcal C}}(w)$, $w\in[u-v]$  is 
\begin{align*}
	\bar{\mathbf{H}}_0\bar{\bf c}_0(w)+\bar{\mathbf{H}}_1\bar{\bf c}_1(w)+\ldots+\bar{\mathbf{H}}_{\bar{n}-1}\bar{\bf c}_{\bar{n}-1}(w)=0,
\end{align*}
where $\bar{\mathbf{H}}_i (i\in[\bar{n}])$ is defined in \eqref{eq3-16}. In the above equation, each $\bar{\mathbf{H}}_i$ is composed of $l$ block rows, and each block row includes $\bar{r}$ rows. This implies that there are $l$ sets of parity check equations, with each set containing $\bar{r}$ equations. When repairing a failed node $\bar{\bf c}_{a\bar{s}+b}(w)$, only $\bar{l} = l/\bar{s}$ of these $l$ sets are utilized. Specifically, the repair matrix $\mathbf{R}_{a,b}$ is employed to choose the $ \bar{l}$ block rows from each $\bar{\mathbf{H}}_i$ for which the block row indices $i$ satisfy $i_a=b$. More accurately, we have  
\begin{align}\label{eq3-21}
\nonumber	(\mathbf{R}_{a,b}\otimes \mathbf{I}_{\bar{r}})\sum_{i\in[\bar{n}]}\bar{\mathbf{H}}_i\bar{\bf c}_i(w) =&\sum_{i\in[\bar{n}]}(\mathbf{R}_{a,b}\otimes \mathbf{I}_{\bar{r}})\bar{\mathbf{H}}_i(\sum_{z\in[\bar{s}]}\mathbf{R}_{a,z}^{\top}\mathbf{R}_{a,z})\bar{\bf c}_i(w)\\
\nonumber	=&\sum_{i\in[\bar{n}]}\sum_{z\in[\bar{s}]}[(\mathbf{R}_{a,b}\otimes \mathbf{I}_{\bar{r}})\bar{\mathbf{H}}_i\mathbf{R}_{a,z}^{\top}](\mathbf{R}_{a,z}\bar{\bf c}_i(w))\\
\nonumber	=&\sum_{z\in[\bar{s}]}[(\mathbf{R}_{a,b}\otimes \mathbf{I}_{\bar{r}})\bar{\mathbf{H}}_{a\bar{s}+b}\mathbf{R}_{a,z}^{\top}](\mathbf{R}_{a,z}\bar{\bf c}_{a\bar{s}+b}(w))\\
	&+\sum_{i\in[\bar{n}]\setminus\{a\bar{s}+b\}}[(\mathbf{R}_{a,b}\otimes \mathbf{I}_{\bar{r}})\bar{\mathbf{H}}_i\mathbf{R}_{a,b}^{\top}](\mathbf{R}_{a,b}\bar{\bf c}_i(w))=0
\end{align}
from \eqref{eq1}, \eqref{eq2}, \eqref{eq3-16}
and Lemma \ref{lm4}. For each $z\in[\bar{s}]$, let
\begin{align*}
	&\widetilde{\mathbf{H}}_{a\bar{s}+b}^{(z)}=(\mathbf{R}_{a,b}\otimes \mathbf{I}_{\bar{r}})\bar{\mathbf{H}}_{a\bar{s}+b}\mathbf{R}_{a,z}^{\top}, \\ &\tilde{\bf c}_{a\bar{s}+b}^{(z)}(w)=\mathbf{R}_{a,z}\bar{\bf c}_{a\bar{s}+b}(w),
\end{align*}
and for each $i\in[\bar{n}]\setminus\{a\bar{s}+b\}$, let 
\begin{align}\label{eq3-21-1}
\nonumber	&\widetilde{\mathbf{H}}_{i}=(\mathbf{R}_{a,b}\otimes \mathbf{I}_{\bar{r}})\bar{\mathbf{H}}_i\mathbf{R}_{a,b}^{\top}, \\ &\tilde{\bf c}_{i}(w)=\mathbf{R}_{a,b}\bar{\bf c}_i(w).
\end{align}
So, the equation \eqref{eq3-21} can be expressed as 
\begin{align}\label{eq3-22}
	\sum_{z\in[\bar{s}]}\widetilde{\mathbf{H}}_{a\bar{s}+b}^{(z)}\tilde{\bf c}_{a\bar{s}+b}^{(z)}(w)+\sum_{i\in[\bar{n}]\setminus\{a\bar{s}+b\}}\widetilde{\mathbf{H}}_{i}\tilde{\bf c}_{i}(w)=0.
\end{align}

According to Lemma \ref{lm4}, the $\bar{s}+\bar{n}-1$ matrices $\widetilde{\mathbf{H}}_{a\bar{s}+b}^{(z)}$, where $z\in[\bar{s}]$, and $\widetilde{\mathbf{H}}_{i}$, where $i\in[\bar{n}]\setminus\{a\bar{s}+b\}$, are block matrices of size $\bar{l}\times\bar{l}$. Each entry in these matrices is a column vector with a length of $\bar{r}$. Furthermore, $\tilde{\bf c}_{a\bar{s}+b}^{(z)}(w)$ for $z\in[\bar{s}]$ and $\tilde{\bf c}_{i}(w)$ for $i\in[\bar{n}]\setminus\{a\bar{s}+b\}$ are column vectors of length $\bar{l}$. Specifically, we obtain\\
(1) For $z\in[\bar{s}]$,
\begin{align*}
	\widetilde{\mathbf{H}}_{a\bar{s}+b}^{(z)}=\begin{cases}
		{\bf I}_{\bar{l}}\otimes L^{(\bar{r})}(\lambda_{(a\bar{s}+b)\bar{s}+b}^u) \ \text{if} \ z=b,\\
		-{\bf I}_{\bar{l}}\otimes L^{(\bar{r})}(\lambda_{(a\bar{s}+b)\bar{s}+z}^u) \ \text{if} \  z\not=b;\\
	\end{cases}
 \end{align*}
(2) for $f\in[\bar{s}]\setminus\{b\}$
\begin{align*}
	\widetilde{\mathbf{H}}_{a\bar{s}+f}={\bf I}_{\bar{l}}\otimes L^{(\bar{r})}(\lambda_{(a\bar{s}+f)\bar{s}+b}^u);
\end{align*}	
(3) and for $e\in[\tilde{n}]\setminus\{a\}$, $f\in[\bar{s}]$,
\begin{align*}
	\widetilde{\mathbf{H}}_{e\bar{s}+f}=\bar{\varPsi}_{\tilde{n},\bar{e}}(\varphi_f^{(\bar{r})}({\boldsymbol{\lambda}}_{(e\bar{s}+f)\bar{s}+[\bar{s}]}^u)),
\end{align*}	
where $\bar{e}$ is defined as in Lemma \ref{lm4}.

It is observed that \eqref{eq3-22} defines an array code of length $\bar{n}+\bar{s}-1$. The nodes of this new array code can be partitioned into $\tilde{n}$ groups as follows. For $e\in[\tilde{n}]\setminus\{a\}$, group $e$ contains the $\bar{s}$ nodes $\tilde{\bf c}_{e\bar{s}+f}(w)$, where $f\in[\bar{s}]$. Group $a$ includes the $2\bar{s}-1$ nodes $\tilde{\bf c}_{a\bar{s}+b}^{(z)}(w)$ for $z\in[\bar{s}]$ and $\tilde{\bf c}_{a\bar{s}+f}(w)$ for $f\in[\bar{s}]\setminus\{b\}$. Additionally, the parity check submatrices $\widetilde{\mathbf{H}}_{e\bar{s}+f}$, where $e\in[\tilde{n}]\setminus\{a\}$ and $f\in[\bar{s}]$, are precisely the $\bar{n}-\bar{s}$ parity check submatrices that would appear in the MSR code construction with code length $\bar{n}-\bar{s}$ and sub-packetization $\bar{l}$. The parity check submatrices of group $a$, specifically $\widetilde{\mathbf{H}}_{a\bar{s}+b}^{(z)}$ for $z\in[\bar{s}]$ and $\widetilde{\mathbf{H}}_{a\bar{s}+f}$ for $f\in[\bar{s}]\setminus\{b\}$, are block diagonal matrices with identical diagonal entries within each matrix. Since the  $\lambda_i^u$'s that appear in $\widetilde{\mathbf{H}}_{a\bar{s}+b}^{(z)}$ for $z\in[\bar{s}]$ and $\widetilde{\mathbf{H}}_{a\bar{s}+f}$ for $f\in[\bar{s}]\setminus\{b\}$ are distinct from those in $\widetilde{\mathbf{H}}_{e\bar{s}+f}$ for $e\in[\tilde{n}]\setminus\{a\}$ and $f\in[\bar{s}]$, then the approach used to establish the MDS property of the $(\bar{n},\bar{k},\bar{l})$ array code $\bar{\mathcal{C}}(w)$ in Lemma \ref{lm12}, along with Lemma \ref{lm3} (3), can be generalized to prove that \eqref{eq3-22} also defines an $(\bar{n}+\bar{s}-1,\bar{k}+\bar{s}-1,\bar{l})$ MDS array code
\begin{align}\label{eq3-22-1}
	(\tilde{\bf c}_{0}(w),\ldots,\tilde{\bf c}_{a\bar{s}+b-1}(w),\tilde{\bf c}_{a\bar{s}+b}^{(0)}(w),\ldots,\tilde{\bf c}_{a\bar{s}+b}^{(\bar{s}-1)}(w),\ldots,\tilde{\bf c}_{\bar{n}-1}(w)).
\end{align}	
It follows that for a given $w\in[u-v]$, the vectors $\tilde{\bf c}_{a\bar{s}+b}^{(z)}(w)$ for $z\in[\bar{s}]$ can be recovered by downloading $\tilde{\bf c}_{j}(w)$ from each helper nodes $j\in{\mathcal H}$, since  $\bar{d}=\bar{k}+\bar{s}-1$. With the values of $\tilde{\bf c}_{a\bar{s}+b}^{(z)}(w)$ for $z\in[\bar{s}]$ determined, the node $\bar{\bf c}_{a\bar{s}+b}(w)$ can be recovered. 
Consequently, both the number of symbols downloaded and accessed during the repair of node $\bar{\bf c}_{a\bar{s}+b}(w)$ equal $\bar{d}\bar{s}^{\tilde{n}-1}$, achieving the optimal access according to Remark \ref{rm1}.
\end{proof}
\begin{remark}\label{rm3}
If $w\in[u-v,u-1]$, the equation \eqref{eq3-22} becomes
\begin{align}\label{eq3-22-2}
	\sum_{z\in[\bar{s}]}\widetilde{\mathbf{H}}_{a\bar{s}+b}^{(z)}\tilde{\bf c}_{a\bar{s}+b}^{(z)}(w)+\sum_{i\in[\bar{n}]\setminus\{a\bar{s}+b\}}\widetilde{\mathbf{H}}_{i}\tilde{\bf c}_{i}(w)=0,
\end{align}
where $\widetilde{\mathbf{H}}_{a\bar{s}+b}^{(z)}=(\mathbf{R}_{a,b}\otimes \mathbf{I}_{\bar{r}-1})\bar{\mathbf{H}}_{a\bar{s}+b}\mathbf{R}_{a,z}^{\top},  \tilde{\bf c}_{a\bar{s}+b}^{(z)}(w)=\mathbf{R}_{a,z}\bar{\bf c}_{a\bar{s}+b}(w)$ for $z\in[\bar{s}]$, and $\widetilde{\mathbf{H}}_{i}=(\mathbf{R}_{a,b}\otimes \mathbf{I}_{\bar{r}-1})\bar{\mathbf{H}}_i\mathbf{R}_{a,b}^{\top}, \tilde{\bf c}_{i}(w)=\mathbf{R}_{a,b}\bar{\bf c}_i(w)$ for $i\in[\bar{n}]\setminus\{a\bar{s}+b\}$. Similarly, \eqref{eq3-22-2} defines an $(\bar{n}+\bar{s}-1,\bar{k}+\bar{s},\bar{l})$ MDS array code
\begin{align}\label{eq3-22-3}
	(\tilde{\bf c}_{0}(w),\ldots,\tilde{\bf c}_{a\bar{s}+b-1}(w),\tilde{\bf c}_{a\bar{s}+b}^{(0)}(w),\ldots,\tilde{\bf c}_{a\bar{s}+b}^{(\bar{s}-1)}(w),\ldots,\tilde{\bf c}_{\bar{n}-1}(w)),
\end{align}	
where $w\in[u-v,u-1]$.

It follows that for a given $w\in[u-v,u-1]$, the vectors $\tilde{\bf c}_{a\bar{s}+b}^{(z)}(w)$ for $z\in[\bar{s}]$ can be recovered by downloading $\tilde{\bf c}_{j}(w)$ from each helper nodes $j\in\mathcal{H}'$ with $\mathcal{H}\subseteq\mathcal{H}'$ and $|\mathcal{H}'|=\bar{d}+1$, since  $\bar{d}=\bar{k}+\bar{s}-1$, $|\mathcal{H}|=\bar{d}$. With the values of $\tilde{\bf c}_{a\bar{s}+b}^{(z)}(w)$ for $z\in[\bar{s}]$ determined, the node $\bar{\bf c}_{a\bar{s}+b}(w)$ can be recovered. 
Consequently, both the number of symbols downloaded during the repair of node $\bar{\bf c}_{a\bar{s}+b}(w)$ equal $(\bar{d}+1)\bar{s}^{\tilde{n}-1}$.
\end{remark}

From \eqref{eq3-21-1}, we know that for $w\in[h]$, $h\in[1,u-v]$ and $i\in{\mathcal H}$,
\begin{align}\label{eq3-22-4}
\nonumber\tilde{\bf c}_{i}(w)=&\mathbf{R}_{a,b}\bar{\bf c}_i(w)\\
\nonumber=&\mathbf{R}_{a,b}\varPsi_{\tilde{n},a}({\rm diag}(\boldsymbol{\lambda}_{i\bar{s}+[\bar{s}]}^w)\sum\limits_{g\in[u]}\theta^{gw}{\bf c}_{iu+g}\\
\nonumber=&\sum_{z\in[\bar{s}]}\mathbf{R}_{a,b}\varPsi_{\tilde{n},a}({\rm diag}(\boldsymbol{\lambda}_{i\bar{s}+[\bar{s}]}^w))\mathbf{R}_{a,z}^{\top}\mathbf{R}_{a,z}\sum\limits_{g\in[u]}\theta^{gw}{\bf c}_{iu+g}\\
=&\mathbf{R}_{a,b}\varPsi_{\tilde{n},a}({\rm diag}(\boldsymbol{\lambda}_{i\bar{s}+[\bar{s}]}^w))\mathbf{R}_{a,b}^{\top}\sum\limits_{g\in[u]}\theta^{gw}\mathbf{R}_{a,b}{\bf c}_{iu+g}.
\end{align}
Here, the second equation is from \eqref{eq3-16-1} and the third equation holds due to $\sum_{z\in[\bar{s}]}\mathbf{R}_{a,z}^{\top}\mathbf{R}_{a,z} = {\bf I}_{l}$. From  equation \eqref{eq3-22-4}, by downloading $\{\sum_{g\in[u]}\theta^{gw}\mathbf{R}_{a,b}{\bf c}_{iu+g}: w\in[h]\}$ from each helper rack $i\in {\mathcal H}$, we can obtain the data $\{\tilde{\bf c}_{i}(w): w\in[h],i\in\mathcal{H}\} $. Since \eqref{eq3-22-1} is an $(\bar{n}+\bar{s}-1,\bar{k}+\bar{s}-1,\bar{l})$ MDS array code, $\{\bar{\bf c}_{a\bar{s}+b}(w): w\in[h]\}$ can be recovered from the data $\{\tilde{\bf c}_{i}(w)=\mathbf{R}_{a,b}\bar{\bf c}_i(w): w\in[h], i\in{\mathcal H}\}$.
During this process, the repair bandwidth is $\bar{d}hl/\bar{s}$, while the total number of symbols accessed is $\bar{d}ul/\bar{s}$. 

Below, we recover the $h$ failed nodes within host rack $a\bar{s}+b$ of code $\mathcal C$ using the data $\{\bar{\bf c}_{a\bar{s}+b}(w)$: $w\in[h]\}$, along with local helper nodes. Note that from \eqref{eq3-16-1}, 
\begin{align*}\bar{\bf c}_{a\bar{s}+b}(w)=\varPsi_{\tilde{n},a}({\rm diag}(\boldsymbol{\lambda}_{(a\bar{s}+b)\bar{s}+[\bar{s}]}^w))\sum\limits_{g\in[u]}\theta^{gw}{\bf c}_{(a\bar{s}+b)u+g},
\end{align*}
and matrix $\varPsi_{\tilde{n},a}({\rm diag}(\boldsymbol{\lambda}_{(a\bar{s}+b)\bar{s}+[\bar{s}]}^w))$ is invertible, we have
\begin{align}\label{eq3-23}
	\sum\limits_{g\in[u]}\theta^{gw}{\bf c}_{(a\bar{s}+b)u+g}=\varPsi_{\tilde{n},a}^{-1}({\rm diag}(\boldsymbol{\lambda}_{(a\bar{s}+b)\bar{s}+[\bar{s}]}^w))\bar{\bf c}_{a\bar{s}+b}(w),
\end{align}
where $w\in[h]$. For each $w\in[h]$, setting ${\bm \eta}_w=\varPsi_{\tilde{n},a}^{-1}({\rm diag}(\boldsymbol{\lambda}_{(a\bar{s}+b)\bar{s}+[\bar{s}]}^w))\bar{\bf c}_{a\bar{s}+b}(w)$, then \eqref{eq3-23} can be rewritten as 
\begin{align}\label{eq3-24}
	\sum\limits_{g\in[u]}\theta^{gw}{\bf c}_{(a\bar{s}+b)u+g}={\bm \eta}_w.
\end{align}
Let ${\mathcal J}=[u]\setminus{\mathcal I}$. For each $w\in[h]$, from \eqref{eq3-24} we have
\begin{align}\label{3-25}
	\sum\limits_{g\in{\mathcal I}}\theta^{gw}c_{(a\bar{s}+b)u+g,j}={ \eta}_{w,j}-\sum\limits_{g\in{\mathcal J}}\theta^{gw}c_{(a\bar{s}+b)u+g,j},
\end{align}
where $j\in[l]$ and ${\bm \eta}_{w}=({ \eta}_{w,0},{ \eta}_{w,1},\ldots,{ \eta}_{w,l-1})$.
When $w$ runs over $\{0,1,\ldots,h-1\}$, \eqref{3-25} can be written in matrix form as follows:
\begin{equation} \label{eq3-25} 
	{\left[\begin{array}{llll}
			1 & 1 & \cdots &  1\\
			\theta^{g_0} & \theta^{g_1} & \cdots &  \theta^{g_{h-1}}\\
			\ \ \vdots &  \ \ \vdots & \ddots &  \ \  \vdots\\
			\theta^{g_0(h-1)} & \theta^{g_1(h-1)} & \cdots &  \theta^{g_{h-1}(h-1)}\\
		\end{array}
		\right]} {\left[\begin{array}{l}  c_{(a\bar{s}+b)u+g_0,j} \\
			c_{(a\bar{s}+b)u+g_1,j} \\ \ \ \ \ \ \vdots \\
			c_{(a\bar{s}+b)u+g_{b-1},j}
		\end{array}
		\right]}={\left[\begin{array}{llll}
			{ \eta}_{0,j}-\sum\limits_{g\in{\mathcal J}}c_{(a\bar{s}+b)u+g,j} \\ { \eta}_{2,j}-\sum\limits_{g\in{\mathcal J}}\theta^{g}c_{(a\bar{s}+b)u+g,j}
			\\ \ \ \indent \indent  \ \vdots\\ { \eta}_{h-1,j}-\sum\limits_{g\in{\mathcal J}}\theta^{(h-1)g}c_{(a\bar{s}+b)u+g,j} \\
		\end{array}
		\right]},
\end{equation}
where $j\in[l]$. Since $\theta^{g_i}\neq
\theta^{g_{i'}}$ for $g_i,g_{i'}\in[u]$ with $g_i\neq g_{i'}$, the
Vandermonde matrix on the left side of the linear system (\ref{eq3-25})
is invertible, then we can recover $c_{(a\bar{s}+b)u+g_0,j},
c_{(a\bar{s}+b)u+g_1,j}, \ldots, c_{(a\bar{s}+b)u+g_{h-1},j},
j\in[l]$, i.e.,
the $h$ failed nodes within  rack $a\bar{s}+b$. Since the intra-rack repair bandwidth is not counted, the repair bandwidth required for recovering the $h$ failed nodes is $\bar{d}h{l}/{\bar{s}}$, which meets the lower bound \eqref{eq1-1} on repair bandwidth. Meanwhile, the amount of data accessed from the $\bar{d}$ helper racks is $\bar{d}ul/{\bar{s}}$, which achieves the lower bound \eqref{eq1-2}  when  $h=u-v$.

For $h\in[u-v+1,u]$,  the proof follows similarly to the case of $h\in[1,u-v]$: For each $w\in[u-v]$, by downloading $\tilde{c}_{i}(w)$ from each helper rack $i\in{\mathcal H}$ with $|{\mathcal H}|=\bar{d}$, we can recover $(\tilde{\bf c}_{a\bar{s}+b}^{(z)}(w):z\in[\bar{s}])=\bar{\bf c}_{a\bar{s}+b}(w)$; for each $w\in[u-v,h-1]$, according to Remark \ref{rm3}, by downloading $\tilde{\bf c}_{i}(w)$ from each helper rack $i\in{\mathcal H}'={\mathcal H}\cup\{i'\}$ with $i'\in[\bar{n}]\setminus({\mathcal H}\cup{\mathcal F})$, we can recover $(\tilde{\bf c}_{a\bar{s}+b}^{(z)}(w):z\in[\bar{s}])=\bar{\bf c}_{a\bar{s}+b}(w)$. Similar to the proof above, we can recover 
the $h$ failed nodes within rack $a\bar{s}+b$. Based on the preceding analysis, the repair bandwidth during the repair process is $\bar{d}hl/{\bar{s}}+(h-u+v)l/\bar{s}$, while the amount of data accessed from the helper racks is $(\bar{d}+1)ul/{\bar{s}}$.
Note that $v<u$, we have $h-u+v<h$, then
\begin{align*}
	\bar{d}h\frac{l}{\bar{s}}+(h-u+v)\frac{l}{\bar{s}}<(\bar{d}+1)h\frac{l}{\bar{s}}.
\end{align*}
In this case, the ratio of the number of downloaded symbols to the optimal repair bandwidth, as given in \eqref{eq1-1}, is less than $ 1 + 1/\bar{d} $. Therefore, the repair bandwidth of the code approaches optimal level when $ \bar{d} $ is sufficiently large and $h>u-v $. 

\section{The proof of Theorem \ref{th2}}
\label{Sect4}

In this section, we present our second family of rack-aware MSR codes, derived by modifying the first construction. This modification involves a careful selection of the elements in $\boldsymbol{\lambda}_{[\bar{n}\bar{s}]}$.
We assume that $\bar{s}+1$ is a divisor of $\bar{n}$ and let the sub-packetization as $l=\bar{s}^{\tilde{n}}$, where $\tilde{n}=\bar{n}/(\bar{s}+1)$. The nodes $({\bf c}_0,{\bf c}_1,\ldots,{\bf c}_{n-1})$ are partitioned into $\bar{n}$ racks, each of size $u$. Here, we utilize $a\in[\tilde{n}]$ to denote the group index and $b\in[\bar{s}+1]$ to denote the rack index within its group. Specifically, the $b$-th rack in the $a$-th group (i.e., rack $a(\bar{s}+1)+b$) comprises $u$ nodes
$({\bf c}_{(a(\bar{s}+1)+b)u}, {\bf c}_{(a(\bar{s}+1)+b)u+1}, \ldots, {\bf c}_{(a(\bar{s}+1)+b)u+u-1})$.

\subsection{Construction}
\label{subsect-scheme-1}

Let $\bar{s}=\bar{d}-\bar{k}+1$ be a positive integer, $\theta$ be an
element of $\mathbb{F}_q$ with multiplicative order $u|(q-1)$, and ${\mathcal Q}=\{\xi^{\alpha}:\alpha\in[(q-1)/u]\}$ be a proper subset of $\mathbb{F}_q$. Assume that $\{\lambda_{a\bar{s}(\bar{s}+1)+b\bar{s}+j}:a\in[\tilde{n}],b\in[\bar{s}+1],j\in[\bar{s}]\}=\boldsymbol{\lambda}_{[\bar{n}\bar{s}]}$ are $\bar{n}\bar{s}$ distinct elements in $\mathbb{F}_q$ such that $\{\theta^g\lambda_{a\bar{s}(\bar{s}+1)+b\bar{s}+j}:a\in[\tilde{n}],b\in[\bar{s}+1],j\in[\bar{s}],g\in[u]\}=\theta^{[u]}\boldsymbol{\lambda}_{[\bar{n}\bar{s}]}$ are $n\bar{s}$ distinct elements in $\mathbb{F}_q$. Consider an array code $\mathcal{C}$ defined by the following parity check equations over $\mathbb{F}_q$:
\begin{equation}\label{eq4-1}
	\mathcal{C}=\{({\bf c}_0,{\bf c}_1,\ldots,{\bf c}_{n-1}):\sum_{a\in[\tilde{n}]}\sum_{b\in[\bar{s}+1]}\sum_{g\in[u]}\mathbf{H}_{au(\bar{s}+1)+bu+g}{\bf c}_{au(\bar{s}+1)+bu+g}=0\},
\end{equation}
\noindent where 
\begin{align}\label{eq4-2}
	\mathbf{H}_{au(\bar{s}+1)+bu+g}=\varPsi_{\tilde{n},a}(\varphi_b^{(r)}(\theta^g\boldsymbol{\lambda}_{a\bar{s}(\bar{s}+1)+b\bar{s}+[\bar{s}]})), a\in[\tilde{n}], b\in[\bar{s}+1], g\in[u].
\end{align}

For a given $a\in[\tilde{n}]$, a non-empty subset ${\mathcal B}=\{b_0,b_1,\ldots,b_{t-1}\}\subseteq [\bar{s}+1]$  with $b_0<b_1<\ldots<b_{t-1}$, and a set ${\mathcal G}_{\mathcal B}=\{{\mathcal G}_{b_0},{\mathcal G}_{b_1},\ldots,{\mathcal G}_{b_{t-1}}\}$ with each $
\mathcal{G}_{b_j}\not=\emptyset$ and $|{\mathcal G}_{\mathcal B}|=\delta$, we define the $\bar{s}m\times\bar{s}\delta$ matrix $\boldsymbol{\phi}_{a,{\mathcal B}}^{(m)}=\boldsymbol{\varphi}_{\mathcal B}^{(m)}(\theta^{{\mathcal G}_{\mathcal B}}\odot\boldsymbol{\lambda}_{a\bar{s}(\bar{s}+1)+{\mathcal B}\bar{s}+[\bar{s}]})$ and the $ml\times \delta l$ matrix $\mathbf{M}_{a,{\mathcal B}}^{(m)}=\boldsymbol{\varPsi}_{\tilde{n},a}(\boldsymbol{\varphi}_{\mathcal B}^{(m)}(\theta^{{\mathcal G}_{\mathcal B}}\odot\boldsymbol{\lambda}_{a\bar{s}(\bar{s}+1)+{\mathcal B}\bar{s}+[\bar{s}]}))$ in the same manner as \eqref{eq3-3} and \eqref{eq3-4}, respectively. 

To guarantee the MDS property of the code $\mathcal C$, it is necessary to satisfy two conditions. First, the $\bar{n}\bar{s}$ distinct elements $\boldsymbol{\lambda}_{[\bar{n}\bar{s}]}$ in $\mathbb{F}_q$ must be chosen such that $\theta^{[u]}\boldsymbol{\lambda}_{[\bar{n}\bar{s}]}$ yields $n\bar{s}$ distinct elements in $\mathbb{F}_q$. Second, the following local constraints must hold:
\begin{align}\label{eq4-3}
	\textit{det}\left(\mathbf{M}_{a,{\mathcal B}}^{(\delta)}\right) \neq 0,  a\in[\tilde{n}],  \emptyset \neq {\mathcal B} \subseteq [\bar{s}+1],  \emptyset\not={\mathcal G}_{b_j} \subseteq [u],  b_j\in\mathcal{B}.
\end{align}
The existence of such $\lambda_i$s in the field of size $O_{\bar{s}}(n)$ is guaranteed by the following two lemmas.
\begin{lemma}\label{lm16}
	If $q\geq n\bar{s}+u\Omega(\bar{s}+1,u)$, where function $\Omega$ is defined in \eqref{eq3-6-1-1},
	then we can find $\bar{n}\bar{s}$ distinct
	elements $\boldsymbol{\lambda}_{[\bar{n}\bar{s}]}$ from $\mathbb{F}_q$ such that $\theta^{[u]}\boldsymbol{\lambda}_{[\bar{n}\bar{s}]}$ are $n\bar{s}$ distinct elements, and satisfy  \eqref{eq4-3}. And
	the time complexity to choose these elements is $O_{\bar{s}}(n)$. To be specific, for each $a\in[\tilde{n}]$, we can find the $\bar{s}^2+\bar{s}$ elements $\boldsymbol{\lambda}_{a\bar{s}(\bar{s}+1)+[\bar{s}(\bar{s}+1)]}$ in any subset ${\mathcal S}\subseteq {\mathcal Q}$ of size $|{\mathcal S}|\geq \bar{s}(\bar{s}+1)+\Omega(\bar{s}+1,u)$ such that $\mathbf{M}_{a,{\mathcal B}}^{(\delta)}$ are invertible for all ${\mathcal B}\subseteq[\bar{s}+1]$.
\end{lemma}
\begin{lemma}\label{lm17}
	Let $\mathbb{F}_q=\mathbb{F}_p[\xi]$, where $p$ is a prime
	and $\xi$ is a primitive element of $\mathbb{F}_q$. If $q\geq max\{n\bar{s}, p^{\bar{s}^5u^2}\}$, 
	then the $\bar{n}\bar{s}$ elements $\lambda_i=\xi^i$, $i\in[\bar{n}\bar{s}]$ satisfy the conditions that $\theta^{[u]}\boldsymbol{\lambda}_{[\bar{n}\bar{s}]}$ are $n\bar{s}$ distinct elements, and that the local constraints \eqref{eq4-3} hold.	
\end{lemma}

We omit the proofs of Lemma \ref{lm16} and Lemma \ref{lm17} since they follow a similar approach to those of Lemma \ref{lm6} and Lemma \ref{lm7}, respectively.

Similar to the derivation method of \eqref{eq3-17}, by \eqref{eq4-1} and \eqref{eq4-2}, it is not difficult to obtain the intermediate code $\bar{\mathcal{C}}(w)$ defined by the following parity check form: 
\begin{align}\label{eq4-5}
	\bar{\mathcal{C}}(w)=\{(\bar{\bf c}_0(w),\bar{\bf c}_1(w),\ldots,\bar{\bf c}_{\bar{n}-1}(w)):\sum_{a\in[\tilde{n}]}\sum_{b\in[\bar{s}+1]}\bar{\mathbf{H}}_{a(\bar{s}+1)+b}\bar{\bf c}_{a(\bar{s}+1)+b}(w)=0\}, 
\end{align}
\noindent where 
\begin{align}\label{eq4-6}
		\bar{\mathbf{H}}_{a(\bar{s}+1)+b}=\begin{cases}
		\varPsi_{\tilde{n},a}(\varphi_b^{(\bar{r})}(\boldsymbol{\lambda}_{a\bar{s}(\bar{s}+1)+b\bar{s}+[\bar{s}]}^u)) \ \text{if} \ w\in[u-v];\\
		\varPsi_{\tilde{n},a}(\varphi_b^{(\bar{r}-1)}(\boldsymbol{\lambda}_{a\bar{s}(\bar{s}+1)+b\bar{s}+[\bar{s}]}^u)) \ \text{if} \ w\in[u-v,u-1],
	\end{cases}
\end{align}
and 
\begin{align}\label{eq4-7}
	\bar{\bf c}_{a(\bar{s}+1)+b}(w)=\varPsi_{\tilde{n},a}({\rm diag}(\boldsymbol{\lambda}_{a\bar{s}(\bar{s}+1)+b\bar{s}+[\bar{s}]}^w))\sum\limits_{g\in[u]}\theta^{gw}{\bf c}_{au(\bar{s}+1)+bu+g}.
\end{align}

For a given $a\in[\tilde{n}]$, a non-empty subset ${\mathcal B}=\{b_0,b_1,\ldots,b_{t-1}\}\subseteq [\bar{s}+1]$, we define the $\bar{s}m\times\bar{s}t$ matrix $\bar{\boldsymbol{\phi}}_{a,{\mathcal B}}^{(m)}=\boldsymbol{\varphi}_{\mathcal B}^{(m)}(\boldsymbol{\lambda}_{a\bar{s}(\bar{s}+1)+{\mathcal B}\bar{s}+[\bar{s}]}^u)$ and the $ml\times tl$ matrix $\bar{\mathbf{M}}_{a,{\mathcal B}}^{(m)}=\boldsymbol{\varPsi}_{\tilde{n},a}(\boldsymbol{\varphi}_{\mathcal B}^{(m)}(\boldsymbol{\lambda}_{a\bar{s}(\bar{s}+1)+{\mathcal B}\bar{s}+[\bar{s}]}^u))$ in the same manner as \eqref{eq3-17-1} and \eqref{eq3-17-2}, respectively. 

To guarantee the MDS property of the code $\bar{\mathcal{C}}(w)$, we require not only that the $\bar{n}\bar{s}$ distinct elements $\boldsymbol{\lambda}_{[\bar{n}\bar{s}]}$ in $\mathbb{F}_q$ satisfy the condition that $\boldsymbol{\lambda}_{[\bar{n}\bar{s}]}^u$ are also $\bar{n}\bar{s}$ distinct elements in $\mathbb{F}_q$, but also that 
\begin{align}\label{eq4-8}
	\text{det}(\bar{\mathbf{M}}_{a,{\mathcal B}}^{(\delta)})\not=0, a\in[\tilde{n}], \emptyset\not={\mathcal B}\subseteq [\bar{s}].
\end{align}
The existence of such $\lambda_i$s in the field of size $O_{\bar{s}}(n)$ is guaranteed by the following two lemmas.
\begin{lemma}\label{lm19}
	Assuming $q\geq n\bar{s}+u\bar{s}2^{\bar{s}-1} $, we can select $\bar{n}\bar{s}$ distinct elements $ \boldsymbol{\lambda}_{[\bar{n}\bar{s}]}$ from $\mathbb{F}_q$ such that $\boldsymbol{\lambda}_{[\bar{n}\bar{s}]}^u $ are also $\bar{n}\bar{s}$ distinct elements, and conditions \eqref{eq4-8} hold. And
	the time complexity to choose these elements is $O_{\bar{s}}(n)$. To be specific, for each $a\in[\tilde{n}]$, we can find the $\bar{s}^2+\bar{s}$ elements $\boldsymbol{\lambda}_{a\bar{s}(\bar{s}+1)+[\bar{s}(\bar{s}+1)]}$ in any subset ${\mathcal S}\subseteq \mathcal Q$ of size $|{\mathcal S}|\geq \bar{s}(\bar{s}+1)+\bar{s}2^{\bar{s}-1}$ such that  $\bar{\mathbf{M}}_{a,{\mathcal B}}^{(\delta)}$ are invertible for all ${\mathcal B}\subseteq[\bar{s}+1]$.
\end{lemma}
\begin{remark}
	It is not hard to see that if $q\geq n\bar{s}+u(\Omega(\bar{s}+1,u)+\bar{s}2^{\bar{s}-1})$,  we can select $\bar{n}\bar{s}$ distinct elements $ \boldsymbol{\lambda}_{[\bar{n}\bar{s}]}$ from $\mathbb{F}_q$ such that Lemmas \ref{lm16} and \ref{lm19} hold simultaneously.
\end{remark}
\begin{lemma}\label{lm20}
	Let $\mathbb{F}_q = \mathbb{F}_p[\xi]$, where $p $ is a prime and $\xi$ is a primitive element of $\mathbb{F}_q$. If $ q \geq \max\{n\bar{s}, p^{\bar{s}^5}\} $, then the $\bar{n}\bar{s}$ elements $\lambda_i = \xi^i$, $i \in [\bar{n}\bar{s}] $, satisfy the conditions that $\boldsymbol{\lambda}_{[\bar{n}\bar{s}]}^u$ are $\bar{n}\bar{s} $ distinct elements and that the local constraints \eqref{eq4-8} hold.
\end{lemma}
We omit the proofs of Lemmas \ref{lm19} and \ref{lm20}, as they follow similarly to those of Lemmas \ref{lm10} and \ref{lm11}, respectively.
\begin{remark}
	It is not hard to see that if $q\geq max\{n\bar{s}, p^{\bar{s}^5u^2}\}$, both Lemma \ref{lm17} and Lemma \ref{lm20} hold simultaneously.
\end{remark}

\subsection{MDS property}
\label{subsect-scheme-3} 

The MDS property of the code $\mathcal C$, as defined by \eqref{eq4-1}, \eqref{eq4-2}, and \eqref{eq4-3}, is a direct consequence of the following global constraints.
\begin{lemma}\label{lm18}
	Let  $a_0,a_1,\ldots,a_{z-1}$ be any  $z$ distinct integers in $[\tilde{n}]$, and let  ${\mathcal B}_0,{\mathcal B}_1,\ldots,{\mathcal B}_{z-1}$ be any $z$ non-empty subsets of 
	$[\bar{s}+1]$ such that $|{\mathcal G}_{{\mathcal B}_0}|+|{\mathcal G}_{{\mathcal B}_1}|+\ldots+|{\mathcal G}_{{\mathcal B}_{z-1}}|=m\leq r$. Then, we obtain
	\begin{align}\label{eq4-4}
		{\rm det}([\mathbf{M}_{a_0,\mathcal{B}_0}^{(m)} \ \mathbf{M}_{a_1,\mathcal{B}_1}^{(m)}	\ldots \mathbf{M}_{a_{z-1},\mathcal{B}_{z-1}}^{(m)}])\not=0.
	\end{align}
\end{lemma}
The proof of Lemma \ref{lm18} is omitted as it closely parallels that of Lemma \ref{lm8}.

Using Lemmas \ref{lm3}, \ref{lm19} and \ref{lm20}, we can easily derive the following result in a manner similar to the proof of Lemma \ref{lm12}.
\begin{lemma}\label{lm21}
	Let  $a_0,a_1,\ldots,a_{z-1}$ be any  $z$ distinct integers in $[\tilde{n}]$, and let  ${\mathcal B}_0,{\mathcal B}_1,\ldots,{\mathcal B}_{z-1}$ be any $z$ non-empty subsets of 
	$[\bar{s}+1]$ such that $|{\mathcal B}_0|+|{\mathcal B}_1|+\ldots+|{\mathcal B}_{z-1}|=m\leq \bar{r}$. Then, we obtain
	\begin{align}\label{eq4-9}
		{\rm det}([\mathbf{M}_{a_0,\mathcal{B}_0}^{(m)} \ \mathbf{M}_{a_1,\mathcal{B}_1}^{(m)}	\ldots \mathbf{M}_{a_{z-1},\mathcal{B}_{z-1}}^{(m)}])\not=0,
	\end{align}
	where $\bar{\mathbf{M}}_{a_i,\mathcal{B}_i}^{(m)}=\boldsymbol{\varPsi}_{\tilde{n},a_i}(\boldsymbol{\varphi}_{{\mathcal B}_i}^{(m)}(\boldsymbol{\lambda}_{a_i\bar{s}(\bar{s}+1)+{\mathcal B}_i\bar{s}+[\bar{s}]}^u))$ for $i\in[z]$.
\end{lemma}

Substituting $m=\bar{r}$ and $m=\bar{r}-1$ into Lemma \ref{lm21}, we respectively obtain  ${\rm det}([\mathbf{M}_{a_0,\mathcal{B}_0}^{(\bar{r})} \ \mathbf{M}_{a_1,\mathcal{B}_1}^{(\bar{r})}	\ldots \mathbf{M}_{a_{z-1},\mathcal{B}_{z-1}}^{(\bar{r})}])\not=0$ 
and  ${\rm det}([\mathbf{M}_{a_0,\mathcal{B}_0}^{(\bar{r}-1)} \ \mathbf{M}_{a_1,\mathcal{B}_1}^{(\bar{r}-1)}	\ldots \mathbf{M}_{a_{z-1},\mathcal{B}_{z-1}}^{(\bar{r}-1)}])\not=0$, which correspond to their respective global constraints. Then, the following rusult is obtained.
\begin{lemma}\label{lm22}
	The code $\bar{\mathcal{C}}(w)$,  defined in \eqref{eq4-5}, is an $(\bar{n},\bar{k},l)$ MDS array code for $w\in[u-v]$. For $w\in[u-v, u-1]$, the code $\bar{\mathcal{C}}(w)$ in \eqref{eq4-5} forms an $(\bar{n},\bar{k}+1,l)$ MDS array code.
\end{lemma}

\subsection{Repair scheme of $(n,k,\bar{s}^{\bar{n}/(\bar{s}+1)})$ rack-aware MSR codes}
\label{subsect-scheme-4} 
 
The discussion is separated into two cases.

\begin{itemize}
	\item  Repairing $h$ failed nodes $({\bf c}_{(a(\bar{s}+1)+b)u+g_0},\ldots,{\bf c}_{(a(\bar{s}+1)+b)u+g_{h-1}})$ in host rack $a(\bar{s}+1)+b$, where $a\in[\tilde{n}]$ and $b\in[\bar{s}]$. Let ${\mathcal H}\subseteq[\bar{n}]\setminus\{a(\bar{s}+1)+b\}$ be the index set of the helper racks with size $|{\mathcal H}|=\bar{d}$.
	 When $h\in[1,u-v]$, these $h$ failed nodes can be recovered from $\{\mathbf{R}_{a,b}\bar{\bf c}_{i}(w):i\in{\mathcal H},w\in[h]\}$ and all surviving nodes in  host rack $a(\bar{s}+1)+b$.	
	 when $h\in[u-v+1,u]$, the $h$ failed nodes  can be recovered from $\{\mathbf{R}_{a,b}\bar{\bf c}_{i}(w):i\in{\mathcal H},w\in[u-v]\}$, $\{\mathbf{R}_{a,b}\bar{\bf c}_{i}(w):i\in{\mathcal H}', \mathcal{H}\subseteq {\mathcal H}', |{\mathcal H}'|=\bar{d}+1,w\in[u-v,h-1]\}$ and all surviving nodes in  host rack $a(\bar{s}+1)+b$.	
	\item  Repairing $h$ failed nodes $({\bf c}_{(a(\bar{s}+1)+\bar{s})u+g_0},\ldots,{\bf c}_{(a(\bar{s}+1)+\bar{s})u+g_{h-1}})$ in host rack $a(\bar{s}+1)+\bar{s}$, where $a\in[\tilde{n}]$. Let ${\mathcal H}\subseteq[\bar{n}]\setminus\{a(\bar{s}+1)+\bar{s}\}$ be the index set of the helper racks with size $|{\mathcal H}|=\bar{d}$. When $h\in[1,u-v]$, we set
	\begin{align}\label{eq4-10-1}
	{\mathcal H}_1={\mathcal H}\cap(a(\bar{s}+1)+[\bar{s}+1])
	\end{align}
	and set 
	\begin{align}\label{eq4-10-2}
		{\mathcal H}_2={\mathcal H}\setminus(a(\bar{s}+1)+[\bar{s}+1]).
	\end{align}
	 These $h$ failed nodes can be recovered from $\{\mathbf{R}_{a,b}\bar{\bf c}_{a(\bar{s}+1)+b}(w):a(\bar{s}+1)+b\in{\mathcal H}_1,w\in[h]\}$, $\{\mathbf{R}_{a,\bar{s}}\bar{\bf c}_{i}(w):i\in{\mathcal H}_2,w\in[h]\}$ and all surviving nodes in  host rack $a(\bar{s}+1)+\bar{s}$. When $h\in[u-v+1,u]$, $\{\bar{\bf c}_{a(\bar{s}+1)+\bar{s}}(w):w\in[u-v]\}$ can be recovered from $\{\mathbf{R}_{a,b}\bar{\bf c}_{a(\bar{s}+1)+b}(w):a(\bar{s}+1)+b\in{\mathcal H}_1,w\in[u-v]\}$ and $\{\mathbf{R}_{a,\bar{s}}\bar{\bf c}_{i}(w):i\in{\mathcal H}_2,:w\in[u-v]\}$.  Let ${\mathcal H}'\subseteq[\bar{n}]\setminus\{a(\bar{s}+1)+\bar{s}\}$ be the index set of the helper racks with size $|{\mathcal H}'|=\bar{d}+1$ and $\mathcal{H}\subseteq {\mathcal H}'$, 
	\begin{align}\label{eq4-10-3}
		{\mathcal H}_1'={\mathcal H}'\cap(a(\bar{s}+1)+[\bar{s}+1])
	\end{align}
	and
	\begin{align}\label{eq4-10-4}
		{\mathcal H}_2'={\mathcal H}'\setminus(a(\bar{s}+1)+[\bar{s}+1]).
	\end{align}
    $\{\bar{\bf c}_{(a\bar{s}+1)+\bar{s}}(w):w\in[u-v,h-1]\}$ can be recovered from $\{\mathbf{R}_{a,b}\bar{\bf c}_{a(\bar{s}+1)+b}(w):a(\bar{s}+1)+b\in{\mathcal H}_1',:w\in[u-v,h-1]\}$ and $\{\mathbf{R}_{a,\bar{s}}\bar{\bf c}_i(w):i\in{\mathcal H}_2',:w\in[u-v,h-1]\}$.   
    So, these $h$ failed nodes in host rack $a(\bar{s}+1)+\bar{s}$ can be recovered from $\{\bar{\bf c}_{a(\bar{s}+1)+\bar{s}}(w):w\in[h]\}$ and all surviving nodes in  host rack $a(\bar{s}+1)+b$.

\end{itemize}

\begin{lemma}\label{lm23}
	Following the notations introduced above.  For each $w\in[u-v]$, the repair bandwidth of code $\bar{\mathcal{C}}(w)$  is $\bar{d}l/\bar{s}$ for  single-node repair. On the other hand, for each $w\in[u-v,u-1]$, the repair bandwidth of code $\bar{\mathcal{C}}(w)$ is  $(\bar{d}+1)l/\bar{s}$ for  single-node repair.
\end{lemma}
\begin{proof}
	We only need to show the first conclusion , and the second conclusion  can be proved similarly. Note that in code $\bar{\mathcal{C}}(w)$, the index of a node corresponds to the index of a rack in code $\mathcal{C}$.	For any fixed $w\in[u-v]$, the procedure to repair the first $\bar{s}$ nodes of each group is exactly the same as the repair procedure of $(\bar{n},\bar{k},\bar{s}^{\bar{n}/\bar{s}})$ MSR code described in the proof of Lemma \ref{lm13}, therefore we omit the detail.
	Assume that node $\bar{\bf c}_{a(\bar{s}+1)+\bar{s}}(w)$, where $a\in[\tilde{n}]$ and $w\in[u-v]$, is failed, and ${\mathcal H}\subseteq[\bar{n}]\setminus\{a(\bar{s}+1)+\bar{s}\}$ is the index set of the helper nodes of size $\bar{d}$. Recall that the parity check equation of the code $ \bar{\mathcal C}(w)$ is 
	\begin{align}\label{eq4-10}
		\bar{\mathbf{H}}_0\bar{\bf c}_0(w)+\bar{\mathbf{H}}_1\bar{\bf c}_1(w)+\ldots+\bar{\mathbf{H}}_{\bar{n}-1}\bar{\bf c}_{\bar{n}-1}(w)=0,
	\end{align}
	where  $w\in[u-v]$ and $\bar{\mathbf{H}}_i,i\in[\bar{n}]$ is defined in \eqref{eq4-6}. In the above equation, each $\bar{\mathbf{H}}_i$ is composed of $l$ block rows, and each block row includes $\bar{r}$ rows. This implies that there are $l$ sets of parity check equations, with each set containing $\bar{r}$ equations. These $l$ sets of parity check equations are divided into $\bar{l}=l/\bar{s}$ groups, each of $\bar{s}$ sets of parity check equations. The set of parity check equations indices $i$ with the same group only differ in the $a$-th digit of their expansion in base $\bar{s}$. To repair $\bar{\mathbf{c}}_{a(\bar{s}+1)+\bar{s}}(w)$, we sum up each group of $\bar{s}$ block rows of all matrices in \eqref{eq4-10}. More  precisely, we have  
	\begin{align}\label{eq4-11}
		\nonumber	(\mathbf{R}_{a,\bar{s}}\otimes \mathbf{I}_{\bar{r}})\sum_{i\in[\bar{n}]}\bar{\mathbf{H}}_i\bar{\bf c}_i(w) =&\sum_{i\in[\bar{n}]}(\mathbf{R}_{a,\bar{s}}\otimes \mathbf{I}_{\bar{r}})\bar{\mathbf{H}}_i(\sum_{z\in[\bar{s}]}\mathbf{R}_{a,z}^{\top}\mathbf{R}_{a,z})\bar{\bf c}_i(w)\\
		\nonumber	=&\sum_{i\in[\bar{n}]}\sum_{z\in[\bar{s}]}[(\mathbf{R}_{a,\bar{s}}\otimes \mathbf{I}_{\bar{r}})\bar{\mathbf{H}}_i\mathbf{R}_{a,z}^{\top}](\mathbf{R}_{a,z}\bar{\bf c}_i(w))\\
		\nonumber	=&\sum_{z\in[\bar{s}]}[(\mathbf{R}_{a,\bar{s}}\otimes \mathbf{I}_{\bar{r}})\bar{\mathbf{H}}_{a(\bar{s}+1)+\bar{s}}\mathbf{R}_{a,z}^{\top}](\mathbf{R}_{a,z}\bar{\bf c}_{a(\bar{s}+1)+\bar{s}}(w))\\
	\nonumber	&+\sum_{b\in[\bar{s}]}[(\mathbf{R}_{a,\bar{s}}\otimes \mathbf{I}_{\bar{r}})\bar{\mathbf{H}}_{a(\bar{s}+1)+b}\mathbf{R}_{a,b}^{\top}](\mathbf{R}_{a,b}\bar{\bf c}_{a(\bar{s}+1)+b}(w))\\
		&+\sum_{i\in[\bar{n}]\setminus(a(\bar{s}+1)+[\bar{s}+1])}[(\mathbf{R}_{a,\bar{s}}\otimes \mathbf{I}_{\bar{r}})\bar{\mathbf{H}}_i\mathbf{R}_{a,0}^{\top}](\mathbf{R}_{a,\bar{s}}\bar{\bf c}_i(w))=0
	\end{align}
	from \eqref{eq1}, \eqref{eq2}, \eqref{eq2-1-1}, \eqref{eq4-6}
	and Lemma \ref{lm5}. For each $z\in[\bar{s}]$,  let
	\begin{align*}
		&\widetilde{\mathbf{H}}_{a(\bar{s}+1)+\bar{s}}^{(z)}=(\mathbf{R}_{a,\bar{s}}\otimes \mathbf{I}_{\bar{r}})\bar{\mathbf{H}}_{a(\bar{s}+1)+\bar{s}}\mathbf{R}_{a,z}^{\top}, \\ &\tilde{\bf c}_{a(\bar{s}+1)+\bar{s}}^{(z)}(w)=\mathbf{R}_{a,z}\bar{\bf c}_{a(\bar{s}+1)+\bar{s}}(w).
	\end{align*}
	 For each $b\in[\bar{s}]$, let
	\begin{align}\label{eq4-11-1}
		\nonumber	&\widetilde{\mathbf{H}}_{a(\bar{s}+1)+b}=(\mathbf{R}_{a,\bar{s}}\otimes \mathbf{I}_{\bar{r}})\bar{\mathbf{H}}_{a(\bar{s}+1)+b}\mathbf{R}_{a,b}^{\top}, \\ &\tilde{\bf c}_{a(\bar{s}+1)+b}(w)=\mathbf{R}_{a,b}\bar{\bf c}_{a(\bar{s}+1)+b}(w).
	\end{align}
	For each $i\in[\bar{n}]\setminus(a(\bar{s}+1)+[\bar{s}+1])$, let
	\begin{align}\label{eq4-11-2}
    \nonumber	\widetilde{\mathbf{H}}_{i}=(\mathbf{R}_{a,\bar{s}}\otimes \mathbf{I}_{\bar{r}})\bar{\mathbf{H}}_i\mathbf{R}_{a,0}^{\top}, \\
	\tilde{\bf c}_i(w)=\mathbf{R}_{a,\bar{s}}\bar{\bf c}_i(w).
	\end{align}
	So, the equation \eqref{eq4-11} can be expressed as 
	\begin{align}\label{eq4-12}
		\sum_{z\in[\bar{s}]}\widetilde{\mathbf{H}}_{a(\bar{s}+1)+\bar{s}}^{(z)}\tilde{\bf c}_{a(\bar{s}+1)+\bar{s}}^{(z)}(w)+\sum_{i\in[\bar{n}]\setminus\{a(\bar{s}+1)+\bar{s}\}}\widetilde{\mathbf{H}}_{i}\tilde{\bf c}_{i}(w)=0.
	\end{align}
	
	According to Lemma \ref{lm5}, the $\bar{s}+\bar{n}-1$ matrices $\widetilde{\mathbf{H}}_{a(\bar{s}+1)+\bar{s}}^{(z)}$ where $z\in[\bar{s}]$, and $\widetilde{\mathbf{H}}_{i}$ where $i\in[\bar{n}]\setminus\{a(\bar{s}+1)+\bar{s}\}$, are block matrices of size $\bar{l}\times\bar{l}$. Each entry in these matrices is a column vector of length  $\bar{r}$. Furthermore, $\tilde{\bf c}_{a(\bar{s}+1)+\bar{s}}^{(z)}(w)$ for $z\in[\bar{s}]$ and $\tilde{\bf c}_{i}(w)$ for $i\in[\bar{n}]\setminus\{a(\bar{s}+1)+\bar{s}\}$ are column vectors of length $\bar{l}$. Specifically, we obtain\\
	 For $z\in[\bar{s}]$,
	\begin{align*}
		\widetilde{\mathbf{H}}_{a(\bar{s}+1)+\bar{s}}^{(z)}=
			{\bf I}_{\bar{l}}\otimes L^{(\bar{r})}(\lambda_{(a\bar{s}(\bar{s}+1)+\bar{s}^2+z}^u);
	\end{align*}
   for $b\in[\bar{s}]$
	\begin{align*}
		\widetilde{\mathbf{H}}_{a(\bar{s}+1)+b}={\bf I}_{\bar{l}}\otimes L^{(\bar{r})}(\lambda_{(a\bar{s}(\bar{s}+1)+b\bar{s}+b}^u);
	\end{align*}	
   and for $e\in[\tilde{n}]\setminus\{a\}$, $f\in[\bar{s}+1]$,
	\begin{align*}
		\widetilde{\mathbf{H}}_{e(\bar{s}+1)+f}=\bar{\varPsi}_{\tilde{n},\bar{e}}(\varphi_f^{(\bar{r})}(\boldsymbol{\lambda}_{e\bar{s}(\bar{s}+1)+f\bar{s}+[\bar{s}]}^u))
	\end{align*}	
	where $\bar{e}$ is defined as in Lemma \ref{lm5}.
	
	We can observe that matrices $\widetilde{\mathbf{H}}_{e(\bar{s}+1)+f}$ where $e\in[\tilde{n}]\setminus\{a\}$ and $f\in[\bar{s}+1]$, are precisely the $\bar{n}-\bar{s}-1$ parity check submatrices that would appear in the MSR code construction with code length $\bar{n}-\bar{s}-1$ and sub-packetization $\bar{l}$. The other matrices $\widetilde{\mathbf{H}}_{a(\bar{s}+1)+\bar{s}}^{(z)}$ for $z\in[\bar{s}]$ and $\widetilde{\mathbf{H}}_{a(\bar{s}+1)+b}$ for $b\in[\bar{s}]$, are block diagonal matrices with identical diagonal entries within each matrix. Since the $\lambda_i^u$ values in $\widetilde{\mathbf{H}}_{a(\bar{s}+1)+\bar{s}}^{(z)}$ for $z\in[\bar{s}]$ and $\widetilde{\mathbf{H}}_{a(\bar{s}+1)+b}$ for $b\in[\bar{s}]$ are distinct from those in $\widetilde{\mathbf{H}}_{e(\bar{s}+1)+f}$ for $e\in[\tilde{n}]\setminus\{a\}$ and $f\in[\bar{s}+1]$, the approach used to establish the MDS property of the $(\bar{n},\bar{k},\bar{l})$ array code in Lemma \ref{lm21}, along with Lemma \ref{lm3} (3), can be generalized to prove that \eqref{eq4-12} also defines an $(\bar{n}+\bar{s}-1,\bar{k}+\bar{s}-1,\bar{l})$ MDS array code
	\begin{align}\label{eq4-13}
		(\tilde{\bf c}_{0}(w),\ldots,\tilde{\bf c}_{a(\bar{s}+1)+\bar{s}-1}(w),\tilde{\bf c}_{a(\bar{s}+1)+\bar{s}}^{(0)}(w),\ldots,\tilde{\bf c}_{a(\bar{s}+1)+\bar{s}}^{(\bar{s}-1)}(w),\ldots,\tilde{\bf c}_{\bar{n}-1}(w)).
	\end{align}	
	It follows that for a given $w\in[u-v]$, the vectors $\tilde{\bf c}_{a(\bar{s}+1)+\bar{s}}^{(z)}(w)$ for $z\in[\bar{s}]$ can be recovered by downloading $\tilde{\bf c}_{j}(w)$ from each helper node $j\in{\mathcal H}$, since  $\bar{d}=\bar{k}+\bar{s}-1$. With the values of $\tilde{\bf c}_{a(\bar{s}+1)+\bar{s}}^{(z)}(w)$ for $z\in[\bar{s}]$ determined, the node $\bar{\bf c}_{a(\bar{s}+1)+\bar{s}}(w)$ can be recovered. 
	Consequently, the total number of symbols downloaded from all helper nodes for repairing node $\bar{\bf c}_{a(\bar{s}+1)+\bar{s}}(w)$ is $\bar{d}\bar{s}^{\tilde{n}-1}$, achieving the lower bound on repair bandwidth given in \eqref{eq1-1}.
\end{proof}
\begin{remark}\label{rm7}
	If $w\in[u-v,u-1]$, the equation \eqref{eq4-12} becomes 
	\begin{align}\label{eq4-14}
	\sum_{z\in[\bar{s}]}\widetilde{\mathbf{H}}_{a(\bar{s}+1)+\bar{s}}^{(z)}\tilde{\bf c}_{a(\bar{s}+1)+\bar{s}}^{(z)}(w)+\sum_{i\in[\bar{n}]\setminus\{a(\bar{s}+1)+\bar{s}\}}\widetilde{\mathbf{H}}_{i}\tilde{\bf c}_{i}(w)=0,
	\end{align}
	where $\widetilde{\mathbf{H}}_{a(\bar{s}+1)+\bar{s}}^{(z)}=(\mathbf{R}_{a,\bar{s}}\otimes \mathbf{I}_{\bar{r}-1})\bar{\mathbf{H}}_{a(\bar{s}+1)+\bar{s}}\mathbf{R}_{a,z}^{\top}, \tilde{\bf c}_{a(\bar{s}+1)+\bar{s}}^{(z)}(w)=\mathbf{R}_{a,z}\bar{\bf c}_{a(\bar{s}+1)+\bar{s}}(w)$ for $z\in[\bar{s}]$, $\widetilde{\mathbf{H}}_{a(\bar{s}+1)+b}=(\mathbf{R}_{a,\bar{s}}\otimes \mathbf{I}_{\bar{r}-1})\bar{\mathbf{H}}_{a(\bar{s}+1)+b}\mathbf{R}_{a,b}^{\top}, \tilde{\bf c}_{a(\bar{s}+1)+b}(w)=\mathbf{R}_{a,b}\bar{\bf c}_{a(\bar{s}+1)+b}(w)$ for $b\in[\bar{s}]$, and  $\widetilde{\mathbf{H}}_{i}=(\mathbf{R}_{a,\bar{s}}\otimes \mathbf{I}_{\bar{r}-1})\bar{\mathbf{H}}_i\mathbf{R}_{a,0}^{\top},
\tilde{\bf c}_i(w)=\mathbf{R}_{a,\bar{s}}\bar{\bf c}_i(w)$ for $i\in[\bar{n}]\setminus\{a(\bar{s}+1)+\bar{s}\}$. Similarly, \eqref{eq4-14} defines an $(\bar{n}+\bar{s}-1,\bar{k}+\bar{s},\bar{l})$ MDS array code
	\begin{align}\label{eq4-15}
	(\tilde{\bf c}_{0}(w),\ldots,\tilde{\bf c}_{a(\bar{s}+1)+\bar{s}-1}(w),\tilde{\bf c}_{a(\bar{s}+1)+\bar{s}}^{(0)}(w),\ldots,\tilde{\bf c}_{a(\bar{s}+1)+\bar{s}}^{(\bar{s}-1)}(w),\ldots,\tilde{\bf c}_{\bar{n}-1}(w)),
	\end{align}	
	where $w\in[u-v,u-1]$.
	
	It follows that for a given $w\in[u-v,u-1]$, the vectors $\tilde{\bf c}_{a(\bar{s}+1)+\bar{s}}^{(z)}(w)$ for $z\in[\bar{s}]$ can be recovered by downloading $\tilde{\bf c}_{j}(w)$ from each helper nodes $j\in\mathcal{H}'$ with $\mathcal{H}\subseteq\mathcal{H}'$ and $|\mathcal{H}'|=\bar{d}+1$, since  $\bar{d}=\bar{k}+\bar{s}-1$. With the values of $\tilde{\bf c}_{a(\bar{s}+1)+\bar{s}}^{(z)}(w)$ for $z\in[\bar{s}]$ determined, the node $\bar{\bf c}_{a(\bar{s}+1)+\bar{s}}(w)$ can be recovered. 
	Consequently, both the number of symbols downloaded during the repair of node $\bar{\bf c}_{a(\bar{s}+1)+\bar{s}}(w)$ equal $(\bar{d}+1)\bar{s}^{\tilde{n}-1}$.
\end{remark}

If the $h$ failed nodes locate in one of the first $\bar{s}$ racks of each group, the proof is exactly the same as that of Theorem \ref{th1}. Now, suppose that the index of the host rack is  $a(\bar{s}+1)+\bar{s}$, where $a\in[\tilde{n}]$. For the case $h\in[1,u-v]$,  let ${\mathcal H}$ be the index set of the helper racks where ${\mathcal H}\subseteq[\bar{n}]\setminus\{a(\bar{s}+1)+\bar{s}\}$ and $|{\mathcal H}|=\bar{d}$, and ${\mathcal I}=\{g_0,g_1,\ldots,g_{h-1}\}$ be the index set of the  $h$ failed nodes within host rack $a(\bar{s}+1)+\bar{s}$. From \eqref{eq4-7}, \eqref{eq4-11-1} and \eqref{eq4-11-2}, we know that for $w\in[h]$,
	\begin{align*}
		\tilde{\bf c}_{a(\bar{s}+1)+b}(w)=&\mathbf{R}_{a,b}\bar{\bf c}_{a(\bar{s}+1)+b}(w)\\
		=&\mathbf{R}_{a,b}\varPsi_{\tilde{n},a}({\rm diag}(\boldsymbol{\lambda}_{  a\bar{s}(\bar{s}+1)+b\bar{s}+[\bar{s}]}^w))\sum\limits_{g\in[u]}\theta^{gw}{\bf c}_{au(\bar{s}+1)+bu+g}, \ a(\bar{s}+1)+b\in {\mathcal H}_1,
	\end{align*}
	and 
	\begin{align*}
		\tilde{\bf c}_{e(\bar{s}+1)+f}(w)=&\mathbf{R}_{a,\bar{s}}\bar{\bf c}_{e(\bar{s}+1)+f}(w)\\
		=&\mathbf{R}_{a,\bar{s}}\varPsi_{\tilde{n},e}({\rm diag}(\boldsymbol{\lambda}_{  e\bar{s}(\bar{s}+1)+f\bar{s}+[\bar{s}]}^w))\sum\limits_{g\in[u]}\theta^{gw}{\bf c}_{eu(\bar{s}+1)+fu+g}, \ e(\bar{s}+1)+f\in {\mathcal H}_2,
	\end{align*}
	 where ${\mathcal H}_1$, ${\mathcal H}_2$ are defined in \eqref{eq4-10-1}, \eqref{eq4-10-2}, respectively. From the equations above, downloading $\mathbf{R}_{a,b}\varPsi_{\tilde{n},a}({\rm diag}(\boldsymbol{\lambda}_{  a\bar{s}(\bar{s}+1)+b\bar{s}+[\bar{s}]}^w))$ $\sum_{g\in[u]}\theta^{gw}{\bf c}_{au(\bar{s}+1)+bu+g}$, $w\in[h]$ from each helper rack $a(\bar{s}+1)+b\in {\mathcal H}_1$, and $\mathbf{R}_{a,\bar{s}}\varPsi_{\tilde{n},e}({\rm diag}(\boldsymbol{\lambda}_{  e\bar{s}(\bar{s}+1)+f\bar{s}+[\bar{s}]}^w))$ $\sum_{g\in[u]}$ $\theta^{gw}{\bf c}_{eu(\bar{s}+1)+fu+g}$, $w\in[h]$ from each helper rack $e(\bar{s}+1)+f\in {\mathcal H}_2$, we can obtain the data $\tilde{\bf c}_{i}(w)$ for all $w\in[h], i\in{\mathcal H}$.
	 From the proof of Lemma \ref{lm23},   $\{\bar{\bf c}_{a(\bar{s}+1)+\bar{s}}(w): w\in[h]\}$ can be recovered from $\{\tilde{\bf c}_{i}(w): i\in{\mathcal H}, w\in[h]\}$.  
	It is easy to see that the repair bandwidth is $\bar{d}hl/\bar{s}$ from the $\bar{d}$ helper racks during this repair process.
	
	Below, we recover the $h$ failed nodes located in host rack $a(\bar{s}+1)+\bar{s}$ of code $\mathcal C$ using the data $\{\bar{\bf c}_{a(\bar{s}+1)+\bar{s}}(w):w\in[h]\}$, along with local helper nodes. Note that from \eqref{eq4-7}, 
	\begin{align*}
		\bar{\bf c}_{a(\bar{s}+1)+\bar{s}}(w)=\varPsi_{\tilde{n},a}({\rm diag}(\boldsymbol{\lambda}_{a\bar{s}(\bar{s}+1)+\bar{s}^2+[\bar{s}]}^w))\sum\limits_{g\in[u]}\theta^{gw}{\bf c}_{au(\bar{s}+1)+u\bar{s}+g},
	\end{align*}
	and the matrix $\varPsi_{\tilde{n},a}({\rm diag}(\boldsymbol{\lambda}_{a\bar{s}(\bar{s}+1)+\bar{s}^2+[\bar{s}]}^w))$ is invertible, we have
	\begin{align}\label{eq4-16}
		\sum\limits_{g\in[u]}\theta^{gw}{\bf c}_{au(\bar{s}+1)+u\bar{s}+g}=\varPsi_{\tilde{n},a}^{-1}({\rm diag}(\boldsymbol{\lambda}_{a\bar{s}(\bar{s}+1)+\bar{s}^2+[\bar{s}]}^w))	\bar{\bf c}_{a(\bar{s}+1)+\bar{s}}(w),
	\end{align}
	where $w\in[h]$. 
	
	The remaining proof is very similar to that of Theorem \ref{th1}, and thus we omit it. 
	
	For $h\in[u-v+1,u]$,  the proof follows similarly to the case $h\in[1,u-v]$: For each $w\in[u-v]$, by downloading $\tilde{c}_{i}(w)$ from each helper rack $i\in{\mathcal H}$ with $|{\mathcal H}|=\bar{d}$, we can recover $(\tilde{\bf c}_{a(\bar{s}+1)+\bar{s}}^{(z)}(w):z\in[\bar{s}])=\bar{\bf c}_{a(\bar{s}+1)+\bar{s}}(w)$; for each $w\in[u-v,h-1]$, according to Remark \ref{rm7}, by downloading $\tilde{\bf c}_{i}(w)$ from each helper rack $
	i\in{\mathcal H}'={\mathcal H}\cup\{i'\}$ with $i'\in[\bar{n}]\setminus({\mathcal H}\cup{\mathcal F})$, we can recover $(\tilde{\bf c}_{a(\bar{s}+1)+\bar{s}}^{(z)}(w):z\in[\bar{s}])=\bar{\bf c}_{a(\bar{s}+1)+\bar{s}}(w)$. Similar to the proof above, we can recover 
	the $h$ failed nodes within rack $a(\bar{s}+1)+\bar{s}$. Based on the preceding analysis, the repair bandwidth during the repair process is $\bar{d}hl/{\bar{s}}+(h-u+v)l/\bar{s}$.
	Note that $v<u$, we have $h-u+v<h$, then
	\begin{align*}
		\bar{d}h\frac{l}{\bar{s}}+(h-u+v)\frac{l}{\bar{s}}<(\bar{d}+1)h\frac{l}{\bar{s}}.
	\end{align*}
	In this case, the ratio of the number of downloaded symbols to the optimal repair bandwidth, as given in \eqref{eq1-1}, is less than $ 1 + 1/\bar{d} $. Therefore, the repair bandwidth of the code approaches optimal level when $ \bar{d} $ is sufficiently large and $h>u-v $.

\section{Conclusion}
\label{sect-conclusion}
In this paper, we considered the rake-aware MSR codes for repairing multiple node failures in the same rack. We presented two classes of explicit rake-aware MSR codes with small sub-packetization $l=\bar{s}^{\tilde{n}}$ and linear field size $q=\textit{O}_{\bar{s}}(n)$, which achieve optimal repair bandwidth for $ h\in[1,u-v]$, and asymptotically  optimal repair bandwidth for $ h\in[u-v+1,u]$.  In particular, they achieve optimal access when $h=u-v$ and $\tilde{n}=\bar{n}/\bar{s}$. Compared to existing rack-aware MSR codes, our constructions attain a much smaller sub-packetization. It is worthwhile to develop a new repair scheme for any $\bar{d}$ helper racks with small sub-packetization.

\begin{appendices}
\section{The proof of Lemma \ref{lm6} and Lemma \ref{lm7}} 
\label{app-A}

Let $\mathbf{x}_{[\bar{s}^2]}$ be $\bar{s}^2$ unknown variables in $\mathbb{F}_q$.
For any non-empty subset ${\mathcal B}=\{b_0,b_1,\ldots,b_{t-1}\}\subseteq [\bar{s}]$ with $b_0<b_1<\ldots<b_{t-1}$, let
\begin{align*}
	\mathbf{x}_{{\mathcal B}\bar{s}+[\bar{s}]}=(x_{b\bar{s}+j}:b\in{\mathcal B},j\in[\bar{s}])
\end{align*} be a sub-vector of $\mathbf{x}_{[\bar{s}^2]}$.
Recall the definition of $\boldsymbol{\varphi}_{\mathcal B}^{(m)}(\theta^{{\mathcal G}_{\mathcal B}}\odot \mathbf{x}_{[\bar{s}t]})$ in \eqref{eq4}. Define 
\begin{align} \label{eqA-1}
f_{\mathcal B}(\mathbf{x}_{{\mathcal B}\bar{s}+[\bar{s}]})=\prod_{\delta=t}^{ut}\prod_{\emptyset\not={\mathcal G}_{b_j}\subseteq[u], j\in[t], \atop m_0+\ldots+m_{t-1}=\delta} \text{det}(\boldsymbol{\varphi}_{\mathcal B}^{(\delta)}(\theta^{{\mathcal G}_{\mathcal B}}\odot \mathbf{x}_{{\mathcal B}\bar{s}+[\bar{s}]})),
\end{align}
and
\begin{align}\label{eqA-2}
	f(\mathbf{x}_{[\bar{s}^2]})=\prod_{{\mathcal B}\subseteq [\bar{s}]}f_{\mathcal B}(\mathbf{x}_{{\mathcal B}\bar{s}+[\bar{s}]}). 
\end{align}
Therefore, the local constraints specified in \eqref{eq3-5} are able to be reformulated as
\begin{align}\label{eqA-3}
	f(\boldsymbol{\lambda}_{a\bar{s}^2+[\bar{s}^2]})\not=0, a\in[\tilde{n}].
\end{align}
\begin{lemma}\label{lmA-1}
Following the notations introduced above, the determnant
$\text{det}(\boldsymbol{\varphi}_{\mathcal B}^{(\delta)}(\theta^{{\mathcal G}_{\mathcal B}}\odot \mathbf{x}_{{\mathcal B}\bar{s}+[\bar{s}]}))$ is a nonzero homogeneous
polynomial of degree $\bar{s}\delta(\delta-1)/2$. In addition, the largest
power (degree) of $x_{b_i\bar{s}+j}\in\mathbf{x}_{{\mathcal B}\bar{s}+[\bar{s}]}$ in $f_{\mathcal B}(\mathbf{x}_{{\mathcal B}\bar{s}+[\bar{s}]})$ is 
\begin{align*}
\text{deg}_{x_{b_i\bar{s}+j}}(f_{\mathcal B})\leq\sum_{\delta=t}^{ut}{\delta-1\choose t-1}{ut\choose \delta}(\delta-t+1)(\delta-1).
\end{align*}
\end{lemma}
\begin{proof} Assume that ${\mathcal B}=\{b_0,b_1,\ldots,b_{t-1}\}\subseteq [\bar{s}]$ with $b_0<b_1<\ldots<b_{t-1}$,  and ${\mathcal G}_{\mathcal B}=\{{\mathcal G}_{b_0},{\mathcal G}_{b_1},\ldots,{\mathcal G}_{b_{t-1}}\}$ with ${\mathcal G}_{b_j}=\{g_{b_j,0},g_{b_j,1},\ldots,g_{b_j,m_j-1}\}\subseteq [u], j\in [t]$ such that $m_0+m_1+\ldots+m_{t-1}=\delta$. From \eqref{eqA-1}, $f_{\mathcal B}(\mathbf{x}_{{\mathcal B}\bar{s}+[\bar{s}]})$ is the product of some determinants $\text{det}(\boldsymbol{\varphi}_{\mathcal B}^{(\delta)}(\theta^{{\mathcal G}_{\mathcal B}}\odot \mathbf{x}_{{\mathcal B}\bar{s}+[\bar{s}]}))$.  For a given determnant $\text{det}(\boldsymbol{\varphi}_{\mathcal B}^{(\delta)}(\theta^{{\mathcal G}_{\mathcal B}}\odot \mathbf{x}_{{\mathcal B}\bar{s}+[\bar{s}]}))$, all the elements that are non-zero and located in the same row share the same degree. More precisely, for $e\in[\bar{s}]$  and $j\in[\delta]$, all the nonzero entries in the $(e\delta+j)$-th row of $\boldsymbol{\varphi}_{\mathcal B}^{(\delta)}(\theta^{{\mathcal G}_{\mathcal B}}\odot \mathbf{x}_{{\mathcal B}\bar{s}+[\bar{s}]})$  have the same degree $j$, it follow that $\text{det}(\boldsymbol{\varphi}_{\mathcal B}^{(\delta)}(\theta^{{\mathcal G}_{\mathcal B}}\odot \mathbf{x}_{{\mathcal B}\bar{s}+[\bar{s}]}))$ is a homogeneous polynomial. So, $f_{\mathcal B}(\mathbf{x}_{{\mathcal B}\bar{s}+[\bar{s}]})$ is also a homogeneous polynomial.
	
Note that the following term 
\begin{align}\label{eqA-3-1}
(\prod_{j\in[m_0]}\prod_{e\in[\bar{s}]}(\theta^{g_{b_0,j}}x_{b_0\bar{s}+e})^{j})\cdot(\prod_{j\in[m_1]}\prod_{e\in[\bar{s}]}(\theta^{g_{b_1,j}}x_{b_1\bar{s}+e})^{m_0+j})\cdots (\prod_{j\in[m_{t-1}]}\prod_{e\in[\bar{s}]}(\theta^{g_{b_{t-1},j}}x_{b_{t-1}\bar{s}+e})^{m_0+\ldots+m_{t-2}+j})
\end{align}
appears in the expansion of $\text{det}(\boldsymbol{\varphi}_{\mathcal B}^{(\delta)}(\theta^{{\mathcal G}_{\mathcal B}}\odot \mathbf{x}_{{\mathcal B}\bar{s}+[\bar{s}]}))$, and can only be obtained by picking constant entries from 
$\varphi_{b_0}^{(\delta)}(\theta^{g_{b_0,0}} \mathbf{x}_{b_0\bar{s}+[\bar{s}]})$, degree-1 entries from $\varphi_{b_0}^{(\delta)}(\theta^{g_{b_0,1}} \mathbf{x}_{b_0\bar{s}+[\bar{s}]})$, $\ldots$, degree-$(m_0-1)$ entries from $\varphi_{b_0}^{(\delta)}(\theta^{g_{b_0,m_0-1}} \mathbf{x}_{b_0\bar{s}+[\bar{s}]})$,  degree-$m_0$ entries from $\varphi_{b_1}^{(\delta)}(\theta^{g_{b_1,0}} \mathbf{x}_{b_1\bar{s}+[\bar{s}]})$, and so on. Then,  we can calculate that $\text{det}(\boldsymbol{\varphi}_{\mathcal B}^{(\delta)}(\theta^{{\mathcal G}_{\mathcal B}}\odot \mathbf{x}_{{\mathcal B}\bar{s}+[\bar{s}]}))$ is a nonzero homogeneous polynomial of degree $(0+1+\ldots+m_0-1+\ldots+(m_0+\ldots+m_{t-1}-1))\bar{s}=\bar{s}\delta(\delta-1)/2$. Moreover, for each $x_{b_i\bar{s}+j}\in{\mathcal B}\bar{s}+[\bar{s}]$, the largest power of $x_{b_i\bar{s}+j}$ occurring in $\text{det}(\boldsymbol{\varphi}_{\mathcal B}^{(\delta)}(\theta^{{\mathcal G}_{\mathcal B}}\odot \mathbf{x}_{{\mathcal B}\bar{s}+[\bar{s}]}))$ is 
\begin{align*}
	\begin{cases}
		\displaystyle\sum_{z=1}^{m_i}(\delta-z) \ \ &\textit{if} \ \ j=b_i,\\
		\displaystyle\sum_{z=1}^{\lfloor m_i/2\rfloor}2(\delta-z)+(m_i\bmod 2)(\delta-\lfloor m_i/2\rfloor-1) \ \ &\textit{if} \ \ j\not=b_i.\\
	\end{cases}
\end{align*}
It is easy to see that 
\begin{align*}
\sum_{z=1}^{m_i}(\delta-z)&\leq\sum_{z=1}^{\lfloor m_i/2\rfloor}2(\delta-z)+(m_i\bmod 2)(\delta-\lfloor m_i/2\rfloor-1)\\
&\leq\sum_{z=1}^{\lfloor m_i/2\rfloor}2(\delta-1)+(m_i\bmod 2)(\delta-1)\\
&=2\lfloor m_i/2\rfloor(\delta-1)+(m_i\bmod 2)(\delta-1)\\
&=m_i(\delta-1).
\end{align*}
Note that  $m_0+\ldots+m_{t-1}=\delta$, $i\in[t]$, and $1\leq m_0,\ldots,m_{t-1}\leq u$,  we have 
\begin{align*}
	m_{i}=\delta-\sum_{j\in[t],\atop j\not=i}m_j\leq \delta-t+1,
\end{align*}
then
\begin{align*}
	\sum_{z=1}^{m_i}(\delta-z)&\leq\sum_{z=1}^{\lfloor m_i/2\rfloor}2(\delta-z)+(m_i\bmod 2)(\delta-\lfloor m_i/2\rfloor-1)\\
	&\leq m_i(\delta-1)\\
	&\leq (\delta-t+1)(\delta-1).
\end{align*}
The largest power of $x_{b_i\bar{s}+j}$ occurring in $\text{det}(\boldsymbol{\varphi}_{\mathcal B}^{(\delta)}(\theta^{{\mathcal G}_{\mathcal B}}\odot \mathbf{x}_{{\mathcal B}\bar{s}+[\bar{s}]}))$ is not greater than $(\delta-t+1)(\delta-1)$. Therefore, the largest power of $x_{b_i\bar{s}+j}$ occurring in $f_{\mathcal B}(\mathbf{x}_{{\mathcal B}\bar{s}+[\bar{s}]})$ is 
\begin{align}\label{eqA-3-1}
\text{deg}_{x_{b_i\bar{s}+j}}(f_{\mathcal{B}})\leq	\sum_{\delta=t}^{ut}\sum_{m_j\in[1,u], j\in[t], \atop m_0+\ldots+m_{t-1}=\delta}{u\choose m_0}\dots{u\choose m_{t-1}}(\delta-t+1)(\delta-1).
\end{align}
Since ${u\choose m_0}\dots{u\choose m_{t-1}}\leq {ut\choose \delta}$, from \eqref{eqA-3-1} we can get
\begin{align*}
	\text{deg}_{x_{b_i\bar{s}+j}}(f_{\mathcal{B}})\leq	\sum_{\delta=t}^{ut}\sum_{m_j\in[1,u], j\in[t], \atop m_0+\ldots+m_{t-1}=\delta}{ut\choose \delta}(\delta-t+1)(\delta-1)\\
	\leq\sum_{\delta=t}^{ut}{\delta-1\choose t-1}{ut\choose \delta}(\delta-t+1)(\delta-1).
\end{align*}
\end{proof}
Note that $\mathcal{B}$ is an arbitrary non-empty of $[\bar{s}]$.  From Lemma \ref{lmA-1}, we have that for any $x_i\in\mathbf{x}_{\mathcal{B}\bar{s}+[\bar{s}]}$,
\begin{align*}
	\textit{deg}_{x_i}(f_{\mathcal B})\leq\sum_{\delta=t}^{ut}{\delta-1\choose t-1}{ut\choose \delta}(\delta-t+1)(\delta-1),
\end{align*}
then we are not hard to obtain the following result.
\begin{corollary}\label{lmA-2}
The polynomial $f(\mathbf{x}_{[\bar{s}^2]})$ defined in \eqref{eqA-2} is a nonzero homogeneous polynomial in $\mathbb{F}_q[x_0,x_1,\ldots,x_{\bar{s}^2-1}]$, and the largest power of $x_{i}\in\mathbf{x}_{[\bar{s}^2]}$ is 
\begin{align*}
\text{deg}_{x_i}(f)&\leq\sum_{t=1}^{\bar{s}}{\bar{s}-1\choose t-1}\sum_{\delta=t}^{ut}{\delta-1\choose t-1}{ut\choose \delta}(\delta-t+1)(\delta-1)\\
&=\sum_{t=1}^{\bar{s}}\sum_{\delta=t}^{ut}{\bar{s}-1\choose t-1}{\delta-1\choose t-1}{ut\choose \delta}(\delta-t+1)(\delta-1).
\end{align*}
\end{corollary}
\subsection{The proof of Lemma \ref{lm6}}
\begin{lemma}\label{lmA-3}(\cite{LWHY})
Let $f=
f(x_1,\ldots, x_n)$ be a nonzero polynomial in $\mathbb{F}_q[x_1, . . . , x_n]$. Let
$d_i = \text{deg}_{x_i}(f)$ be the largest power of $x_i$ occurring in any
monomial of $f$, and $d=max(d_1,\ldots,d_n)$. Let ${\mathcal S}$ be a subset
of  $\mathbb{F}_q$ with $|{\mathcal S}|\geq n+d$. Then there exist $n$ distinct elements
$s_1, s_2,\ldots, s_n\in S$ such that $f(s_1,\ldots,s_n)\not=0$.
\end{lemma}

Recall that we need to choose $\bar{n}\bar{s}$ distinct elements $\boldsymbol{\lambda}_{[\bar{n}\bar{s}]}$ from $\mathbb{F}_q$ such that $\theta^{[u]}\boldsymbol{\lambda}_{[\bar{n}\bar{s}]}$ are $n\bar{s}$ distinct elements, and satisfy \eqref{eqA-3}, i,e, $f(\boldsymbol{\lambda}_{a\bar{s}^2+[\bar{s}^2]})\not=0$, $a\in [\tilde{n}]$. According to
Corollary \ref{lmA-2}, $f(\mathbf{x}_{[\bar{s}^
2]})$ is a nonzero polynomial in $\mathbb{F}_q[\mathbf{x}_{[\bar{s}^2]}]$, and
its degree of each variable $x_i$, $i\in[\bar{s}^2]$ does not exceed $\Omega(\bar{s},u)$, it follow that $\Omega(\bar{s},u)\geq max(d_0,\ldots,d_{\bar{s}^2-1})$ where $d_i=\text{deg}_{x_i}(f)$. Then by
Lemma \ref{lmA-3}, for each $a\in[\tilde{n}]$, we can choose $\boldsymbol{\lambda}_{a\bar{s}^2+[\bar{s}^2]}$ from a
subset ${\mathcal S}\subseteq Q$ with size $|{\mathcal S}|\geq\bar{s}^2+\Omega(\bar{s},u)\geq\bar{s}^2+max(d_1,\ldots,d_{\bar{s}^2-1})$ 
satisfying that
$f(\boldsymbol{\lambda}_{a\bar{s}^2+[\bar{s}^2]})\not=0$. As $q/u\geq \bar{n}\bar{s}+\Omega(\bar{s},u)$, we can repeat
this process $\bar{n}/\bar{s}$ times, and the total complexity is $O_{\bar{s}}(n)$.
Since $\bar{n}\bar{s}$ distinct elements $\boldsymbol{\lambda}_{[\bar{n}\bar{s}]}\subset Q$, then it is easy to see that $\theta^{[u]}\boldsymbol{\lambda}_{[\bar{n}\bar{s}]}$ are $n\bar{s}$ distinct elements in $\mathbb{F}_q$. This
concludes the proof of Lemma \ref{lm6}. 
\subsection{The proof of Lemma \ref{lm7}}
Since $f(\mathbf{x}_{[\bar{s}^2]})$ is a nonzero homogeneous polynomial, it follows that
\begin{align*}
	f(x^{a\bar{s}^2},x^{a\bar{s}^2+1},\ldots,x^{a\bar{s}^2+\bar{s}^2-1})=x^{a\bar{s}^2deg(f)}f(1,x,\ldots,x^{\bar{s}-1}).
\end{align*}
Consider the polynomial
\begin{align*}
	h(x)=f(1,x,\ldots,x^{\bar{s}-1}),
\end{align*}
that is,  we set the value of $x_i$ to be $x^i$  in the polynomial $f(\mathbf{x}_{[\bar{s}^2]})$. If $h(\xi)\not=0$, then the $\bar{n}\bar{s}$ values $\lambda_i=\xi^i$, where $i\in[\bar{n}\bar{s}]$, are all distinct and satisfy the the local constraints \eqref{eq3-5} and \eqref{eq3-6}.	

Given any non-empty subset ${\mathcal B}\subseteq[\bar{s}]$, we set 
\begin{align*}
	h_{{\mathcal B}}(x)=f_{{\mathcal B}}(x^i:i\in{\mathcal B}\bar{s}+[\bar{s}])\in\mathbb{F}_p[x].
\end{align*}
Let ${\mathcal B}=\{b_0,b_1,\ldots,b_{t-1}\}\subseteq [\bar{s}]$ with $b_0<b_1<\ldots<b_{t-1}$. By applying the same approach as \eqref{eqA-3-1}, we can check that the term 
\begin{align*}
\prod_{z\in[t]}\prod_{j\in[m_z]}\prod_{e\in[\bar{s}]}(\theta^{g_{b_z,j}}x^{b_z\bar{s}+e})^{m_0+\ldots+m_{z-1}+j}
\end{align*}
uniquely has the highest degree in the polynomial $h_{{\mathcal B}}(x)$. In the same way, the term 
\begin{align*}
	\prod_{z\in[t]}\prod_{j\in[m_z]}\prod_{e\in[\bar{s}]}(\theta^{g_{b_z,j}}x^{b_z\bar{s}+e})^{\delta-1-(m_0+\ldots+m_{z-1}+j)}
\end{align*}
uniquely has the lowest degree in the polynomial $h_{{\mathcal B}}(x)$. Set
\begin{align*}
	M_{\mathcal B}=\sum_{z\in[t]}\sum_{j\in[m_z]}\sum_{e\in[\bar{s}]}({b_z\bar{s}+e})(m_0+\ldots+m_{z-1}+j),
\end{align*}
\begin{align*}
	m_{\mathcal B}=\sum_{z\in[t]}\sum_{j\in[m_z]}\sum_{e\in[\bar{s}]}({b_z\bar{s}+e})(\delta-1-m_0-\ldots-m_{z-1}-j),
\end{align*}
we have $h_{\mathcal B}(x)=x^{m_{\mathcal B}}\bar{h}_{\mathcal B}(x)$, where $\bar{h}_{\mathcal B}(x)$ is a polynomial with degree $deg(\bar{h}_{\mathcal B})=M_{\mathcal B}-m_{\mathcal B}$.
We can verify that 
\begin{align*}
	deg(\bar{h}_{\mathcal B})&=M_{\mathcal B}-m_{\mathcal B}\\
	&=\sum_{z\in[t]}\sum_{j\in[m_z]}\sum_{e\in[\bar{s}]}({b_z\bar{s}+e})(2(m_0+\ldots+m_{z-1}+j)-\delta+1)\\
	&=\sum_{z\in[t]}\sum_{j\in[m_z]}(2(m_0+\ldots+m_{z-1}+j)-\delta+1)(b_z\bar{s}^2+\frac{\bar{s}(\bar{s}-1)}{2})\\
	&=\bar{s}^2\sum_{z\in[t]}\sum_{j\in[m_z]}(2(m_0+\ldots+m_{z-1}+j)-\delta+1)b_z\\
	&\leq \bar{s}^2\sum_{z\in[t]}\sum_{j\in[m_z]}(m_0+\ldots+m_{z-1}+j)b_z\\
	&\leq \bar{s}^2\sum_{z\in[t]}u^2(z+1)b_z
	\leq u^2\bar{s}^2\sum_{z\in[\bar{s}]}(z+1)z
	\leq u^2\bar{s}^2\sum_{z\in[\bar{s}]}\bar{s}^2=u^2\bar{s}^5.
\end{align*}
Let $q=p^m$. If $m>u^2\bar{s}^5$, then the primitive element $\xi$ satisfies $\bar{h}_{\mathcal B}(\xi)\not=0$.  Consequently, for every subset ${\mathcal B}\subseteq[\bar{s}]$, we have $h_{\mathcal B}(\xi)\not=0$. Furthermore, we impose the condition that the field size  $q>n\bar{s}$ to ensure that the $\bar{n}\bar{s}$ elements $\lambda_i=\xi^i$ for $i\in[\bar{n}\bar{s}]$ are distinct, and that the $n\bar{s}$ elements $\theta^{[u]}\boldsymbol{\lambda}_{[\bar{n}\bar{s}]}$ are also  distinct. This concludes our proof of Lemma \ref{lm7}.

\section{The proof of Lemma \ref{lm10}} \label{app-B}

Let $\mathbf{x}_{[\bar{s}^2]}$ be $\bar{s}^2$ unknown variables in $\mathbb{F}_q$.
For a non-empty subset ${\mathcal B}=\{b_0,b_1,\ldots,b_{t-1}\}\subseteq [\bar{s}]$ with $b_0<b_1<\ldots<b_{t-1}$, let
\begin{align*}
	\mathbf{x}_{{\mathcal B}\bar{s}+[\bar{s}]}=(x_{b\bar{s}+j}:b\in{\mathcal B},j\in[\bar{s}])
\end{align*} be a sub-vector of $\mathbf{x}_{[\bar{s}^2]}$.
Recall the definition of $\boldsymbol{\varphi}_{\mathcal B}^{(m)}(\mathbf{x}_{[\bar{s}t]})$ in \eqref{eq501}. We define 
\begin{align} \label{eqA-5}
	f'_{\mathcal B}(\mathbf{x}_{{\mathcal B}\bar{s}+[\bar{s}]})=\text{det}(\boldsymbol{\varphi}_{\mathcal B}(\mathbf{x}_{{\mathcal B}\bar{s}+[\bar{s}]}))
\end{align}
and
\begin{align}\label{eqA-6}
	f'(\mathbf{x}_{[\bar{s}^2]})=\prod_{{\mathcal B}\subseteq [\bar{s}]}f'_{\mathcal B}(\mathbf{x}_{{\mathcal B}\bar{s}+[\bar{s}]}). 
\end{align}
Therefore, the local constraints specified in \eqref{eq3-18} are able to be reformulated as
\begin{align}\label{eqA-7}
	f'(\boldsymbol{\lambda}_{a\bar{s}^2+[\bar{s}^2]}^u)\not=0, a\in[\tilde{n}].
\end{align}
\begin{lemma}\label{lmA-4}
	For each non-empty ${\mathcal B}\subseteq [\bar{s}]$ of size $|{\mathcal B}|=t$,
	the polynomial $f'_{\mathcal B}(\mathbf{x}_{{\mathcal B}\bar{s}+[\bar{s}]})$ defined in \eqref{eqA-1} is a nonzero homogeneous
	polynomial of degree $\bar{s}t(t-1)/2$. In addition, the largest
	power of $x_i$, $i\in{\mathcal B}\bar{s}+[\bar{s}]$ in $f'_{\mathcal B}(\mathbf{x}_{{\mathcal B}\bar{s}+[\bar{s}]})$ is $t-1$.
\end{lemma}
We omit the proof of Lemma \ref{lmA-4}, as it is almost the same as the proof of Lemma 12 in \cite{LWHY}.

\begin{corollary}\label{lmA-5}
	The polynomial $f'(\mathbf{x}_{[\bar{s}^2]})$ defined in \eqref{eqA-6} is a nonzero homogeneous polynomial in $\mathbb{F}_q[x_0,x_1,\ldots,x_{\bar{s}^2-1}]$, and the largest power of $x_i$, $i\in[\bar{s}^2]$ is 
	\begin{align*}
		\text{deg}_{x_i}=\sum_{t=1}^{\bar{s}}{\bar{s}-1\choose t-1}(t-1)=(\bar{s}-1)2^{\bar{s}-2}.
	\end{align*}
\end{corollary}

Recall that we need to choose $\bar{n}\bar{s}$ distinct elements $\boldsymbol{\lambda}_{[\bar{n}\bar{s}]}$ from $\mathbb{F}_q$ such that $\boldsymbol{\lambda}_{[\bar{n}\bar{s}]}^u$ are also $\bar{n}\bar{s}$ distinct elements, and satisfy \eqref{eqA-7}, i,e, $f'(\boldsymbol{\lambda}_{a\bar{s}^2+[\bar{s}^2]})\not=0$, $a\in [\tilde{n}]$. According to
Corollary \ref{lmA-5}, $f(\mathbf{x}_{[\bar{s}^2]})$ is a nonzero polynomial in $\mathbb{F}_q[\mathbf{x}_{[\bar{s}^2]}]$, and
its degree of each variable $x_i$, $i\in[\bar{s}^2]$ is $(\bar{s}-1)2^{\bar{s}-2}$. Then by
Lemma \ref{lmA-3}, for each $a\in[\tilde{n}]$, we can choose $\boldsymbol{\lambda}_{a\bar{s}^2+[\bar{s}^2]}$ from a
subset ${\mathcal S}\subseteq Q\subset \mathbb{F}_q$ with size $|{\mathcal S}|\geq\bar{s}^2+(\bar{s}-1)2^{\bar{s}-2}$ 
satisfying that
$f(\boldsymbol{\lambda}_{a\bar{s}^2+[\bar{s}^2]})\not=0$. As $q/u\geq \bar{n}\bar{s}+(\bar{s}-1)2^{\bar{s}-2}$, we can repeat
this process $\bar{n}/\bar{s}$ times, and the total complexity is $O_{\bar{s}}(n)$.
Since the $\bar{n}\bar{s}$  elements $\boldsymbol{\lambda}_{[\bar{n}\bar{s}]}\subseteq Q$ are distinct, then it is easy to see that $\boldsymbol{\lambda}_{[\bar{n}\bar{s}]}^u$ are also the $\bar{n}\bar{s}$ distinct elements in $\mathbb{F}_q$. This
concludes the proof of Lemma \ref{lm10}. 

\end{appendices}

\bibliographystyle{IEEEtran}
\bibliography{reference}	
 
\end{document}